\documentclass[11pt,reqno,letterpaper]{article}
\usepackage{wrapfig}
\usepackage{ifthen}
\newboolean{shortver}
\newcommand{\vecone}{\mathbbm{1}}
\setboolean{shortver}{false}% for long version
\usepackage{amsmath,amsthm,amssymb}
\usepackage{extarrows}
\usepackage{units}
\usepackage[boxed, linesnumbered, noend, noline]{algorithm2e}
\usepackage{url}
\usepackage{calc, graphicx} %fuer die Bilder
\usepackage{bbm}
\usepackage{scalerel}
\usepackage{dsfont}
\usepackage{lmodern} 
\usepackage{nicematrix}
\usepackage{arxiv}
\usepackage{color}
\usepackage[utf8]{inputenc} % allow utf-8 input
\usepackage[T1]{fontenc}    % use 8-bit T1 fonts
\usepackage{hyperref}       % hyperlinks
\usepackage{url}            % simple URL typesetting
\usepackage{booktabs}       % professional-quality tables
\usepackage{tabularx}
\usepackage{amsfonts}       % blackboard math symbols
\usepackage{nicefrac}       % compact symbols for 1/2, etc.
  \usepackage{pgfplots}
  \pgfplotsset{compat=newest}
  \usepackage{rotating}
\usepackage{tikz}

\usetikzlibrary{matrix, decorations.pathreplacing}
  \usetikzlibrary{plotmarks}
  \usetikzlibrary{arrows.meta}
  \usepgfplotslibrary{patchplots}
  \usepackage{grffile}
  \usepackage{amsmath}
  \usepackage{nicematrix}
  \usepackage{import}
\usepackage{xifthen}
\usepackage{pdfpages}
\usepackage{transparent}
\usepackage{tikz}
\usepackage{extarrows}
\usepackage{caption}
\DeclareCaptionLabelFormat{cont}{#1~#2\alph{ContinuedFloat}}
\usepackage{microtype}      % microtypography
\usepackage{lipsum}		% Can be removed after putting your text content
\usepackage{graphicx}
\usepackage{float}
\usepackage{enumitem}
\usepackage{subfigure}
\usepackage{doi}
\usepackage{amsfonts}
\usepackage{mathrsfs}
\usepackage{bm}
 \usepackage{nicematrix}
\usepackage{latexsym}
\usepackage{arydshln}
\usepackage[toc,page]{appendix}
\usepackage[capitalize]{cleveref}
\usepackage{thmtools}
\usepackage{thm-restate}

\usepackage{hyperref}

\usepackage{cleveref}
\usepackage{stmaryrd}
\Crefformat{equation}{(#2#1#3)}

\newtheorem{theorem}{Theorem}%[section]

\newtheorem{lemma}{Lemma}[section]
\newtheorem{corollary}[lemma]{Corollary}

\newtheorem{claim}{Claim}[section]

\newtheorem{fact}[lemma]{Fact}
\newtheorem{observation}[lemma]{Observation}
\numberwithin{equation}{section}
\theoremstyle{definition}
\newtheorem{definitiony}[lemma]{Definition}
\newenvironment{definition}
  {\pushQED{\qed}\definitiony}
  {\popQED\enddefinitiony}

\theoremstyle{remark}

\newcommand\slv{\ifthenelse{\boolean{shortver}}{You can find the proof in arXiv. }{We defer this proof to \Cref{app_pre_results} in the appendix. }}
\newcommand\slvv{\ifthenelse{\boolean{shortver}}{You can find the proofs in arXiv. }{We defer these proofs to \Cref{app_pre_results} in the appendix. }}

\newcommand\rou[1]{\lfloor{#1}\rfloor}
\newcommand\ceil[1]{\lceil{#1}\rceil}
\newcommand\bc[1]{\left({#1}\right)}
\newcommand\cbc[1]{\left\{{#1}\right\}}
\newcommand\brk[1]{\left\lbrack{#1}\right\rbrack}

\newcommand\abs[1]{\left|{#1}\right|}

\newcommand\trk[1]{\rank_{\mathbb{F}}}
\newcommand\ber[1]{{\rm Ber}\left({#1}\right)}
\newcommand\bin[1]{{\rm Bin}\left({#1}\right)}

\newcommand*{\dif}{\mathop{}\!\mathrm{d}}

\newcommand\teo[1]{\ensuremath{\mathds{1}}{\left\{{#1}\right\}}}

\DeclareRobustCommand{\VAN}[3]{#2}

\newcommand\vsigma{\bm \sigma}
\newcommand\vx{\bm x}

\newcommand\vG{\bm{G}}

\newcommand\vT{\bm{T}}
\newcommand\vA{\bm{A}}
\newcommand\cC{\mathcal{C}}
\newcommand\cD{\mathcal{D}}

\newcommand\vepp{\varepsilon'}

\newcommand\nn{\nonumber}
\DeclareMathOperator*{\argmax}{arg\,max}

\newcommand\cX{\mathcal{X}}

\newcommand\cI{\mathcal{I}}

\newcommand\eul{\mathrm{e}}
\newcommand\eps{\varepsilon}

\newcommand\NN{\mathbb{N}}

\newcommand\soG{\mathfrak{G}}

\renewcommand{\log}{\ln}

\newcommand{\SmallProb}[1]{{\operatorname{\mathbb{P}}[#1]}}
\newcommand{\SmallProbCond}[2]{{\operatorname{\mathbb{P}}[#1\mid #2]}}
\newcommand{\Exp}[1]{{\operatorname{\mathbb{E}}\mathopen{}\left[#1\right]\mathclose{}}}
\newcommand{\ExpCond}[2]{{\operatorname{\mathbb{E}}\mathopen{}\left[#1\,\middle\vert\,#2\right]\mathclose{}}}
\newcommand{\Prob}[1]{{\operatorname{\mathbb{P}}\mathopen{}\left[#1\right]\mathclose{}}}
\newcommand{\ProbCond}[2]{{\operatorname{\mathbb{P}}\mathopen{}\left[#1\,\middle\vert\,#2\right]\mathclose{}}}
\newcommand{\Var}[1]{{\operatorname{Var}\mathopen{}\left[#1\right]\mathclose{}}}
\newcommand{\VarCond}[2]{{\operatorname{Var}\mathopen{}\left[#1\,\middle\vert\,#2\right]\mathclose{}}}

\newcommand{\ind}{\ensuremath{\mathds{1}}}

\DeclareMathOperator{\rank}{rk}

\newcommand{\pflipoo}{p_{00}}
\newcommand{\pflipoi}{p_{01}}
\newcommand{\pflipii}{p_{11}}
\newcommand{\pflipio}{p_{10}}

\usepackage{todonotes}
\usepackage{color}

\newcommand\lrh[1]{\textcolor{green}{#1}}
\newcommand\zhu[1]{\textcolor{blue}{#1}}

\def\?#1{}
\makeatletter
\def\whp{w.h.p\@ifnextchar-{.}{\@ifnextchar.{.\?}{\@ifnextchar,{.}{\@ifnextchar){.}{\@ifnextchar:{.:\?}{.\ }}}}}}
\makeatother

\makeatletter
\def\Whp{W.h.p\@ifnextchar-{.}{\@ifnextchar.{.\?}{\@ifnextchar,{.}{\@ifnextchar){.}{\@ifnextchar:{.:\?}{.\ }}}}}}
\makeatother

\newcommand\blfootnote[1]{%
  \begingroup
  \renewcommand\thefootnote{}\footnote{#1}%
  \addtocounter{footnote}{-1}%
  \endgroup
}

\DeclareSymbolFont{extraup}{U}{zavm}{m}{n}
\DeclareMathSymbol{\varheart}{\mathalpha}{extraup}{86}
\DeclareMathSymbol{\vardiamond}{\mathalpha}{extraup}{87}
\usepackage{indentfirst}
\usepackage{authblk}
\title{Noisy Linear Group Testing \\ Exact Thresholds and Efficient Algorithms}
\date{} 					% Or removing it
\begin{document}

\author[1]{Lukas Hintze} %, 
\author[2]{Lena Krieg} %, 
\author[3]{Olga Scheftelowitsch} %, 
\author[4]{Haodong Zhu}
% % \thanks{Lena Krieg are supported by DFG CO 646/3 and DFG CO 646/5}
\affil[1]{{\tt lukas.rasmus.hintze@uni-hamburg.de}, Universitty of Hamburg, Department of Informatics, Vogt-K\"olln-Str.\ 30, 22527 Hamburg, Germany.}
\affil[2]{{\tt lena.krieg@tu-dortmund.de}, TU Dortmund, Faculty of Computer Science, 12 Otto-Hahn-St, Dortmund 44227, Germany.}
\affil[3]{{\tt olga.scheftelowitsch@tu-dortmund.de}, TU Dortmund, Faculty of Computer Science, 12 Otto-Hahn-St, Dortmund 44227, Germany.}
\affil[4]{{\tt h.zhu1@tue.nl}, Department of Mathematics and Computer Science, Eindhoven University of Technology, P.O. Box 513, Eindhoven, 5600 MB, The Netherlands}
\blfootnote{Lena~Krieg is supported by DFG CO 646/3 and DFG CO 646/5.
	Olga~Scheftelowitsch is supported by DFG CO 646/5.
        Haodong~Zhu is supported by the NWO Gravitation project NETWORKS under grant no.\ 024.002.003 and the European Union’s Horizon 2020 research and innovation programme under the Marie Skłodowska-Curie grant agreement no.\ 945045. We thank Amin~Coja-Oghlan and Mihyun~Kang for organizing and inviting us to a workshop in Strobl in 2023, where we initiated the collaboration that led to this work. We also thank Jonathan Scarlett for a valuable comment on a previous version of this work.}

\maketitle
\begin{abstract}
In group testing, the task is to identify defective items by testing groups of them together using as few tests as possible.
We consider the setting where each item is defective with a constant probability~$\alpha$, independent of all other items.
In the (over-)idealized noiseless setting, tests are positive exactly if any of the tested items are defective.
We study a more realistic model in which observed test results are subject to noise,
    i.e., tests can display false positive or false negative results with constant positive probabilities.
We determine precise constants~$c$ such that $cn\log n$ tests are required to recover the infection status of every individual for both adaptive and non-adaptive group testing:
    in the former, the selection of groups to test can depend on previously observed test results,
        whereas it cannot in the latter.
Additionally, for both settings, we provide efficient algorithms that identify all defective items with the optimal amount of tests with high probability.
Thus, we completely solve the problem of binary noisy group testing in the studied setting.

\end{abstract}
% \keywords{group testing \and information theory \and combinatorics \and random graphs}

\newcommand{\nnew}{\check{n}}
\newcommand{\mnew}{\check{m}}
\newcommand{\pat}{\sigma}
\newcommand{\pats}{\bm\pat}
\newcommand{\patsmod}{\tilde\pats}
\newcommand{\atest}{\tau}
\newcommand{\atests}{\bm{\atest}}
\newcommand{\otest}{{\tilde{\atest}}}
\newcommand{\otests}{\bm\otest}
\newcommand{\inform}{I}
\newcommand{\tconst}{\kappa}
\newcommand\KL[2]{\operatorname{D_{\mathrm{KL}}}\nolimits\bc{{#1}\|{#2}}}
\newcommand{\uniquedegree}{\Gamma}
\newcommand{\punique}{\gamma}
\newcommand{\infected}{\cI}
\newcommand{\infectedprime}{{\mathcal{J}}}
\newcommand{\numinfected}{\bm{d}}
\newcommand{\healthy}{\overline{\infected}}
\newcommand{\nonadm}{m_{\mathrm{na}} }
\newcommand{\mna}{\nonadm}
\newcommand{\cna}{c_{\text{na}}}
\newcommand{\cad}{c_{\text{ad}}}
\newcommand{\gdep}{\eta}
\newcommand{\adm}{m_{\mathrm{ad}} }
\newcommand{\MAP}{\mathrm{MAP}}
\newcommand{\map}{{\hat{\sigma}_{\MAP}}}
\newcommand{\genie}{{\hat{\sigma}_{\mathrm{gen}}}}
\newcommand{\misInfected}{\cI_{\mathrm{gen}, 0}}
\newcommand{\misD}{{\bm D}_{\mathrm{err}}}
\newcommand{\modGraph}{G_{\eta}}
\newcommand{\good}{g}
\newcommand{\alg}{{\bm{\mathcal{A}}}}
\newcommand{\vp}{\bm{p}}

\newcommand{\SPOG}{\texttt{SPOG}}
\newcommand{\Oh}{\mathcal{O}}
\newcommand{\estaf}{\hat{\sigma}_1}
\newcommand{\estas}{\hat{\sigma}_2}
\newcommand{\estat}{\hat{\sigma}_3}
\newcommand{\estSPEX}{\hat{\pats}_{\mathrm{\SPOG}}}
\newcommand{\SPEX}{\mathrm{SPEX}}
\newcommand{\PRESTO}{\texttt{PRESTO}}

\newcommand{\proportion}[1]{\rho_{\otests}({#1})}
\newcommand{\notationTable}{
Add: bold for random variables \lrh{temporarily i have commented out lines which i've defined in prose.}
\begin{table}[H]
\begin{tabularx}{\linewidth}{l p{6.5cm} X}
Symbol   & Definition or Domain & Meaning \\ \toprule
% $n$ & & number of individuals\\
$m$ & & number of total tests\\
$\vT$ &$\vT=(\vT_i)_{i\in [m]}$ & the vector of tests\\
% $\alpha$&&probability of an individual being infected\\
% $\pats$ & $\pats = (\pats(i))_{i \in [n]} \in \{0,1\}^n$& Ground truth (random variable), true infection states of individuals \\
$\pats_{-i}$&  & ground truth without $i$ \\
% $\atests$ &  $\atests = (\atests(i))_{i \in [m]} \in \{0,1\}^m$ & actual, true test results before noise        \\
% $\otests$ &  $\otests = (\otests(i))_{i \in [m]} \in \{0,1\}^m$ & observed test results, after noise\\
$\otests_{[k]}$ &  $\otests_k = (\otests(i))_{i \in [k]} \in \{0,1\}^k$ & observed first $k$ test results, after noise\\
$\genie^G$ & & Genie estimator on test setup $G$\\
$\map^G$ & & MAP estimator on test setup $G$\\
% $\KL{p}{q}$ & $\KL{p}{q} = \sum_{x \in \cX} p(x) \log\bc{\frac{p(x)}{q(x)}}$& Relative entropy\\
% $p_{01}, p_{10}$ & & Test flip probabilities\\
% $\vp$& $\vp=(p_{00}, p_{01}, p_{10}, p_{11} )$ & \\
$\Gamma$ & $\Gamma = \arg\sup_{w} \Prob{\abs{\partial a \cap \infected}=1 } = \arg\sup_{w} w(1-\alpha)^{w-1} = \lfloor\frac 1 \alpha \rfloor $ & Degree that maximizes Probability of containing exactly one infected individual\\
$\gamma$ &$\gamma =\sup_{w} \Prob{\abs{\partial a \cap \infected}=1 } = \sup_{w} \alpha w(1-\alpha)^{w-1} = \alpha \Gamma (1-\alpha)^{\Gamma -1}$ & Maximal probability of a test containing exactly one infected individual\\
% $\nonadm$&$\nonadm = \frac{\alpha n \log(n)}{\beta \gamma}$&Threshold non-adaptive group testing\\
% $\adm$& $\adm = \frac{\alpha n \log(n)}{\KL{p_{11}}{p_{01}}}$&Threshold adaptive group testing\\
$\infected$ &$\infected = \cbc{i \mid \pats_i = 1}$ & Infected individuals \\
$\modGraph$ & & adjusted Graph (further defined in lower bound for non-adaptive) \\
$\misInfected^G$ & & misclassified infected individuals by genie\\
$\good_i$  & & good test for individual $i$ \\
$\alg$ & $\alg= \bc{\alg_i}_{i}$  & adaptive test scheme\\
$\alg_i(G_{i-1}, \alg_{i-1})$ && $i$th test of algorithm $i$ depending on earlier tests and results\\
%\zhu{$\alg_i(\alg_{[i-1]},\otests_{[i-1]})$}&&
\bottomrule
\end{tabularx}
\end{table}
}
\tikzstyle{fun}=[rectangle, minimum size=13, text height=5, draw=black]
\tikzstyle{var}=[circle,draw=black,minimum size=13, text height=5]
\newcommand{\adpic}{
    \begin{tikzpicture}
        \def\n{6}
        \def\vhdiff{0.7}
        \def\avdiff{0.8}
        \def\avdiffs{1.6}
        \def\s{2}
        \def\sr{3}
        \def\u{4}
        \def\ur{5}
        \def\secondlevel{-\avdiffs-1.3}
        \def\thirdlevel{\secondlevel-\avdiffs-1.2}
        \foreach\i in{1,...,\n}{
            \draw (\i*\vhdiff, 0) node[var] (v\i) {$\cdot$};
            \draw (\i*\vhdiff, -\avdiff) node[fun] (f\i1){};
            \draw (\i*\vhdiff, -\avdiffs) node[fun] (f\i2){};
            \draw[gray,dashed] (v\i) -- (f\i1);
            \draw[gray,dashed] (f\i1) --(f\i2);
                }
            \draw [decorate,decoration={brace,amplitude=10pt}] (f\n1.north east) -- (f\n2.south east) node [black,midway,yshift=0, xshift=43, align=left] {$\eta\log(n)$\\individual tests};
        
            % \draw [gray, dashed] (1, -\avdiffs-1) -- (\n*\vhdiff, -\avdiffs-1);
            \draw[fill = white] (\vhdiff-0.2, \thirdlevel-\avdiff) rectangle (\s*\vhdiff+0.2, \thirdlevel-\avdiffs) node[pos=.5](sublin1) {\SPOG};
            \foreach\i in {1,...,\s}{
                \draw (\i*\vhdiff, \thirdlevel) node[var] (v\i2) {$\cdot$};
                \draw[gray,dashed] (v\i2) -- (sublin1);
            }
           \draw[fill = white] (\vhdiff-0.2, \thirdlevel-\avdiff) rectangle (\s*\vhdiff+0.2, \thirdlevel-\avdiffs) node[pos=.5](sublin1) {\SPOG};
            \draw[gray,dashed](\s*\vhdiff+\vhdiff/2, -\avdiffs)--(\s*\vhdiff+\vhdiff/2, \secondlevel);

            \foreach\i in {\sr,...,\n}{
                \draw (\i*\vhdiff, \secondlevel) node[var] (v\i2) {$\cdot$};
                \draw (\i*\vhdiff, \secondlevel-\avdiff) node[fun] (f\i21){};
                \draw (\i*\vhdiff, \secondlevel-\avdiffs) node[fun] (f\i22){};
                \draw[gray,dashed] (v\i2) -- (f\i21);
                \draw[gray,dashed] (f\i21) --(f\i22);
            }
            \draw [decorate,decoration={brace,amplitude=10pt}] (f\n21.north east) -- (f\n22.south east) node [black,midway,yshift=0, xshift=43, align=left] {$ \frac{(1+\frac\eps4) \log(n)}{\KL{p_{11} }{p_{01}}}$\\individual tests};
             \draw [decorate,decoration={brace,amplitude=8pt}] (v12.north west) -- (v\s2.north east) node [black,midway,yshift=14, xshift=0] {$S_0$};
             \draw [decorate,decoration={brace,amplitude=9pt}] (v\sr2.north west) -- (v\n2.north east) node [black,midway,yshift=16, xshift=11] {$S_1: \abs{F^+(i)} \geq C\abs{F(0)}$};
             
            \draw[gray,dashed](\u*\vhdiff+\vhdiff/2, \secondlevel-\avdiffs)--(\u*\vhdiff+\vhdiff/2, \thirdlevel);

          \draw[fill = white] (\sr*\vhdiff-0.2, \thirdlevel-\avdiff) rectangle (\u*\vhdiff+0.2, \thirdlevel-\avdiffs) node[pos=.5](sublin2) {\SPOG};
            \foreach\i in {\sr, ..., \u}{
                \draw (\i*\vhdiff, \thirdlevel) node[var] (v\i3) {$\cdot$};   
                \draw[gray,dashed] (v\i3) -- (sublin2); 
            }
          \draw[fill = white] (\sr*\vhdiff-0.2, \thirdlevel-\avdiff) rectangle (\u*\vhdiff+0.2, \thirdlevel-\avdiffs) node[pos=.5](sublin2) {\SPOG};
            \foreach\i in {\ur, ..., \n}{
                \draw (\i*\vhdiff, \thirdlevel) node[var] (v\i3) {$+$};    
            }
            
           \draw [decorate,decoration={brace,amplitude=8pt}] (v\sr3.north west) -- (v\u3.north east) node [black,midway,yshift=14, xshift=0] {$U_0$};
         \draw [decorate,decoration={brace,amplitude=10pt}]  (v\n3.south east)--(v\ur3.south west) node [black,midway,yshift=-20, xshift=30, align=left] {$U_1: \abs{F_2^+(i)} $\\ $\geq (p_{11}-\epsuthresh)\abs{F_2(0)}$};
    
    \end{tikzpicture}
}
\newcommand{\napic}{
    \begin{tikzpicture}
        \def\n{5}
        \def\vhdiff{0.7}
        \def\avdiff{0.7}
        \def\gvdiff{1}
        \def\ivdiff{1}
        \def\avdiffs{1.5}
        \def\mg{5}
    
         \foreach\i in{1,...,\n}{
            \draw (\i*\vhdiff, 0) node[var] (v\i) {$\cdot$};
            \draw (\i*\vhdiff, -\avdiff) node[fun] (f\i1){};
            \draw (\i*\vhdiff, -\avdiffs) node[fun] (f\i2){};
            \draw[gray,dashed] (v\i) -- (f\i1);
            \draw[gray,dashed] (f\i1) --(f\i2);
                }
        \draw [decorate,decoration={brace,amplitude=10pt}] (f\n1.north east) -- (f\n2.south east) node [black,midway,yshift=0, xshift=43, align=left] {$F_1:$ $\ceil{\xi\log(n)}$\\individual tests};
        % \draw [decorate,decoration={brace,amplitude=10pt}] (f\n2.south east)--(f12.south west) node [black,midway,yshift=-20, xshift=0, align=left] {$F_1:$ $\ceil{\xi\log(n)}$\\individual tests};
        \foreach \a/\vl in {
            1/{1,2,3}, 
            2/{1,2,4}, 
            3/{2,3,1}}{
            \draw (\a*\vhdiff*\n/\mg, +\gvdiff) node[fun] (fg\a){};
            \foreach \i in \vl{
                \draw[gray,dashed] (v\i) -- (fg\a);
            }
        }
        \draw [decorate,decoration={brace,amplitude=10pt}] (fg1.north west) -- (fg3.north east) node [black,midway,yshift=20, xshift=0, align=left] {$F_{2,grp}$:\\group tests};
    
        \foreach\i in {4,5}{
            \draw (\i*\vhdiff, +\ivdiff) node[fun] (f\i3){};        
            \draw[gray,dashed] (v\i) -- (f\i3);
        }
        \draw [decorate,decoration={brace,amplitude=10pt}] (f43.north west) -- (f53.north east) node [black,midway,yshift=20, xshift=13, align=left] {$F_{2,idv}$: extra\\ individual tests};
        % \draw [decorate,decoration={brace,amplitude=7pt}] (\n*\vhdiff+0.2, \ivdiff+0.3) -- (\n*\vhdiff+0.2, \ivdiff-0.3) node [black,midway,yshift=0, xshift=43, align=left] {$F_{2,idv}$: extra\\ individual tests};        
    \end{tikzpicture}
}

% Unify:
% \begin{itemize}
%     %\item sick, infected,infected-> infected
%     %\item individuals, agents, patients -> individuals
%     \item american or british english? -> american
%     \item $\wedge$ or just $,$ in probabilities
%     %\item adaptive: Algorithm and Graph distinguishing
%     %\item $\NN^+$ or $\NN$
%     %\item $\eps>0$ or $\eps \in \mathbb{R}^+$ -> $\eps > 0$
%     %\item $e$ to $\eul$ where appropriate
%     %\item $\infected_\pat$ or $\infected(\pat)$ -> $\infected_\pat$
%     %\item write down which chernoff
%     %\item Spelling of quite (I am sorry, if have misspelled it "quiet" often)
%     %\item to and too
%     % \item KL div in adaptive part $\KL{p_{10}}{p_{00}}$ or $\KL{p_{11}}{p_{01}}$ ->$\KL{p_{11}}{p_{01}}$  
%     % \item group testing or Group Testing \lrh{i vote for ``group testing''}
%     \item change small numbers to $\delta_1$ and so on: (no, maybe at end)
%     \item non-adaptive, nonadaptive, non-adaptive
%     \item uninfected / non-infected (non infected) / not infected
% \end{itemize}

\newpage
\pagenumbering{gobble}
\tableofcontents
\newpage
\pagenumbering{arabic}
\setcounter{page}{1}

\section{Introduction}
Originally motivated by large-scale screening of soldiers for syphilis in the Second World War, the study of group testing dates back to the 1940s \cite{Dorfman}.
But its relevance exceeds its original inspiration by far,
    emphasized by the vast diversity of applications of group testing to other areas.
\cite{AldridgeJS209,cheraghchiNakosDecoding} list many applications, among them
    DNA library screening \cite{dnalibngo2000survey,dnalibscreening_bruno1995efficient,dnaschliep2003group,dnawu2006error},
    molecular biology \cite{molbicheng2008new,molbifarach1997group,molbiknill1995group,molbiwu2007molecular},
    multiple access communication \cite{mulcomwolf1985born},
    data compression and storage \cite{compressionhong2001group}, secure key distribution \cite{seckeychen2007unexpected},
    pattern matching \cite{patternmatchclifford2007k},
    and quality control \cite{qualcontrollsobel1959group}.
Yet, even though this problem has been of interest for over 80 years, and despite being easy to state, it is far from being completely understood.
This gives group testing an especially interesting appeal.

To describe the general problem setting, consider a large population of individuals, some of which are infected.
The aim is to identify infected individuals such that the number of tests is minimized.
To achieve this, one pools groups of individuals together,
    testing a whole group with just one test.
%We assume here that each individual can participate in an arbitrary number of tests and that a single test can group an arbitrary number of individuals.
In an idealized \emph{noiseless} setting, a test renders a positive result if and only if at least one participating individual is infected.
However, in reality, tests can show false positive and false negative results \cite{chlamydia,drosten,Jasima2022ComparativeEO,KahramanKlba2022AMO}.
We study a \emph{noisy} setting where tests err independently, i.e., a should-be-positive test is observed as negative with probability $p_{10}$, and a should-be-negative result is observed as positive with probability $p_{01}$.
In other words, idealized test results are distorted by a noisy binary channel.

% a test that contains an infected individual is observed as negative with probability $p_{10}$ and as positive with probability $p_{11}$. 
% Conversely, a test with no positive individuals displays negative with probability $p_{00}$ and positive with probability $p_{01}$. 
% The ultimate goal is to minimize the number of tests used to correctly identify every single individual \whp. 

A significant proportion of recent literature, discussed in detail below, studies the \emph{sparse regime}, where the number of infected people is \emph{sublinear} in the population size.
Noteworthy recent achievements include the complete understanding of noiseless group testing \cite{CojaOghlanGHL21} as well as the almost complete understanding of noisy group testing \cite{cojaoghlan2024noisygrouptestingspatial, scarlettchen}, both in the sparse regime.
In contrast, results for the linear regime studied here lag far behind.
For the noisy setting, they boil down to a non-tight lower bound on the number of necessary tests for the special case of symmetric noise \cite{Scarlett19}.
The lack of rigorous results in the linear regime is in stark contrast to the substantial amount of potential applications.
E.g., considering medical conditions, a constant prevalence is most intuitive (the early stages of a pandemic being a notable exception).
% Especially in endemic scenarios or inherited diseases a the probability of being infected is constant \cite{}.
Moreover, the actually available tests are far from being perfectly accurate \cite{chlamydia,drosten,Jasima2022ComparativeEO,KahramanKlba2022AMO}.
Of course, interest in noisy settings extends to the aforementioned variety of applications of group testing. 
Considering this, one might call noisy linear group testing one of the most realistic scenarios studied so far.

In group testing, one differentiates between adaptive and non-adaptive strategies:
In the former, tests are conducted in (multiple) rounds,
    where the choice of tests in one round may depend on the results of previous tests.
For the latter, there is only one round,
    so that the test design is fixed at the beginning of testing.
While this reduces latency, it comes at the expense of (potentially) more used tests.
Both variants have applications depending on the more expensive resource (testing time vs.\ the amount of tests), whereas the former is the more intuitive variant of group testing and was also used for example in pooled PCR tests for SARS-CoV-2 \cite{COVIDmutesa2021pooled,Singh2020EvaluationOP}.

In this work, we provide a rigorous and complete understanding of both adaptive and non-adaptive group testing in the aforementioned noisy setting for the linear regime.
To be precise, we provide exact thresholds on the number of tests such that exact recovery is impossible below the threshold, and possible above the threshold.
Moreover, we actually provide efficient algorithms for the latter case.

\subsection{Contribution}\label{sec:intro:contribution}

We consider group testing under the \emph{i.i.d.\ prior}, where the $n$ individuals are infected independently with constant probability~$\alpha$,\footnote{Note that this can be transformed to the other commonly studied studied prior, the \emph{combinatorial prior}, where the number of infected individuals is fixed; cf.\ Appendix to Chapter 1 in \cite{AldridgeJS209}.}
    with test results being observed through a noisy binary channel.
For both adaptive and non-adaptive schemes,
    we pin down the number $m=cn\log(n)$ of necessary and sufficient tests down to the constant $c$.
We call this threshold $\nonadm$ for the non-adaptive case and $\adm$ for the adaptive case.
To prove these thresholds, we develop a fundamental understanding of the combinatorial conditions necessary for the exact identification of infected individuals.
Furthermore, we harness this understanding
    to provide \emph{efficient} algorithms for both settings
        using $(1+\eps)\nonadm$ (resp., $(1+\eps)\adm$) tests.

We now give a brief overview of the main obstacles for estimating an individual's infection status.
In the non-adaptive case,
    each individual requires tests whose results change when the individual's status changes; call such tests \emph{good} (for said individual).
With a fixed test design, the number of good tests fluctuates depending on the random infection statuses.
While there is no fixed lower bound on the necessary number of good tests required for \emph{every single} individual,
    there need to be sufficiently many tests so that there are not too many individuals that don't have enough good tests.
If not, even estimators having oracle access to the ground truth except for the individual in question (\emph{genie-based estimators}) fail.

In the adaptive case, on a conceptual level, there are two important differences from the non-adaptive case:
First, an adaptive test design can ensure that most tests are used on infected individuals.
Second, an adaptive test design can ensure that almost all of these tests are good for some (infected) individuals, removing the fluctuation of the non-adaptive case.  
In this sense, the main barrier is that all infected individuals need to participate in enough tests such that there are enough good tests to exclude the possibility of the individual \emph{not} being infected after all.

Our efficient algorithms, the non-adaptive $\SPOG$ (\emph{\textbf{s}ynthetic \textbf{p}seud\textbf{o}-\textbf{g}enie}) and the three-stage adaptive $\PRESTO$ (\emph{\textbf{pre}-\textbf{s}orting \textbf{t}hresh\textbf{o}lder}), overcome these obstacles with a test count just above the thresholds.
$\SPOG$ uses almost all its tests on random groups of an optimal size $\Gamma$ depending on both $\alpha$ and a measure of noisiness.
Few individual tests are used to construct a ``synthetic pseudo-genie'';
    $\SPOG$'s output is based on the group tests which are good by the pseudo-genie.
Similarly, $\PRESTO$'s first stage uses a few individual tests to pre-sort individuals into a mostly-infected and a mostly-uninfected set.
Its second stage weeds out all uninfected individuals among the former by thresholding on many individual tests.
The final stage uses $\SPOG$ twice to identify the few infected individuals among the first stage's mostly uninfected set, as well as the second stage's rejects.
We note that both algorithms classify almost all individuals correctly as pre-processing based on only a negligible fraction of all tests.
Hence, approximate recovery requires significantly fewer tests in both settings; determining the exact number of tests required for this remains open.

\begin{figure}\centering
    \vspace{3mm}
    \subfigure[Symmetric noise]{\includegraphics[trim={0.5cm 0.4cm 1.6cm 1.5cm}, width=0.45\linewidth]{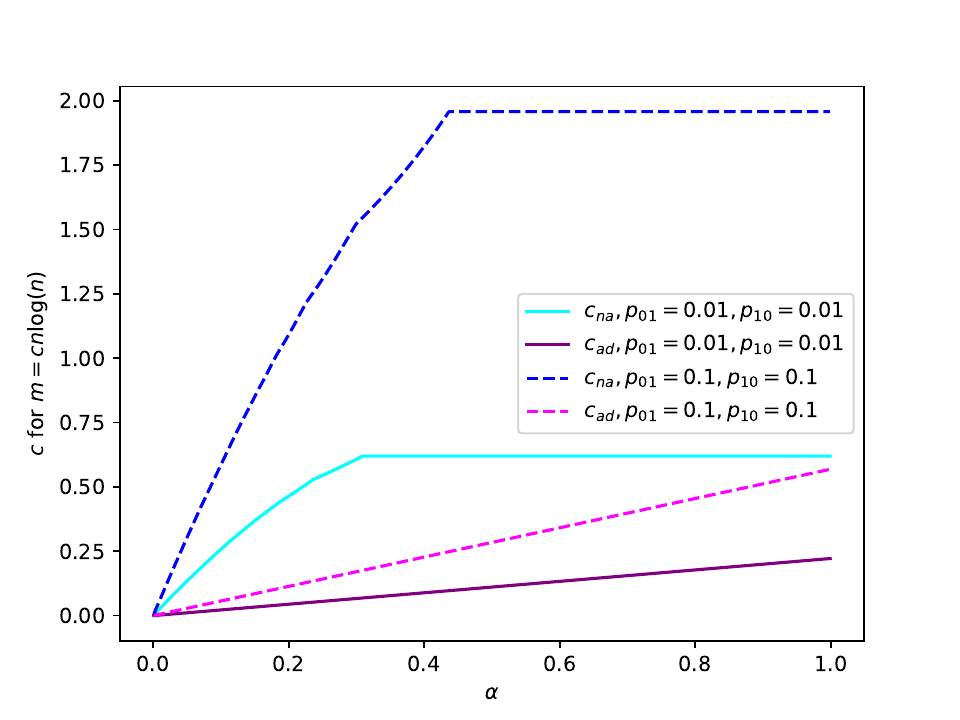}
         \label{fig:plotsym}
    }
    \subfigure[Asymmetric noise]{
         \includegraphics[trim={0.5cm 0.4cm 1.6cm 1.5cm}, width=0.45\linewidth]{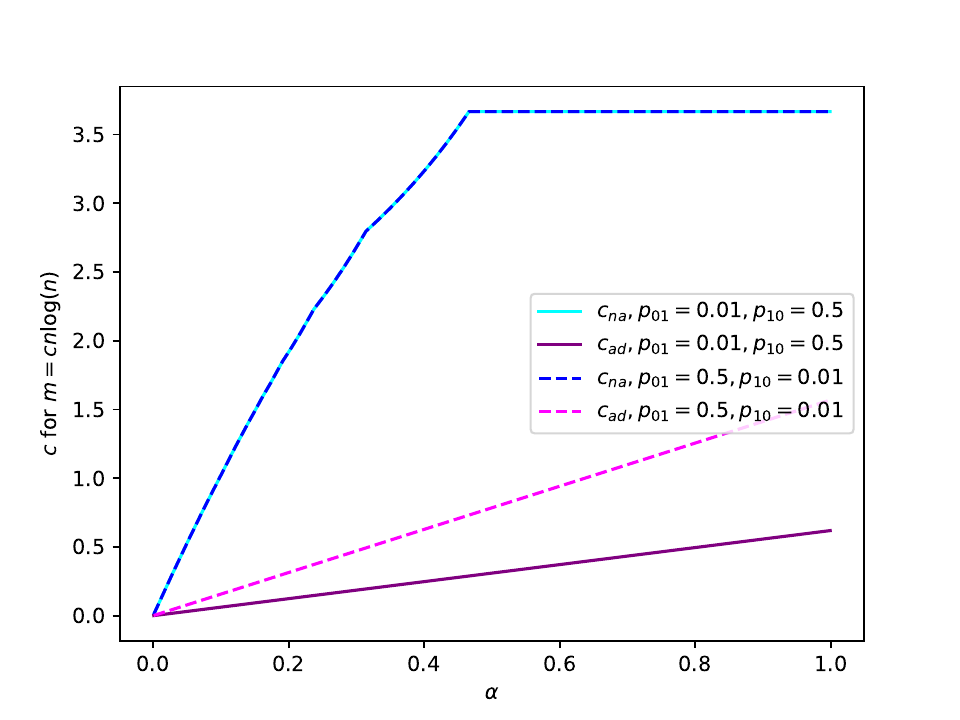}
         \label{fig:plotasym}
    }
    \caption{Comparison of adaptive and non-adaptive thresholds $\cad=\adm/(n\log n)$ and $\cna=\nonadm/(n\log n)$.}
    \label{fig:threshplot}
\end{figure}

\cref{fig:threshplot} depicts the thresholds' constants $c$ for different values of $\alpha$ and different channels.
It is quite blatant how adaptivity saves a significant amount of tests in all depicted scenarios:
    in \Cref{fig:plotsym}, we see that for a symmetric channel, the adaptive algorithm uses fewer tests for even tenfold higher noise.
Additionally, viewing the graphs of the non-adaptive threshold, we see that they do not seem to be differentiable for all $\alpha$.
Indeed, these ``kinks'' appear due to the optimization over the (discrete) test degrees.
On another note, \Cref{fig:plotasym} shows two additional interesting properties:
First, the non-adaptive threshold is not sensitive to switching $p_{01}$ and $p_{10}$, as the blue and cyan lines overlap.
On the other hand, in the adaptive case, false positive results require more tests than false negatives.

\subsection{Related work}\label{sec:intro:related}
% \todo{check approximate recovery literature!}
% \begin{table}[]
%     \centering
%     \begin{tabularx}{\linewidth}{l p{0.44\linewidth} X}
%     \toprule
%                     & Sublinear &  Linear\\ \midrule
%     Noiseless       & Some more cites.. &  Cutoff $\alpha$ for group tests \cite{Ungar} \\
%                     & Sharp bounds \& efficient algorithms \cite{CojaOghlanGHL21} & Non tight upper bounds \cite{AldridgesBumpyRide}\\
%     Symmetric noise & Sharp bounds on restricted test designs \cite{scarlettchen} & Non tight lower bounds \cite{Scarlett19} \\
%     General noise & Sharp bounds \& efficient algorithms \cite{cojaoghlan2024noisygrouptestingspatial} &  \textbf{This work:} Sharp bounds \& efficient algorithms\\ \bottomrule 
%     %             sublinear 
%     %             & Sharp Thresholds and efficient algorithms: \cite{CojaOghlanGHL21} 
%     %             & Sharp thresholds for restricted test designs: \cite{scarlettchen} 
%     %             & Sharp thresholds and efficient algorithms: \cite{cojaoghlan2024noisygrouptestingspatial} \\
%     % linear & & Non tight lower bound on tests: \cite{Scarlett19} & This work: Sharp Thresholds and efficient algorithms \\
%     % \bottomrule
%     \end{tabularx}
%     \caption{Overview on results of group testing}
%     \label{tab:my_label}
% \end{table}

In recent years, there has been a variety of work on group testing and its variants.
An excellent overview of results prior to 2019 can be found in \cite{AldridgeJS209}.
As mentioned above, the overwhelming majority of recent rigorous results focus on the case of a sublinear number of infected individuals.
Following a considerable amount of prior work \cite{aldridge14,scarlett2016phase,johnson2018performance,coja2020information}, in the sublinear regime, noiseless group testing is now understood completely \cite{CojaOghlanGHL21}, while noisy non-adaptive group testing is understood almost completely \cite{cojaoghlan2024noisygrouptestingspatial, scarlettchen}. 
However, results for the linear regime are much more limited.
In the noiseless case, one cannot improve on individual testing once the fraction of infected individual exceeds $\alpha = \frac{1}{2}(3-\sqrt{5})$ \cite{Ungar}.
Below this threshold, one has to differentiate between non-adaptive and adaptive group testing.
For non-adaptive group testing, one cannot do better than individual testing either \cite{Aldridge19}.
In the adaptive case, the most important contribution is an algorithm \cite{bumpyride} based on binary splitting \cite{HwangSplitting}.
However, determining a tight threshold for noiseless linear adaptive group testing remains open. 

If we turn from noiseless to noisy group testing rigorous results are even more sparse. 
To the best of our knowledge, the only result in this setting is a lower bound on the necessary tests for the special case of symmetric noise \cite{Scarlett19}, which is a factor $(1-2p_{01})^{-1}$ smaller than the actual bound.
However, \cite{Scarlett19} does not consider the question of achievability not to mention the question of efficient algorithms. 
A recent contribution \cite{cheraghchiNakosDecoding} studies the decoding complexity of group testing, i.e., the possibility of finding the infected individuals in optimal, possibly sublinear time.
Leaving classical group testing there is a huge diversity of variants of this problem, such as quantitative group testing \cite{feige2020quantitativegrouptestingrank,SoleymaniEtAl24,Tan24qgt}, pooled data \cite{hahn-klimroth22apooled,hahn-klimroth23pooled}, tropical group testing \cite{Wang23tropical,Paligadu24} and many more. 

The present work's proofs are partially inspired by techniques from community detection such as the use of a genie-based estimator for impossibility proofs \cite{jmlr/Abbe17} as well as focusing on a small portion of individuals that are already hard to estimate \cite{arxiv/AbbeS15}.

\section{Model and Results}\label{sec:model}
\paragraph{Model and commonly used notation}
We consider group testing under the \emph{i.i.d.\ prior},
    where the infection statuses of all of the $n$ individuals
        are independent, with each individual's infection probability being $\alpha \in (0, 1)$.
We call a $\pat \in \{0, 1\}^n$ an \emph{infection vector},
    where $\pat(i)=1$ if $i$ is infected under $\pat$, and $\pat(i)=0$ otherwise.
We then write $\pats$ for the random infection vector (the \emph{ground truth}) as described above.
Often, we consider a test design as a bipartite graph $G = (V, F, E)$,
    with $V = V(G) = [n]$ being the set of individuals,
    $F = F(G)$ being the set of tests,
    and the set $E$ consisting of the edges between individuals and tests,
        with an edge being present if and only if a given test contains a given individual.
We write $\soG_m$ for the set of all such graphs with $m$ tests.

For a test $a \in F$, we write $\partial_G a$ for the set of individuals contained in $a$,
    for an individual $i \in [n]$, we write $\partial_G i$ for the set of tests containing $i$.
When $G$ is clear from context (as is the case most of the time), we omit the subscript.
If the tests were perfectly accurate,
    then a test would appear positively if and only if it contains an infected individual.
So writing $\atests = (\atests(a))_{a \in F} \in \{0, 1\}^F$ for the vector of these hypothetical test results,
    we have $\atests(a) = \max\cbc{\pats(i) \mid i \in \partial a}$.
However, we consider the test results as being transmitted through a noisy channel:
    We write $\otests$ for the vector of observed test results.
    $\otests(a)$ depends only on $\atests(a)$, with the channel noise being independent for all tests.
Often, for a set of tests $S \subseteq F$,
    we write $S^+ = \cbc{a \in S \mid \otests(a) = 1}$ for the subset of $S$ with positive observed test result,
        and $S^- = S \setminus S^+$ for the subset of $S$ with negative observed test result.

For $k,l \in \{0, 1\}$,
    we let $p_{kl} \in \bc{0,1}$ be the probability that a test with hypothetical test result $\atests(a) = k$
    displays as $l$.
So formally, for any $\atest, \otest \in \{0, 1\}^m$,
    \(\ProbCond{\otests = \otest}{\atests = \atest} = \prod_{j \in [m]} p_{\atest_{j} \otest_{j}}.\)
We write $\vp = (\pflipoo, \pflipoi, \pflipio, \pflipii)$ for the vector of these probabilities,
    and call $\vp$ a \emph{noisy channel}.
Note that we must have $\pflipoo + \pflipoi = \pflipio + \pflipii = 1$,
    that $\pflipii = 1-\pflipio$ is the \emph{sensitivity} (true positive rate) of the test,
    and that $\pflipoo = 1-\pflipoi$ is the \emph{specificity} (true negative rate) of the test.
Additionally, we require that $\pflipoi + \pflipio < 1$:
    if $\pflipoi + \pflipio = 1$, then tests containing, or not containing, an infected individual would be indistinguishable (the channel has capacity $0$);
    and if $\pflipoi + \pflipio > 1$, one can flip all observed test results to obtain a channel with $\pflipoi + \pflipio < 1$.
We call a noisy channel satisfying these conditions \emph{valid}.

For $p, q \in [0, 1]$, we write $\KL{p}{q}$ for the Kullback--Leibler divergence (or relative entropy) of $\ber{p}$ from $\ber{q}$, i.e.,
    \(\KL{p}{q} = p\log\frac{p}{q} + (1-p)\log\frac{1-p}{1-q}.\)
Furthermore, for any valid noisy channel $\vp$,
    define 
    \begin{align}\label{eq:beta}
        \beta = \beta(\vp) = \max_{c \in [\pflipoi, \pflipii]} \min\cbc{\KL{c}{\pflipoi}, \KL{c}{\pflipii}},
    \end{align}
    and let $C = C(\vp)$ be that $c$ where the maximum is obtained.
Note that as $\pflipoi, \pflipii \in (0, 1)$,
    we have $\beta(\vp) = \KL{C(\vp)}{\pflipoi} = \KL{C(\vp)}{\pflipii}$,
    and $0 < \beta(\vp) < \infty$ for any valid noisy channel $\vp$ (see \Cref{lem:c-thr} for a proof). 
    
Finally, we say that an event occurs \emph{with high probability} (\whp) if it occurs with probability $1-o_n(1)$.

\paragraph{Results}
\label{sec:results}

\newcommand{\novars}{\vp, \alpha, \eps}
\newcommand{\fapre}{For any valid noisy channel $\vp$, $\alpha>0$ and $\eps>0$, there exists some $n_0 = n_0(\novars)$ such that for every $n>n_0$}

We now state our results formally.
For the non-adaptive case, we show in \cref{sec:nonadative_lb} that for
\begin{align}\label{def:ma}
    \mna = \mna(\alpha, \vp) = \min_{\Gamma \in \NN^+} \frac{n\log n}{-\Gamma \cdot \log\bc{1 - (1-\alpha)^{\Gamma - 1} \cdot \bc{1 - \eul^{-\beta(\vp)}}}}
\end{align}
    when $(1 - \varepsilon)\mna$ tests are used,
        any non-adaptive algorithm will fail with high probability:

\begin{theorem}\label{thm:nalower}
    \fapre, all test designs $G$ with $m<(1-\eps)\mna(\alpha, \vp)$ tests and every estimation function $f_G:\cbc{0,1}^m \rightarrow \cbc{0,1}^n$:
        \begin{align*}
            \SmallProb{f_G(\otests_G) = \pats} \leq \eps.
        \end{align*}
    \end{theorem}

Furthermore, we show that the non-adaptive algorithm $\SPOG$ given in \cref{sec:nonadative_ub}
    recovers $\pats$ with high probability using $(1+\varepsilon)\mna$ tests.
\begin{theorem}\label{thm:naalg}
    % For all noisy channels $\vp$, $0 < \alpha < 1$, $\varepsilon > 0$
    % if $m>(1+\eps)\nonadm $ then there exists a $n_0$, a randomized test design $G$ and a estimation function $g_G$ s.t.\ for all $n > n_0$
    \fapre, there is a randomized test design $\vG$ using $m \leq (1+\eps)\nonadm(\alpha, \vp)$ tests \whp and a deterministic polynomial time algorithm $\SPOG$ such that\
    \begin{align*}
        \SmallProb{\SPOG(\vG, \otests_{\vG}) = \pats} \geq 1 - \eps.
    \end{align*}
\end{theorem}

Our second pair of results states that
    $\adm = \adm(\alpha, \vp) = \frac{\alpha}{\KL{p_{11}}{p_{01}}} \cdot n \log n$
    is the same kind of threshold for adaptive schemes.
In \cref{sec:adaptive_impossible},
    we see that when only $(1 - \varepsilon)\adm$ tests are used,
    any adaptive algorithm fails with high probability:

\begin{theorem}\label{thm:adp:low}
    \fapre, any
    adaptive test scheme $\alg$ using $m_{\alg} < (1-\eps) \adm(\alpha, \vp)$, and any estimation algorithm $f_{G_{\alg}} :\cbc{0,1}^{m_{\alg}} \rightarrow \cbc{0,1}^n$:
    \begin{align}
         \SmallProb{f_{G_{\alg}}(\otests_\alg) = \pats} \leq \eps.
    \end{align}
\end{theorem}

And we show that the non-adaptive algorithm $\PRESTO$ given in \cref{sec:nonadative_ub}
    recovers $\pats$ with high probability using $(1+\varepsilon)\adm$ tests.
\begin{theorem}\label{thm:adp:upper}
  % if $m>(1+\eps)\adm$ then there exists a $3-$stage adaptive test design and a algorithm that 
  \fapre, the three-stage adaptive test scheme \PRESTO\ uses at most $m\leq(1+\eps) \adm(\alpha, \vp)$ 
  % \lrh{I think technically the number of tests is \whp. I think we talked about this before (that one could just ``abort'' and then that gets added to the failure probability.} 
  \whp such that
  \begin{align}
     \SmallProb{\PRESTO(\otests_\PRESTO) = \pats} \geq 1 - \eps.
  \end{align}
\end{theorem}

% \subsection{Organization}

\section{Preliminaries}
\paragraph{MAP estimation}
For the impossibility results we have to prove that \emph{any} estimator fails \whp if not supplied with enough tests.
Since dealing with an arbitrary estimator is tedious, instead we first consider the estimator maximizing the probability of exactly recovering $\pats$, namely the maximum a posteriori (MAP) estimator, which chooses the infection vector that maximizes the a posteriori probability given observed test results.

\begin{definition}[MAP estimate]\label{def:map}
    For any test design $G\in\soG_m$ and any observed results $\otest\in \cbc{0,1}^m$,
        the \emph{MAP estimate} $\map(G, \otest) = \map^{g,\otest}$ of $\pats$ is given by 
    \begin{align}\label{eq:map-argmax}
        \map^{G,\otest} = \argmax_{\hat{\sigma}\in \cbc{0,1}^n} \SmallProbCond{\hat{\sigma}= \pats}{G,\otest},
    \end{align}
        where we choose that $\hat{\sigma}$ with the most zeros if the $\argmax$ is not unique.
\end{definition}

As is commonly known, no estimator is better than the MAP estimator  for achieving exact recovery (see, e.g., \cite[Section 3.2]{jmlr/Abbe17}):

\begin{fact}[MAP estimator is the best estimator given the tests and observed test results] \label{lem:map-best}
     For any test design $G\in\soG_m$ and any estimator $f:\soG_m\times  \cbc{0,1}^m\mapsto \cbc{0,1}^n$, $f$ is at most as good as the MAP estimate, i.e., 
    \begin{align*}
        \Prob{\pats = \map^{G,\otests}} \geq \SmallProb{\pats = f(G,\otests) }.
    \end{align*}
\end{fact}

\paragraph{Notation used throughout}
For brevity, for any infection vector $\pat \in \{0,1\}^n$, we write $\infected_\pat$ for the set of infected individuals under $\pat$, i.e., $\infected_\pat = \cbc{i \in [n] \mid \pat(i) = 1}$, and hence $\overline{\infected_\pat} = [n] \setminus \infected_\pat$ is the set of uninfected individuals.
Consequently, $\infected_{\pats}$ is the (random) set of infected individuals as in the model.

Furthermore, throughout our paper, we distinguish between tests that are \emph{good} for estimating an individual and those that are not. 
This differentiation of tests has been commonly used in group testing (see, e.g., \cite{cojaoghlan2024noisygrouptestingspatial, Aldridge19}).
Intuitively, a good test for an individual actually carries information about the individual's infection status, which is only the case if no other infected individual is included.
\begin{definition}[good tests]\label{def:goodtest}
    For a given test design $G$,
        a test $a$ is \emph{good} for an individual $i$ under the infection vector $\pat \in \{0, 1\}^i$
        if $a$ contains $i$,
        and no individual in $a$, except for (perhaps) $i$, is infected.
    We write $\good_i(\pat, G)$ for the number of such tests for $i$ under $\pat$.
    For a given vector $\otest$ of observed test results,
        we also write $\good_i^-(\pat, G, \otest)$ (resp., $\good_i^+(\pat, G, \otest)$) for the number of such tests displaying negatively (resp., positively).
    When $\pat$ is clear from context, we may simply write $\good_i$ (etc.).
    Formally,
    \begin{align*}
        \good_i(\pat, G) &= \abs{\cbc{a \in \partial i \mid \bc{\partial a\setminus\cbc{i}} \cap \infected_\pat = \emptyset }} \\
        \good_i^-(\pat, G, \otests) &= \abs{\cbc{a\in \partial i \mid \otests(a)=0 
        \land \bc{\partial a\setminus\cbc{i}} \cap \infected_\pat= \emptyset}} \textup{ and }  \good_i^+(\pat, G, \otests) =  \good_i(\pat, G) - \good_i^-(\pat, G, \otests)\,.
        % \good_i^+(\pat, G, \otests) &= \abs{\cbc{a\in \partial i \mid \otests(a)=1 \wedge 
        %  \bc{\partial a\setminus\cbc{i} \cap \infected_\pat} = \emptyset}}.
        \end{align*}
    As in other places, we omit the parameters when the values are clear from context.
\end{definition}

\section{Non-adaptive Group Testing}
In this section, we outline the proof strategy for both impossibility as well as achievability in the non-adaptive setting.
In both,
    we (conceptually) use an estimator which is even better than the MAP estimator,
    the so-called \emph{genie-based estimator}.
This technique is inherited from community detection \cite{jmlr/Abbe17,arxiv/AbbeS15} but so far not used in the context of group testing to the best of our knowledge.

\begin{definition}[Genie estimator]\label{def:genie}
    Given a test design $G$ and the observed results $\otest$, with $\pats_{-i}$ the ground truth for every other individual except $i$, the \emph{genie-based estimator} (or just \emph{genie estimator}) is given by
    \begin{align}\label{eq:genie}
        \genie^{G,\otest}(i)=\genie^{G,\otest,\pats_{-i}}(i) = \argmax_{s\in \{0,1\}} \ProbCond{\pats(i) = s }{ G,\otests=\otest,\pats_{-i}}.
    \end{align}
    If both $0$ and $1$ maximize $\ProbCond{\pats(i) = s }{ G,\otests=\otest,\pats_{-i}}$, let $\genie^{G,\otest}(i)=0$.
\end{definition}

Intuitively, it is clear that the genie estimator outperforms the MAP estimator:
    it has access not only to the test design and the observed results but also to the whole ground truth $\pats_{-i}$ to determine $\pats(i)$.
So proving that even the genie estimator fails when the test design contains too few tests---as done in \cref{sec:nonadative_lb}---implies that all estimators fail.
The following lemma, proven in \cref{apx:geniebetter}, formalizes this intuition.
\begin{restatable}[Genie estimator is better than MAP estimator]{lemma}{geniebetter}
 \label{lem:geniebetter}
For every test design $G$, 
    \(\Prob{\pats = \genie^{G,\otests}} \geq \SmallProb{\pats = \map^{G,\otests}}.\)
\end{restatable}
   
Of course, the genie estimator is not realizable:
    when classifying a given individual,
        it requires oracle access to the true infection status of all other individuals.
However, our algorithm $\SPOG$ presented in \cref{sec:nonadative_ub}  uses the concept by emulating the oracle access well enough for the classification to still be correct \whp. 

\subsection{Impossibility: proof strategy for \Cref{thm:nalower}}\label{sec:nonadative_lb}
For the lower bound in the non-adaptive case, we consider an arbitrary test design $G$ with fewer than $(1-\eps) \mna$ tests.
First, we modify the test design to achieve some properties crucial to our analysis without decreasing the success probability of the genie estimator significantly, if at all.
We identify individuals that are especially hard to estimate for the genie estimator when using $G$.
Then, we show that there must be a large set of such individuals with independent tests, such that the genie estimator fails \whp.

To execute this plan,
    we first need to characterize the genie estimator.
Observe that its estimate for an individual $i$ depends only on the observed results of good tests for $i$, i.e., tests that do not contain any other infected individual.
The following lemma shows that the genie estimator is equivalent to a
thresholding function
    on the fraction of positively displayed good tests. A proof can be found in \Cref{apx:lem:geniethr}.
Note that the $C$ defined in the lemma statement is equal to the maximizer of \cref{eq:beta} as we prove in \Cref{lem:c-thr}.

\begin{restatable}{lemma}{geniethr} \label{lem:geniethr}
    Let $G$ be a test design, $C= \ln(\frac{p_{00}}{p_{10}}) / \ln(\frac{p_{11}p_{00}}{p_{01}p_{10}})$, and $\tconst = \tconst(\alpha,\vp) = \ln(\frac{\alpha}{1-\alpha}) / \ln(\frac{p_{01} p_{10}}{p_{11}p_{00}})$. 
    Then 
    \begin{align*}
        \genie^G(i) = \begin{cases}
            0 & \quad \textup{if $\good_i(\pats)=0$ and $\alpha \leq \frac{1}{2}$, or $\good_i^+(\pats) \leq C \good_i(\pats)  + \tconst$},\\
            1 & \quad\textup{otherwise.}
        \end{cases}
    \end{align*}
\end{restatable}

We now describe our modifications to the test design.
First, we limit the degrees of both individuals and tests in the modified design.
This enables us to find a sufficiently large set of individuals such that their tests, and hence their genie estimates, are independent.
As for tests, we show the following in \Cref{apx:lem:degtest}:
\begin{restatable}{lemma}{degtest} \label{lem:degtest}
    If $m = \Oh(n\log(n))$, then \whp, all tests $a$ with $\abs{\partial a} \geq \log^2(n)$ contain at least two infected individuals.
\end{restatable}
Consequently,
    all such tests are not good for \emph{any} individual \whp,
    so they are not considered by the genie estimator,
Hence, they are safe to remove without affecting the genie estimator's success probability too much.

As for the individuals, we remove all which participate in too many tests, and assume that they are correctly estimated.
This cannot decrease the genie estimator's success probability.
Additionally, for each individual, we add $\eta \log(n)$ individual tests for some small $\eta>0$ to be determined later.
This enforces a lower bound on the number of good tests for each individual,
    which we use in a few calculations. 
To summarize:
\begin{definition}\label{def:G_eta}
    Given test design $G$, we the \emph{modified test design} $\modGraph$ is constructed as follows:
    \begin{itemize}
        \item remove tests in $L = \cbc{a \in F \mid \abs{\partial_G a} \geq \log^2(n) }$ to obtain design $G'$,
        %\item $-$ individuals with very many small tests (with more than $Lm\eta$ test of size smaller than $L<1/\alpha$). $\to \patsmod$
        \item remove individuals in $J = \cbc{i\in[n] \mid \abs{\partial_{G_{\eta,L}} i}>\log^4(n) - \ceil{\eta \log(n)}}$ to obtain design $G''$, and
        \item add $\lfloor\eta \log(n) \rfloor$ individual tests for each individual $i$ with $\eta>0$ to obtain design $\modGraph$. \qedhere
    \end{itemize}
\end{definition}
Note that once we remove individuals with a large degree we still have 
$\nnew= n-cn\log(n)\log^2(n)/\log^4(n) = n-cn/\log(n) = \Theta(n)$ 
% $\nnew= n-\frac{cn\log(n)\log^2(n)}{\log^4(n)} = n-\frac{cn}{\log(n)} = \Theta(n)$ 
individuals left to be estimated.
W.l.o.g., we assume that these are the first $\nnew$ individuals $[\nnew]$.
The following lemma confirms the intuition that the genie estimator can only be improved by the adjustments in the modified test design; we prove it in \cref{apx:lem:modifiedsetup}.

\begin{restatable}[modified test setup is easier]{lemma}{modifiedtest}
\label{lem:modifiedsetup}
    Let $\pats[\nnew]$ be the infection statuses for individuals in $[\nnew]$ and $\hat{\sigma}[\nnew]$ be the prediction of the infection statuses for individuals in  $[\nnew]$ based on estimator $\hat{\sigma}$. Then
    \begin{align}
        \label{eq:easierlong}
         \Prob{\pats[\nnew] = \genie^{\modGraph}[\nnew]}
        \geq\Prob{\pats[\nnew] = \map^{\modGraph}[\nnew]}
        \geq \Prob{\pats = \map^{G'}}
        \geq \Prob{\pats = \map^{G}} - o(1) \; .
    \end{align}
\end{restatable}

%For the sake of brevity we will write $\patsmod$ for $\pats[\nnew]$.
% \todo{does this need to be defined?}
Now in this modified test design, we find a relatively large set of individuals whose tests are independent.
This is ensured by not including individuals whose tests intersect,
    i.e., by making sure that their second neighborhoods are disjoint.
We call such sets of individuals \emph{distant}.
We choose such a set greedily by repeatedly selecting an individual $i$ (by some arbitrary criterion) and removing every individual of distance at most four from $i$ in $\modGraph$.
Since there are at most $(\log^4(n) \cdot \log^2(n))^2 = \log^{12}(n)$ removed in every iteration (by our limit on the respective degrees),
    we can repeat this at least $\nnew/\log^{13}(n)$ times without running out of individuals.

Our goal is to show that in any modified test design we can choose a distant set such that \whp at least one individual in this set is misclassified.
To that end, for any (distant) set of individuals $D$,
    let $\misD$ be the set of those individuals in $D$ that are misclassified  by the genie estimator on $\modGraph$, i.e.,
    % \begin{align}
        $\misD = \{i \in D \mid \genie^{\modGraph}[\nnew](i) \neq \pats[\nnew](i)\}$. 
    % \end{align}
\newcommand{\modgood}{\tilde{\good}}

We first bound the probability that the genie estimator is correct, conditioned on $\pats$.
The proof in \cref{apx:lem:genie_correct_upper_bound} exploits the fact that correctness for all individuals
    requires correctness for any given distant set of individuals.
Then, we apply Chebyshev's inequality to $\abs{\misD}$,
    also using the independence of distant individuals' tests.

\begin{restatable}{lemma}{geniecorrectupperbound}\label{lem:genie_correct_upper_bound}
    For any set of distant individuals $D \subseteq [\nnew]$,
        \(\SmallProbCond{\pats[\nnew] = \genie^{\modGraph}[\nnew]}{\pats}
            \leq \bc{\ExpCond{\abs{\misD}}{\pats}}^{-1}.\)
\end{restatable}

To bound $\ExpCond{\abs{\misD}}{\pats}$ from below,
    we obtain the following bound on the probability that any specific individual in $D$ is misclassified by the genie estimator (conditioned on $\pats$) in \cref{apx:lem:genie_problow}. 

\begin{restatable}[probability of misclassification]{lemma}{genieproblow}\label{lem:genie_problow}
Let $i \in [\nnew]$,
    and let $0 < \delta < 1 - p_{10} - C$.
Then for sufficiently large $n$, the probability that $i$ is misclassified by the genie estimator is bounded from below as
    \begin{align}
        \ProbCond{\genie^{\modGraph}(i)\neq \pats(i) }{\pats} \geq \exp\bc{-\good_i(\pats, \modGraph)\KL{C-\delta}{p_{11}}}.
    \end{align}
\end{restatable}

All that remains is to find a distant set $D$ such that the sum of the lower bounds given by \cref{lem:genie_problow} over all $i \in D$ is large for most $\pats$.
The following lemma achieves this.
\begin{restatable}{lemma}{lematozero}\label{lem:atozero}
For any $\eps>0$, there exist $\delta$ and $\eta>0$ such that 
    for any $\delta'' > 0$, there is  a $n_0(\delta'')$ so that for all $n \geq n_0(\delta'')$, any test design $G$ on $n$ individuals using at most $(1-\eps)\mna$ tests, there is a distant set $D \subseteq [\nnew]$ so that
    \begin{align}
        \label{eq:atozero}
         \Prob{\sum_{i\in D} \exp\bc{-\KL{C-\delta}{ p_{11}}  \good_i(\pats, \modGraph)} > \frac{1}{2\delta''}} \geq 1 - 4\delta''.
    \end{align}
\end{restatable}

The proof in \cref{apx:lem:atozero}
    exploits the fact that for any distant set $D$,
        the sum in question
        is a sum of $\abs{D}$ independent and bounded random variables.
So it is well-concentrated, and it suffices to bound it from below in expectation.
That bound in turn hinges on the fact that the events of tests being good for individuals are monotonic in $\pats$ which enables the application of the FKG inequality.
Using the fact that a test $a$ is good for an individual with probability $(1-\alpha)^{\abs{\partial a}-1}$,
    some further calculation shows that,
        writing $b = \KL{C-\delta}{p_{11}}$,
    \[\log \bc{\Exp{\exp(-b\good_i(\pats,\modGraph))}} \geq \sum_{a \in \partial i}\log(1-(1-\alpha)^{\abs{\partial a}-1}(1-e^{-b})).\]
When choosing the distant set $D$ greedily to maximize the right-hand side,
    the average of this over $D$ is at worst just slightly below the average over \emph{all} individuals.
This average, after exchanging the order of summation and grouping terms, is
\[\frac{1}{n} \sum_{i \in [n]} \sum_{a \in \partial i} \log(1-(1-\alpha)^{\abs{\partial a}-1}(1-e^{-b}))
    = \frac{1}{n} \sum_{a \in F(\modGraph)} \abs{\partial a} \log(1-(1-\alpha)^{\abs{\partial a}-1}(1-e^{-b})).\]
For $\abs{\partial a} = \Gamma$,
    and $b \stackrel{\delta \to 0}{\to} \beta$,
    this is the value optimized over in the definition of the threshold $\mna$.
    
Now we plug everything together to conclude the proof of \cref{thm:nalower}.

\begin{proof}[Proof of \cref{thm:nalower}]
    For sufficient large $n$, the $o(1)$ term in \cref{lem:modifiedsetup} is at most $\eps/2$. It now suffices to show that the probability of the genie estimator being correct on $\modGraph$ is at most $\eps/2$ for sufficiently large $n$.
    To bound the probability of the genie estimator being correct on $[\nnew]$,
        combining \cref{lem:genie_correct_upper_bound,lem:genie_problow},
        for any set of distant individuals $D \subseteq [\nnew]$,
    \begin{align}
    \ProbCond{\pats[\nnew] = \genie^{\modGraph}[\nnew]}{\pats}
        \leq \biggl(\sum_{i \in D} \exp\bc{-\good_i(\pats, \modGraph) \KL{C-\delta}{p_{11}}}\biggr)^{-1}.\end{align}
    Now write $\bm U_D$ for this upper bound.
    By \cref{lem:atozero}, for any $\delta''$ and sufficiently small $\delta''$,
        we can choose $D$ so that $\Prob{\bm U_D \leq 2 \delta''} \geq 1 - 4 \delta''$.
    And as trivially $\SmallProbCond{\pats[\nnew] = \genie^{\modGraph}[\nnew]}{\pats} \leq 1$,
        we have
    \begin{align}
    \Prob{\pats[\nnew] = \genie^{\modGraph}[\nnew]}
       &= \Exp{\ProbCond{\pats[\nnew] = \genie^{\modGraph}[\nnew]}{\pats}}
        \leq \Exp{\min\cbc{1, \bm U_D}}
    \\ &\leq \ExpCond{\min\cbc{1, \bm U_D}}{\bm U_D \leq 2 \delta''} + \Prob{\bm U_D > 2\delta''}
        = 2\delta'' + 4\delta'',
    \end{align}
        which yields the theorem with $\delta'' = \eps/12$.
\end{proof}

\subsection{The algorithm $\SPOG$: proof strategy for \Cref{thm:naalg}}
\label{sec:nonadative_ub}

\newcommand{\patsSPG}{\hat{\pats}^{(1)}}
\newcommand{\patsNA}{\hat{\pats}_{\SPOG}}
\newcommand{\patSPG}[1]{\hat{\pats}^{(1)}(#1)}
\newcommand{\patNA}[1]{\hat{\pats}_{\SPOG}(#1)}
\newcommand{\TestSet}[1]{F_{#1}}
\newcommand{\NAIdvTests}{\TestSet{1}}
\newcommand{\NAGroupTests}{F_{2,grp}}
\newcommand{\NAAddIdvTests}{F_{2,idv}}

We now introduce our non-adaptive test design as well as the corresponding estimation algorithm and prove that it satisfies \cref{thm:naalg}.
Conceptually, the algorithm performs classification in two stages---however, all the tests are determined in advance, so it is indeed a non-adaptive algorithm.

In the first stage, $\ceil{\eta \log(n)}$ individual tests per individual (for a $\eta > 0$ we determine later) are used for pre-classification:
    $\patsSPG$ classifies an individual as infected if at least a $C$ proportion of these individual tests display positively,
    where $C$ is the threshold defined in \cref{sec:results}.
This then acts as a synthetic ``pseudo-genie'' in the second stage:
    it essentially behaves similarly to a genie estimator, with the always-correct genie being replaced by the mostly-correct ``pseudo-genie'' created in the first stage.

The second stage's tests are almost entirely group tests all having the same size $\Gamma$,
    with some extra individual tests to guarantee that each individual participates in sufficiently many tests.
The classification is not based on \emph{all} second-stage tests containing a given individual, but rather just a subset:
First, the algorithm constructs a set of tests $D_i$
    such that any two tests in $D_i$ only intersect in $i$.
We call such sets \emph{distinctive}.
These sets have the useful property their observed results are independent when conditioning on $\pats(i)$.
Then, $\SPOG$ only considers tests in $D_i$ which are good under the pseudo-genie $\patsSPG$;
    this is the set of \emph{pseudo-good} tests~$P_i$.

To formally describe the design $\vG_{\SPOG}$,
    let $\xi(\alpha, \vp, \Gamma) = \bc{-\log\bc{1-(1-\alpha)^{\Gamma - 1} \cdot (1-\eul^{-\beta(\vp)})}}^{-1}$.
This ensures that an individual participating in $\xi \log(n)$ tests of size $\Gamma$
    is misclassified by a genie estimator probability in $\Oh(n^{-1})$.

\begin{definition}[Test design for \SPOG]\label{def:g_spog}
    For $\hat{\alpha} \in (0,1)$, $\Gamma \in \NN^+$ with $\Gamma \leq \ceil{\hat{\alpha}^{-1}}$, a valid noisy channel $\vp$, and $\eta, \varepsilon > 0$,
        we let $\vG_{\SPOG}(n, \hat{\alpha}, \vp, \Gamma, \eta, \varepsilon)$ be the test design having $n$ individuals and with its test set $F$ being the union of\begin{itemize}
        \item the set $\NAIdvTests$ of $\ceil{\eta \log n}$ individual tests $F_1(i)$ per individual $i \in [n]$,
        \item the set $\NAGroupTests$ of $\ceil{(1+\varepsilon/3)\xi(\alpha, \vp, \Gamma) \cdot n \log n/\Gamma}$ group tests chosen i.u.r.\ among tests of size $\Gamma$, and
        \item the set $\NAAddIdvTests$ of $\max\cbc{0, \bc{1+\frac{\varepsilon}{6}} \xi \log n + 1 - \abs{\NAGroupTests \cap \partial i}}$ individual tests for each $i \in [n]$.
    \end{itemize}
    Furthermore, we write $\TestSet{2} = \NAGroupTests \cup \NAAddIdvTests$.
\end{definition}

\hspace{-7mm}
\begin{minipage}{0.55\linewidth}
\begin{algorithm}[H]
 	\KwData{Instance $G$ of $\vG_{\SPOG}$, $\otests_G \in \{0, 1\}^n$}
 	% \KwResult{an estimate $\patsNA$ of $\pats$}

    \For{$i \in [n]$}{
        $\patSPG{i} \gets \vecone\brk{\abs{\NAIdvTests^+(i)} \geq C \cdot \abs{\NAIdvTests(i)}}$
    }

    \For{$i \in [n]$}{
        $S \gets \emptyset$, $D_i \gets \emptyset$
        
        \For{$a \in \TestSet{2} \cap \partial i$}{
            \If{$\partial a \setminus \cbc{i} \cap S \neq \emptyset$}{
                $D_i \gets D_i \cup \{a\}$,
                $S \gets S \cup (\partial a \setminus \{i\})$
            }
        }

        $P_i \gets \{a \in D_i \mid \forall j \in \partial a \setminus \cbc{i}: \patSPG{i} = 0\}$

        $\patNA{i} \gets \vecone\brk{\abs{a \in P_i \cap \partial i \mid \otests(a) = 1} \geq C \cdot \abs{P_i}}$
    }

    \Return $\patsNA$

    \caption{$\SPOG$}
    \label{algo:nonad}
\end{algorithm}    
\end{minipage}
\hspace{4mm}
\begin{minipage}{0.4\linewidth}
\vspace{-5mm}
    \begin{figure}[H]
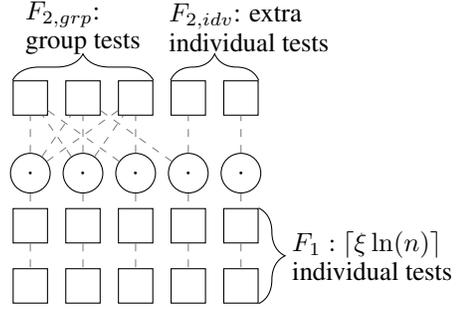

        \centering        
        \napic
        \caption{Illustration of $\vG_{\SPOG}$. Circles represent individuals and squares tests.}
        % \label{fig:gna}
    \end{figure}
\end{minipage}

\medskip

Our analysis makes frequent use of the following bounds on the value of $\xi$ proven in \cref{apx:nonad:xi_calc}.
Part \emph{(b)} justifies the restriction $\Gamma \leq \ceil{\alpha^{-1}}$ in the test design,
    as it implies that the test degree minimizing the total number of tests satisfies said restriction.

\begin{restatable}{lemma}{xicalc}
    \label{obs:naalg_xi}
    Let $\alpha \in (0, 1)$, and $\vp$ a valid noisy channel.
    \begin{enumerate}[label=(\alph*)]
    \item For all $\Gamma \leq \ceil{\alpha^{-1}}$, $\xi(\alpha, \vp, \Gamma) \in [\beta^{-1}, \frac{\eul}{(1-\alpha)(1-\eul^{-\beta})}]$,
        and $\xi(\alpha, \vp, \Gamma)$ is increasing in $\alpha$.
    \item Let $\hat{\Gamma} \in \argmax_{\Gamma \in \NN^+} -\Gamma \log\bc{1 - (1-\alpha)^{\Gamma - 1} \bc{1 - \eul^{-\beta(\vp)}}}$.
    Then $\hat{\Gamma} \leq \ceil{\alpha^{-1}}$.
    \end{enumerate}
\end{restatable}

The following lemma bounds the number of tests in $G_{\SPOG}$.
The proof in \cref{apx:lem:naalg_test_count} is a standard calculation;
    its crux is that the probability of adding extra individual tests (i.e., tests in $\NAAddIdvTests$) for a given $i$ is polynomially small, so that even if $\Theta(\log(n))$ tests were added for any such individual, $\abs{\NAAddIdvTests}$ is negligible.

\begin{restatable}{lemma}{lemnaalgtestcount}\label{lem:naalg_test_count}
    $\vG_{\SPOG}(\hat{\alpha}, \vp, \Gamma, \eta, \eps)$ uses at most $\bc{\eta + \bc{1+\frac{\varepsilon}{3}}\frac{\xi(\alpha, \vp, \Gamma)}{\Gamma}}n\log n + \Oh(n)$ tests \whp.
\end{restatable}

We now turn to the analysis of $\SPOG$'s correctness.
First, we see that the classification of the pseudo-genie $\patsSPG$ approximates $\pats$ well, in the sense that for each fixed individual $i$,
    it correctly identifies $i$'s status (although of course, this does not apply for all individuals at the same time).
The following lemma, proven in \cref{apx:lem:naalg_peudogenie}, collects this fact as well as some other useful properties of $\patsSPG$.
We note that the value $\beta$ stems from the choice of $C$ as threshold.

\begin{restatable}{lemma}{naalgpseudogenie}\label{lem:naalg_pseudogenie}
    For all $i \in [n]$: $\Prob{\patSPG{i} \neq \pats(i)} \leq n^{-\beta \eta}$, and these events are independent.
    Furthermore, the upper bound on the probability also holds when conditioning on $\pats(i) = 1$ or $\pats(i) = 0$.
\end{restatable}

Due to the restriction $\Gamma \leq \ceil{\hat{\alpha}^{-1}}$,
    it is just as unlikely that a pseudo-good test contains \emph{any} individual misclassified by $\patsSPG$,
        i.e., that a pseudo-good test is \emph{not} good.
The proof of this bound in \cref{apx:lem:naalg_test_tainted} additionally uses \Cref{lem:naalg_pseudogenie} as well as Bayes' theorem.

\begin{restatable}{lemma}{lemnaalgtesttainted}\label{lem:naalg_test_tainted}
    For any $i \in [n]$, the probability of a test $a \in P_i$ containing a misclassified individual (other than $i$) is in $\Oh(n^{-\beta \eta})$.
    To be precise: for all $i \in n$ and $a \in \partial i$,
        \[\ProbCond{\bigvee_{j \in \partial a \setminus \cbc{i}} \patSPG{j} \neq \pats(j) }{ a \in P_i} = \Oh(n^{-\beta\eta}).\]
\end{restatable}

This entails that the probability of an individual $i$ being misclassified by $\SPOG$
    is close to the probability of it being misclassified by the genie estimator if all pseudo-good tests were good.
The proof of this lemma, found in \cref{apx:lem:naalg:misclassify_probcond_abs_Pi},
    additionally uses the independence of observed test results $P_i$ conditioned on $\sigma(i)$ (as they are distinctive by construction).

\begin{restatable}{lemma}{lemnaalgmisclassifyprobcondabsPi}\label{lem:naalg:misclassify_probcond_abs_Pi}
    For any individual $i \in [n]$, \(\ProbCond{\patNA{i} \neq \pats(i)}{\abs{P_i}} \leq \exp\bc{-\abs{P_i}\bc{\beta - \Oh(n^{-\beta\eta})}}.\)
\end{restatable}
    
To actually use this bound,
    we require a lower bound on the number of distinctive tests for each individual.
The following bound is good enough for our purposes;
    its proof in \cref{apx:lem:naalg_many_distinctive_tests} is rather similar to birthday-paradox-type calculations.

\begin{restatable}{lemma}{lemnaalgmanydistinctivetests}
\label{lem:naalg_many_distinctive_tests}
    Let $\cD$ be the event that for all $i \in [n]$, we have $\abs{D_i} \geq \bc{1 + \frac \eps 6}\xi \log n$.
    Then
    \(\Prob{\cD} = 1 - \Oh\bc{\frac{\Gamma^4 \log^4 n}{n}}.\)
\end{restatable}

The final lemma bounds the probability of misclassifying an individual (given that there are sufficiently many distinctive tests).
In its proof (in \cref{apx:lem:naalg_misclassify}),
    we see that the probability of an individual $i$ being misclassified
    is bounded by $\bc{1 - (1-\alpha)^{\Gamma - 1}(1-e^{-\beta}) - \Oh(n^{-\delta})}^{\abs{D_i}}$ given $\abs{D_i}$.
And $\xi$ is defined to obtain the upper bound of the lemma as long as $\abs{D_i} \geq (1+\frac \eps 6)\xi \log n$.

\begin{restatable}{lemma}{lemnaalgmisclassify}\label{lem:naalg_misclassify}
    Assume that $\Gamma = \Oh(n^{\beta \eta - \delta})$ for some $0 < \delta \leq \beta\eta$,
        and that $\alpha \leq \hat{\alpha}$.
    Then for $\mathcal{D}$ as in \Cref{lem:naalg_many_distinctive_tests},
        we have, for all $i \in [n]$,
        \[\ProbCond{\patNA{i} \neq \pats(i) }{ \mathcal{D}} = n^{-\bc{1+\frac{\varepsilon}{6}}+\Oh(n^{-\delta})}.\]
\end{restatable}

The proof of \cref{thm:naalg} now boils down to choosing appropriate parameters and combining the lemmas.

\begin{proof}[Proof of \cref{thm:naalg}]
    We run $\SPOG$ with on $G \sim \vG_{\SPOG}(n, \alpha, \vp, \Gamma, \eta, \varepsilon)$,
        where $\Gamma \in \NN^+$ is a maximizer of $- \Gamma\log\bc{1 - (1-\alpha)^{\Gamma - 1}\cdot \bc{1 - \eul^{-\beta}}}$ (and this $\Gamma$ is at most $\ceil{\alpha^{-1}}$ by \cref{obs:naalg_xi} (b), as required),
        and $\eta = \frac{\varepsilon}{3} \cdot \frac{\xi(\hat{\alpha}, \beta, \Gamma)}{\Gamma}$.
    This graph contains at most $\bc{1 + \frac{2\varepsilon}{3}} \cdot \frac{\xi}{\Gamma} \cdot n \log n + \Oh(n)$ tests \whp by \cref{lem:naalg_test_count},
        which is at most $\bc{1 + \varepsilon} \mna$ for sufficiently large $n$
        (as $\mna = \frac{\xi}{\Gamma} n\log n = \Theta(n \log n)$).

    Now since $\xi(\alpha, \vp, \Gamma) = \Theta(1)$ and $\Gamma = \Theta(1)$ (by \Cref{obs:naalg_xi} and the fact that $\alpha = \Theta(1)$),
        the event $\mathcal{D}$ of \cref{lem:naalg_many_distinctive_tests} occurs with probability $1 - \Oh(\log^4 n / n)$,
        which is at least $1 - n^{-\varepsilon/6}$ for sufficiently large $n$.
    Combining this with \cref{lem:naalg_misclassify} (where $\delta = \beta\eta$) and the union bound over all $i \in [n]$,
        the probability of there being \emph{any} individual which is misclassified is in $\Oh(n^{-\frac{\varepsilon}{6} + \Oh(n^{-\beta\eta})}) = \Oh(n^{-\frac{\varepsilon}{6}})$,
        which is at most $\varepsilon$ for sufficiently large $n$, as claimed.
\end{proof}

We also show that for suitable parameters, $\SPOG$ is correct for sub-constant $\alpha$:
The adaptive algorithm $\PRESTO$ requires a non-adaptive algorithm in this regime that is still correct even if we only know an upper bound on the true value of $\alpha$, of which we are otherwise aware.
As an added benefit, our results are self-contained.
We defer the proof to \cref{apx:prop:sublin}.

\begin{restatable}{proposition}{propsublin}\label{prop:sublin}
    For any valid noisy channel $\vp$, and any $\hat{\theta} \in (0, 1)$ and $\varepsilon > 0$,
            there is a randomized test design $G$ using $m \leq \eps n \log n$ tests \whp so that as long as $\alpha \leq n^{-\hat{\theta}}$,
    \begin{align*}
        \SmallProb{\SPOG(G, \otests_G) = \pats} \geq 1 - \eps.
    \end{align*}
\end{restatable}

%\newpage

\section{Adaptive Group Testing}
In the following, we give an overview of the proof of impossibility when less than $\adm$ tests are used in the adaptive case,
    as well as the efficient algorithm \SPOG\ that correctly identifies $\pats$ \whp using no more than $(1+\eps)\adm$ tests.  

\subsection{Impossibility: proof strategy for \Cref{thm:adp:low}}
\label{sec:adaptive_impossible}

\newcommand{\veppp}{\vepp'}
\newcommand{\TypicalConfigs}{\cC}
The high-level proof strategy is as follows:
As in the proof of \cref{thm:nalower},
    we first modify the test scheme by adding a few individual tests to ease analysis, which can only increase the probability of successful estimation (\cref{obs:adp:reduction}).
For a ``typical'' infection vector $\pat$ (see \cref{def:typical_config}),
    we see that its posterior probability mass is dwarfed by that of a set of vectors that only differ in one coordinate from $\pat$.
This gives an upper bound on the posterior probability of $\pat$ (\cref{lem:hd:ubound-posterior-typical}).
Further, the ground truth $\pats$ is indeed ``typical'' \whp for our modified test design (\cref{lemma:prob_J}).
Combining all of the above then yields \Cref{thm:adp:low}.

\medskip

First, we formalize the notion of an adaptive test scheme. 
\begin{definition}[Adaptive test scheme]\label{def:ada-alg}
    An adaptive test scheme $\alg=(\alg_i)_{i\in[m]}$ is a sequence of sets forming a test design $G_\alg = (V_\alg, F_\alg, E_\alg)$ with $F_\alg=\cbc{a_1,\ldots,a_m}$ such that for each $i\in[m]$, $\alg_i=\alg_{a_i}=\partial_\alg a_i$ is a (possibly random) function on the previous tests and their observed results. Formally, 
    \( \alg_i=\alg_i((\alg_j,\otest_j)_{j\in[i-1]})\subseteq [n],\)
    with $\otest_j$ as the observed test result of $\alg_j$. 
    % We denote the set of all adaptive algorithms using $m$ tests as $\soA_m$.
\end{definition}
Consider any adaptive test scheme $\alg$ and fix $\eps > 0$.
% We write $G_{\alg}$ for the test design (graph) induced by $\alg$ (where for a test $a \in [m]$, $\partial a = \alg_a((\alg_b, \otests_b)_{b < a})$), and $\otests_\alg$ for the observed test results for $\alg$'s tests.
% We modify $\alg$ obtaining $\alg'$ as follows:
Let $\eta = \frac{\eps}{2} \cdot \frac{\adm}{n\log n} = \frac{\eps}{2} \cdot \frac{\alpha}{\KL{p_{11}}{p_{01}}}$.
% Then $\alg'$ is exactly like $\alg$,
%     except that after all the tests of $\alg$,
We obtain a modified test scheme $\alg'$ from $\alg$ by adding $\lfloor \eta \ln(n)\rfloor$ individual tests for each individual.% $\alg_j=\cbc{i}$ $(m+(i-1)\lfloor \eta \ln(n)\rfloor+1\leq j\leq N+i\lfloor \eta \ln(n)\rfloor)$ to the algorithm $\alg$ for each $i\in [n]$, write $\alg'$ for this modified algorithm,
%As we did in the proof of the lower bound for the non-adaptive case,
    % in order to guarantee a lower bound on the number of the good tests,
        % and $\otests_{\alg'}$ for its observed test results.
% \todo{Write $G_{\alg}$ for the graph induced by $\alg$?}
Similar to the proof of \cref{lem:modifiedsetup}, the best possible estimator on $\alg'$ is superior to any estimator using $\alg$:
\begin{observation}\label{obs:adp:reduction}
    For any valid noisy channel $\vp$, $\alpha > 0$, any adaptive test scheme $\alg$ and estimator $f$,
        \[\Prob{\pats = f(G_\alg, \otests_{\alg})} \leq \max_{f'} \Prob{\pats = f'(G_{\alg'}, \otests_{\alg'})}.\]
\end{observation}

Recall that $\infected_{\pat}$ is the set of infected individuals under $\pat$, and that $g_i^\pm$ is the number of good tests (\Cref{def:goodtest}).
Given this, a \emph{typical} infection vector is defined as follows:
\begin{definition}[typical infection vector]\label{def:typical_config}
    Let $\vepp>0$.
    For any infection vector $\pat \in \{0, 1\}^n$, test design $G$, and displayed test results $\otest$, define the set of infected individuals with a $\vepp$-typical ratio of good tests displaying positively as
    \begin{align*}
    \infectedprime^{\vepp}_{\pat}
        = \infectedprime^{\vepp}_{\pat,G,\otest}
        = \cbc{i\in \infected_{\pat} \mid \,\abs{\good_i^+(\pat)-p_{11}\good_i(\pat)} \leq \vepp \good_i(\pat)}.
    \end{align*}
    An infection vector $\sigma$ is \emph{$\vepp$-typical} for $G$ and $\otest$ if $|\infectedprime^{\vepp}_{\pat,G,\otest} - \alpha n| \leq \vepp n$;
        Let $\TypicalConfigs^{\vepp}(G, \otest)$ be the set of these~$\sigma$.
\end{definition}

Note that the number of infected individuals is concentrated on $\alpha n$.
When an infected individual $i \in \infected_{\pats}$ has a sufficient number of good tests---which is given in $\alg'$---then since $g_i^+$ is concentrated, $i \in \infectedprime^{\vepp}_{\pats}$ \whp.
Indeed, $\pats$ is typical \whp for $\alg'$ and $\otests_{\alg'}$ as stated in the following lemma.
We prove this in \Cref{sec_lemma:prob_J}.
\begin{restatable}{lemma}{probJ}\label{lemma:prob_J}
For any $\vepp > 0$, $\pats$ is $\vepp$-typical for $G_{\alg'}$ and $\otests_{\alg'}$ \whp, i.e.,
    \(\SmallProb{\pats \not\in \TypicalConfigs^{\vepp}(G_{\alg'}, \otests_{\alg'})} = o(1).\)
\end{restatable}

We now turn to bounding the posterior probability of typical infection vectors $\pat$.
To that end, for any infection vector $\pat \in \{0,1\}^n$ and any $i \in \infected_{\pat}$,
    write $\pat^{\downarrow i}$ for the infection vector obtained from $\pat$ by setting the $i$th coordinate to $0$.
Given any test design and observed test results,
    we compare the posterior odds ratio of 
    $\pat$ and the vectors $\pat^{\downarrow j}$
        for individuals $j \in \infectedprime^{\vepp}_{\pat,G,\otest}$ in \cref{apx:lem:hd:lbound-nei-vec}, leading to the following lower bound:

\begin{restatable}{lemma}{lboundneivec}\label{lem:hd:lbound-nei-vec}
For any $\pat \in \{0, 1\}^n$, test design $G$, and observed test vector $\otest$, it holds for all $j \in \infectedprime^{\vepp}_{\pat, G, \otest}$ that
\begin{align}\label{eq:hd:lbound-nei-vec}
\frac{\ProbCond{\pats=\sigma^{\downarrow j}}{G_\alg=G,\otests_\alg=\otest}}{\ProbCond{\pats=\pat}{G_\alg=G,\otests_\alg=\otest}}
    \geq \frac{1-\alpha}{\alpha}\exp\bc{-\good_j(\pat, G) \bc{\KL{p_{11}}{p_{01}}+\vepp \bc{\ln\frac{p_{00}}{p_{10}}+\ln\frac{p_{11}}{p_{01}}}}}.
\end{align}
\end{restatable}

Using this, as well as the fact that typical infection vectors $\pat$ have many such individuals $i$,
    we bound the posterior probability for a typical infection vector $\pat$ in \cref{sec:hd:ubound-posterior-typical}.
There, we exploit the fact that the sets of good tests for infected individuals are disjoint,
    and hence the total number of such tests is bounded by $m_{\alg'}$.

\begin{restatable}{lemma}{upt}\label{lem:hd:ubound-posterior-typical}
    For all sufficiently small $\vepp > 0$,
        any test design $G$ and observed test results $\otest$,
        and any $\vepp$-typical infection vector $\pat \in \TypicalConfigs^{\vepp}(G, \otest)$,
    \(\SmallProbCond{\pats=\pat}{G_{\alg'}=G, \otests_{\alg'}=\otest}
        = o(1).\)
\end{restatable}

With all this in place, we can prove the main theorem of this section.

\begin{proof}[Proof of \Cref{thm:adp:low}]
For any $f$, any $\vepp > 0$,
    let $\bm f = f_{G_{\alg'}}(\otests_{\alg'})$
        and $\bm \TypicalConfigs^{\eps'} = \TypicalConfigs^{\eps'}(G_{\alg'}, \otests_{\alg'})$. Then
    \begin{align*}
    \Prob{\pats = \bm f}
        = \Prob{\pats = \bm f \wedge \bm f \in \bm \TypicalConfigs^{\vepp}}
        + \Prob{\pats = \bm f \wedge \bm f \not\in \bm \TypicalConfigs^{\vepp}}
        \leq 
        \ProbCond{\pats = \bm f}{\bm f \in \bm \TypicalConfigs^{\vepp}}
           + \Prob{\pats \not\in \bm \TypicalConfigs^{\vepp}}.
    \end{align*}
Since $\SmallProb{\pats \not\in \bm \TypicalConfigs^{\vepp}} = o(1)$ for any $\vepp > 0$ by \cref{lemma:prob_J}, 
% (which uses the modifications to the test design).
 the claim follows by \Cref{obs:adp:reduction} 
 if the first term is in $o(1)$ as well.
To that end, we see that for a sufficiently small $\vepp > 0$,
    by \Cref{lem:hd:ubound-posterior-typical} and the law of total probability:
\begin{align*}
\ProbCond{\pats = \bm f}{\bm f \in \bm \TypicalConfigs^{\vepp}}
   &= \hspace{-0.45cm}\sum_{\substack{G,\otest,\\ \pat\in\TypicalConfigs^{\vepp}(G, \otest)}} \hspace{-0.45cm}\ProbCond{\vsigma=\pat}{\pat\in \bm \TypicalConfigs^{\vepp}, \bm f = \pat, G_{\alg'}=G,\otests_{\alg'}=\otest} \cdot \ProbCond{\bm f = \pat,G_{\alg'}=G,\otests_{\alg'}=\otest}{\bm f \in \bm \TypicalConfigs^{\vepp}}
\\ &= o(1) \cdot \sum_{\pat}\ProbCond{\bm f = \pat}{\bm f \in \bm \TypicalConfigs^{\vepp}}
    = o(1). \quad \qedhere
\end{align*}
\end{proof}

\subsection{The algorithm $\PRESTO$: proof strategy for \Cref{thm:adp:upper}}
\label{sec:adaptive_algo}
% 
% \lrh{I think there is a lot of redundancy in the text now: first there is the introductory paragraph which explains the stages briefly, then the longer descriptions of the stages below, then the formal algorithm, and then the graphic. }
% 
% \toapx{
% From the theorem: $\eps > 0$,\\
% Threshold for first stage: $z$ with $p_{01} < z < p_{11}$, \\
% Testnumber parameter: $\eta >0$ s.t.$ \eta\bc{\KL{1-z}{p_{10}} }<1$\\
% Some other error parameter $\eps_1 >0$\\
% Eps for sublinear algorithm $\eps_2>0$\\
% Mismatched theta eps: $0<\hat{\eps}_1< \eta\bc{\KL{1-z}{p_{10}} }$\\
% $0 < \hat{\eps}_2 <  (1+\frac\eps2) \frac{\KL{p_{01} +\eps_1  }{p_{10}}}{\KL{p_{10} }{p_{00}}}$
% \\
% Define mismatched density parameters:
%     \begin{align}
%         \hat{\theta}_1= \hat{\theta}_1 (\hat{\eps}_1) =&  1 -\eta\bc{\KL{1-z}{p_{10}} } + \hat{\eps}_1  \label{eq:hattheta1}\\
%         \hat{\theta}_2 = \hat{\theta}_2 (\hat{\eps}_2) =& \max(\eps,  1-  (1+\frac\eps2) \frac{\KL{p_{01} +\eps_1  }{p_{10}}}{\KL{p_{10} }{p_{00}}} + \hat{\eps}_2)
%         \label{eq:hattheta2}
%     \end{align}

% \begin{minipage}{0.5\linewidth}
% \begin{wrapfigure}{r}{0.45\textwidth}
%     \vspace{-10mm}
%     \includegraphics[width = 0.48\textwidth]{Journal/pics/adalg.png}
%     \caption{An illustration of the adaptive test scheme of \PRESTO. The circles represent the individuals and the squares represent tests.
%     \label{fig:adtests}
% \end{wrapfigure}
\newcommand{\epsuthresh}{\delta_1}
\newcommand{\epssizes}{\delta_2}
\newcommand{\epsnaalg}{\delta_{\SPOG}}
\newcommand{\epsnaalgproof}{\delta'}
\hspace{-8mm}
\begin{minipage}{0.59\linewidth}
The adaptive algorithm runs three stages of tests.
In the first stage, each individual is tested $\ceil{\eta \log(n)}$ times individually (for a small $\eta$).
The population is then partitioned via thresholding on the fraction of positive tests:
    the set $S_1$ contains almost only infected individuals,
        while the other, $S_0$, contains almost only uninfected individuals.
In the second stage,
    all individuals in $S_1$ (of which there are $\approx \alpha n$) undergo further $\ceil{\bc{1+\frac \eps 4} \adm / (\alpha n)}$ individual tests.
Again, $S_1$ is partitioned via thresholding:
    the threshold is $p_{11} - \epsuthresh$ for a small $\epsuthresh$
    so that one set, $U_1$, contains \emph{only} infected individuals \whp,
        while the other, $U_0$, again contains almost only healthy individuals.
The third stage cleans up $S_0$ and $U_0$ by performing \SPOG\ separately on both sets, using the variant for sub-constant $\alpha$.
We keep the sets separate to maintain an i.i.d.\ prior. 

The complete algorithm is given formally in \cref{algo:adap}.
The parameters $\eta > 0$ and $\epsuthresh \in (0, p_{11} - p_{10})$ are chosen so that
\begin{align}
    \bc{1+\frac\eps4} \frac{\KL{p_{11} - \epsuthresh }{p_{01}}}{\KL{p_{11} }{p_{01}}} &> 1 \label{eq:eps2_bound2},
\end{align}
    and
\begin{align}
    \eta\KL{C}{p_{01}} < \min\cbc{1, \bc{1+\frac{\eps}{4}}\frac{\KL{p_{11}-\epsuthresh}{p_{11}}}{\KL{p_{11}}{p_{01}}}} \label{eq:eps2_nu}.
\end{align}
% \end{minipage}  
% \begin{minipage}{0.45\linewidth}
%     \begin{figure}[H]
%         \centering
%         \includegraphics[width=1\linewidth]{Journal/pics/adalg.png}
%         \caption{Caption}
%         \label{fig:enter-label}
%     \end{figure}
% \end{minipage}
\end{minipage}
\hspace{0mm}
\begin{minipage}{0.39\linewidth}
    \begin{figure}[H]
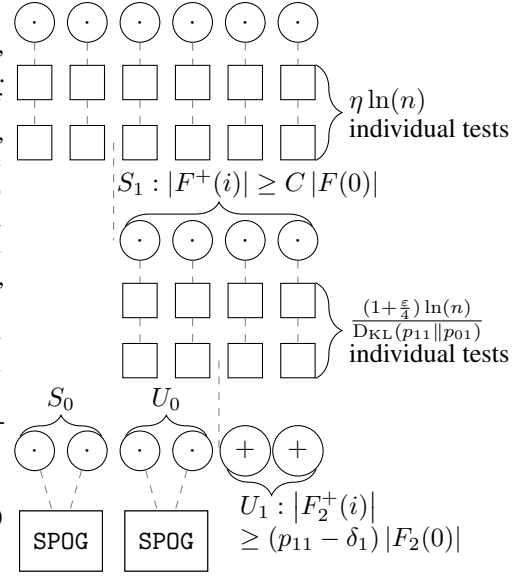

        \vspace{-8mm}
        \adpic
        \caption{An illustration of the adaptive test scheme of \PRESTO. The circles represent the individuals and the squares represent tests.}
        \label{fig:adtests}
    \end{figure}
\end{minipage}

\begin{algorithm}[H]
    \SetKwInOut{KwParameters}{Parameters}
 	\KwParameters{$\eps>0$, $\epsuthresh \in (0, p_{11} - p_{01}), \eta > 0$ satisfying \cref{eq:eps2_bound2,eq:eps2_nu}.}
 	%\KwResult{an estimate of $\pats$}
        \label{alg:ind1}
            Conduct $\ceil{\eta \log(n)}$ individual tests $F_1(i)$ for each $i \in [n]$ \label{alg:stg1}; let $F^+_1(i)$ be the positive ones.

        $S_1 \gets \bc{i \in [n] \mid \abs{F_1^+(i)} > C\abs{F_1(i)}}$, $ \qquad S_0 \gets [n] \setminus S_1$
        
        \label{alg:ind2}
        Conduct $\ceil{\bc{1+\frac \eps 4} \frac{ \log(n) }{ \KL {p_{11}} {p_{01}}  }}$ individual tests $F_2(i)$ for each $i \in S_1$\label{alg:testS1};
        let $F^+_2(i)$ be the positive ones.

        $U_1 \gets \bc{i \in S_1 \mid \abs{F_2^+(i)} > (p_{11}-\epsuthresh) \abs{F_2(i)}}$, $ \qquad U_0 \gets S_1 \setminus U_1$
        
        $\estSPEX^S \gets $ result of $\SPOG$ run on $S_0$ as described in \cref{prop:sublin}, with $\hat{\theta} = \hat{\theta}^S = \frac{\eta}{2}\KL{C}{p_{11}}$.\label{alg:NAS} \\
        $\estSPEX^U \gets $ result of $\SPOG$ sun on $U_0$ as described in \cref{prop:sublin}, with $\hat{\theta} = \hat{\theta}^U = \frac{1}{2}\bc{\bc{1+\frac \eps 4}\frac{\KL{p_{11}-\epsuthresh}{p_{11}}}{\KL{p_{11}}{p_{01}}}-\eta \KL{C}{p_{01}}}$.\label{alg:NAU}\\
        \Return $\hat{\pats} = \bc{\vecone\brk{i \in U_1 \cup \infected_{\estSPEX^S} \cup \infected_{\estSPEX^U}}}_{i \in [n]}$\label{alg:ret}
        \caption{\PRESTO}
        \label{algo:adap}
\end{algorithm}
To prove the correctness of $\PRESTO$, we require the following intermediate results:
First, we need to ensure that the tests used do not exceed $(1+\eps)\adm$ \whp. 
The main share of tests is used on individual tests for $S_1$ in \cref{alg:testS1} of \cref{algo:adap},
    so we require an upper bound on $\abs{S_1}$.
This is given by the following lemma, proven in \cref{sec:prf:lem_sizeS}.
\begin{restatable}[Size of $S_1, S_0$]
{lemma}{lemsizeS} \label{lem_sizeS}
    Let $\epssizes > 0$.
    Then \whp,
        $\abs{\abs{S_1} - \alpha n } \leq \alpha n \epssizes$
        and
        $\abs{\abs{S_0} - (1-\alpha) n } \leq \alpha n \epssizes$
\end{restatable}
% \paragraph{Stage Two: Fine Individual Testing}
% The second stage conducts $\ceil{(1+\frac \eps 4)  \log(n)/ \KL {p_{11}} {p_{01}}}$ further individual tests for each individual in $S_1$.
% Similar to the first stage, $S_1$ is partitioned into two sets $U_0$ and $U_1$ by thresholding:
%     $U_0$ will contain individuals with a subconstant probability of being infected in addition to uninfected individuals,
%     and $U_1$ will only contain infected individuals \whp.
% This time, the threshold is $p_{11} - \eps_2$ with $\eps_2$ satisfying \cref{eq:eps2_bound2}.
% \begin{algorithm}[H]
%  	% \KwData{$\pats, S_0, S_1$ of First stage}
%  	% \KwResult{an estimate of $\pats$}
%         % \If{$\abs{S_1}>(1+\frac{\eps}{\eps+4}) \alpha n$}{
%         %     \Return\\
%         % }
%         \For{$i\in S_1$}{\label{alg:ind2}
%             Conduct $\ceil{\bc{1+\frac \eps 4} \frac{ \log(n) }{ \KL {p_{11}} {p_{01}}  }}$ individual tests $F_2(i)$ for $i$\\
%             Let $F^+_2(i)$ be the positive ones
%         }
%         $U_1 \gets \bc{i \in [n] \mid \abs{F_2^+(i)} > (p_{11}-\eps_2) \abs{F_2(i)}}$
        
%         $U_0 \gets [n] \setminus U_1$
%         % \Return $\estas$
%         \caption{Second stage of \PRESTO}
%         \label{algo:adap2}
% \end{algorithm}
Second, we need to show the following, which we prove in \cref{sec:prf:lem_falsPos}.
\begin{restatable}[No false positives in $U_1$]{lemma}{lemfalsPos}\label{lem_falsPos}
    With high probability, all individuals in $U_1$ are infected, i.e., $U_1 \subseteq \infected_{\pats}$.
\end{restatable}
% \paragraph{Stage Three: Tidying up}
% Stage three of \PRESTO\ finally runs a round of \SPOG\ each for $S_0$ and $U_0$ as seen in \cref{alg:NAS,alg:NAU} of \cref{algo:adap}.
% Finally in \cref{alg:ret} we combine the set of individuals $\SPOG$ classified as infected with $U_1$ to obtain the set of all individuals $\PRESTO$ classifies as infected.
% The following lemma, proven in \cref{sec:prf:lem:use_of_naalg_correc} provides that \SPOG\ indeed estimates $S_0$ and $U_0$ correct \whp.
% \begin{algorithm}[H]
%  	% \KwData{$\pats, U_0, U_1$ of Second stage}
%  	% \KwResult{an estimate of $\pats$}i
%         $\estSPEX^S \gets $ result of $\SPOG$ run on $S_0$ as described in \cref{prop:sublin}, with $\hat{\theta} = \hat{\theta}^S = \frac{\eta}{2}\KL{C}{p_{11}}$. \\
%         $\estSPEX^U \gets $ result of $\SPOG$ sun on $U_0$ as described in \cref{prop:sublin}, with $\hat{\theta} = \hat{\theta}^U = \frac{1}{2}\bc{\bc{1+\frac \eps 4}\frac{\KL{p_{11}-\eps_2}{p_{11}}}{\KL{p_{11}}{p_{01}}}-\eta (1 + \frac{1}{\eta \log n})\KL{C-\vepp}{p_{01}}}$. \\
%         \Return $\hat{\pats} = (\vecone\brk{i \in U_1 \cup \infected_{\estSPEX^S} \cup \infected_{\estSPEX^U}})_{i \in [n]}$
% \end{algorithm}
Third, the correctness of \SPOG\ in \cref{alg:NAS,alg:NAU} is left to prove.
The following lemma, proven in \cref{sec:prf:lem:use_of_naalg_correc} provides that \SPOG\ indeed estimates $S_0$ and $U_0$ correct \whp.
\begin{restatable}{lemma}{lemuseofnaalgcorrect}\label{lem:use_of_naalg_correct}
    For any $\epsnaalg > 0$,
    $\estSPEX^S$ (resp., $\estSPEX^U$) correctly identify all individuals in $S_0$ (resp., $U_0$) using at most $\epsnaalg n \log n$ tests, \whp.
\end{restatable}
Finally, we can plug everything together to prove \cref{thm:adp:upper}.
\begin{proof}[Proof of \cref{thm:adp:upper}]
    The sets $S_0$, $U_0$, and $U_1$ partition $[n]$,
        all individuals in the former two sets are correctly classified \whp by \cref{lem:use_of_naalg_correct},
        and all individuals in $U_1$ are correctly classified as infected \whp by \cref{lem_falsPos}.
    So all individuals are correctly classified \whp.

    As for the number of tests:
        the first stage uses $n \ceil{\eta \log(n)} \leq \eta n \log(n) + n$ tests,
        the second stage uses $\abs{S_1} \ceil{\bc{1+\frac \eps 4} \frac{\log n}{\KL{p_{11}}{p_{10}}}}$ tests,
        and the third stage uses $\epsnaalg n \log n$ tests \whp.
    Since $\abs{S_1} \leq (1 + \epssizes)\alpha n$ \whp by \cref{lem_sizeS},
        the theorem follows for sufficiently small $\eta$, $\epssizes$, $\epsnaalg$ and sufficiently large $n$.
    % Adding that $\abs{S_1} \leq (1+\frac{\eps}{\eps+4}) \alpha n $ \whp, choosing $\eta$ and $\eps_2$ small enough, for $n$ big enough this sums up to
    % \begin{align}
    %     n \eta \log(n) + 2 \eps_2 c n \log(n) + \bc{1+\frac \eps 2} \alpha n  \frac{ \log(n) }{ \KL {p_{10}} {p_{00}}  }
    %     \leq (1+\eps) \alpha n \frac{ \log(n) }{ \KL {p_{10}} {p_{00}}  }
    % \end{align}
\end{proof}

% \section{Conclusion}

\newpage

\appendix
\section{Table of Notation}
\begin{table}[H]
\begin{tabularx}{\linewidth}{l p{6.3cm} X}
Symbol    & Meaning & Definition or Domain \\ \toprule
$\alg$ & Adaptive test scheme & \cref{def:ada-alg}\\
$\alpha$ & Probability of an individual being infected  & $\alpha \in (0, 1)$ \\
$\beta$                     & KL-Divergence of $p_{11}$ and optimal threshold                        & $ \beta= \max_{c \in [\pflipoi, \pflipii]} \min\cbc{\KL{c}{\pflipoi}, \KL{c}{\pflipii}}$     \\
$C$                         & Optimal threshold                              & $C=\argmax_{c \in [\pflipoi, \pflipii]} \min\cbc{\KL{c}{\pflipoi}, \KL{c}{\pflipii}}$ \\
$\TypicalConfigs^{\vepp}(G, \otest)$ & Set of typical configurations & \cref{def:typical_config} \\
$\KL{p}{q}$                 & Relative entropy                                                     & $\KL{p}{q} = \sum_{x \in \cX} p(x) \log\bc{\frac{p(x)}{q(x)}}$ \\
$\partial_G a$        & Individuals participating in a test $a$ in $G$ &  $\partial_G a = \cbc{i \in V(G) \mid (i,a) \in E(G)}$\\
$\partial_G i$        & Tests containing  an individual $i$ in $G$                                  &  $\partial_G i = \cbc{a \in F(G) \mid (i,a) \in E(G)}$\\
$F(G)$                   & Set of all tests in test design $G$                                  & \\
$G$& Test design as bipartite graph& $G=(V, F, E)$ \\
$\modGraph$ & Modified test design in non-adaptive impossibility proof &  \cref{def:G_eta}\\
$\vG_\SPOG$                            & Proposed non-adaptive test design      &      \cref{def:g_spog}\\
$\good_i$                   & Good test for individual $i$                                         & $\good_i = \abs{\cbc{a \in \partial i \mid \bc{\partial a\setminus\cbc{i}} \cap \infected_\pat = \emptyset }}$  \\
$\good_i^+$                 & Good positive displayed test for individual $i$                      & $\good_i^+=\abs{\cbc{a\in \partial i \mid \otests(a)=1 \wedge 				\bc{\partial a\setminus\cbc{i} \cap \infected_\pat} = \emptyset}}$ \\
$\good_i^-$                 & Good negative displayed test for individual $i$                      & $\good_i^+= \abs{\cbc{a\in \partial i \mid \otests(a)=0 				\land \bc{\partial a\setminus\cbc{i}} \cap \infected_\pat= \emptyset}}$  \\
$\Gamma$                    & Optimal test degree for non-adaptive case                                 &\\
$\infected_\pat$            & Infected individuals for infection vector $\pat$                     & $\infected = \cbc{i \mid \pat_i = 1}$     \\
$\overline{\infected_\pat}$ & Uninfected individuals for infection vector $\pat$                 & $\overline{\infected_\pat} = [n]\setminus\infected_\pat$  \\
$\infectedprime^{\vepp}_\pat$ & Set of infected individuals in typical infection vectors & \cref{def:typical_config}\\
$m$ & Number of tests &      \\
$\adm$                      & Threshold adaptive group testing                                     & $\adm = \cad n \log(n) $               \\
$\nonadm$                   & Threshold non-adaptive group testing                                 & $\nonadm = \cna n \log(n) $            \\
$n$ & Number of individuals &\\  
$\nnew$ & Number of individuals in $\modGraph$ & $\nnew =  \abs{V(\modGraph)}$ \\
$p_{01}, p_{10}$                & Test flip probabilities                                              &\\
$\vp$                       & Noise vector                                                       & $\vp=(p_{00}, p_{01}, p_{10}, p_{11} )$ \\
\PRESTO                            & Proposed adaptive algorithm &        \cref{algo:adap}            \\
$S^+$                   & Subset of positively displayed tests of a test set $S$                                   & $S^+ = \cbc{a \in S \mid \otests(a) = 1}$ \\
$S^-$                   & Subset of negatively displayed tests of a test set $S$                                 &$S^- = \cbc{a \in S \mid \otests(a) = 0}$  \\
\SPOG                            & Proposed estimator for non-adaptive case                             &  \cref{algo:nonad}\\
$\pats$ & Ground truth (random variable), true infection states of individuals & $\pats = (\pats(i))_{i \in [n]} \in \{0,1\}^n$\\
$\pats_{-i}$ & Ground truth without $i$     &               \\
$\genie^{G,\otests}$          & Genie estimator on test setup $G$                                    & \cref{def:genie}                                   \\
$\map^{G,\otests}$            & MAP estimator on test setup $G$                                      & \Cref{def:map}                                 \\
$\estSPEX$                            & Estimation result of \SPOG                             &      \\
$\atests$                   & Actual, true test results before noise                               & $\atests = (\atests(a))_a \in [m] \in \{0,1\}^m$ \\
$\otests$                   & Observed test results, after noise                                   & $\otests = (\otests(a))_a \in [m] \in \{0,1\}^m$ \\
\bottomrule
\end{tabularx}
\end{table}

\section{A Basic Probabilistic Toolbox}

Here, we give some fundamental probabilistic tools in the form we use throughout.

\renewcommand{\thetheorem}{\Alph{section}.\arabic{theorem}}
\setcounter{theorem}{0}

\begin{theorem}[Chernoff bound, cf.\ eqn.\ (2.4) and proof of Thm.\ 2.1 in \cite{JLRRandomGraphs00}]\label{thm:chernoff_kl}
    Let $X \sim \bin{n, p}$. Then
    \begin{align}
        \Prob{X \geq q n} &\leq \exp\bc{-n \KL{q}{p}}\quad\textup{for $p < q < 1$,} \\
        \Prob{X \leq q n} &\leq \exp\bc{-n \KL{q}{p}}\quad\textup{for $0 < q < p$.}
    \end{align}
\end{theorem}

\begin{theorem}[Chernoff bound, cf.\ Theorem 2.1 in \cite{JLRRandomGraphs00}]\label{thm:chernoff_t} % If we need the version for $X$ being a sum of indep X_i with different probabilities, we can update this.
    Let $X \sim \bin{n, p}$, and $t \geq 0$. Then with $\mu = \Exp{X} = np$, we have:
    \begin{align}
        \Prob{X \geq \mu + t} &\leq \exp\bc{-\,\frac{t^2}{2(\mu+t/3)}} \\
        \Prob{X \leq \mu - t} &\leq \exp\bc{-\,\frac{t^2}{2\mu}}
    \end{align}
\end{theorem}

\setcounter{lemma}{2}

The following statement in terms of relative error is a immediate consequence for $t = \eps \Exp{X}$:
\begin{corollary}\label{thm:chernoff_eps}
    Let $X \sim \bin{n, p}$, and $\eps \geq 0$. Then with $\mu = \Exp{X} = np$, we have:
    \begin{align}
        \Prob{X \geq (1+\eps)\mu} &\leq \exp\bc{-\,\frac{\eps^2 \mu}{2(1+\eps/3)}} \\
        \Prob{X \leq (1-\eps)\mu} &\leq \exp\bc{-\,\frac{\eps^2 \mu}{2}}
    \end{align}
\end{corollary}

We also have the following two-sided bound:
\begin{corollary}[Corollary 2.3 in \cite{JLRRandomGraphs00}]\label{thm:chernoff_eps_twoside}
    Let $X \sim \bin{n, p}$, and $0 \leq \eps \leq 3/2$. Then with $\mu = \Exp{X} = np$, we have:
    \begin{align}
        \Prob{\abs{X - \mu} \geq \eps\mu} &\leq 2\exp\bc{-\,\frac{\eps^2 \mu}{3}}.
    \end{align}
\end{corollary}

\setcounter{theorem}{4}

\begin{theorem}[Cram\'er's Theorem for Binomial r.v.s, cf.\ Thm.\ 2.2.3 and Exercise 2.2.23 (b) in \cite{largedeviationstechniquesandapplications}]\label{thm:cramers}\ 
    \begin{enumerate}[label=(\alph*)]
        \item For any closed set $F \subseteq [0, 1]$,
        \[\limsup_{n \to \infty} \frac{1}{n} \log \Prob{\frac{1}{n} \bin{n, p} \in F } \leq -\inf_{q \in F} \KL{q}{p}.\]
        \item For any open set $G \subseteq [0, 1]$,
        \[\liminf_{n \to \infty} \frac{1}{n} \log \Prob{\frac{1}{n} \bin{n, p} \in G} \geq -\inf_{q \in G} \KL{q}{p}.\]
    \end{enumerate}
\end{theorem}

\begin{theorem}[FKG inequality, cf.\ Thm.\ 2.12 in \cite{JLRRandomGraphs00}]\label{thm:fkg}
    If the random variables $\bm X_1$ and $\bm X_2$ are two increasing or two decreasing functions of $\bm Y = (\bm Y_i)_{i \in n} \in \{0, 1\}^n$ where the $\bm Y_i$ are independent indicator random variables,
        then
        \[\Exp{\bm X_1 \bm X_2} \geq \Exp{\bm X_1}\Exp{\bm Y_1}.\]
\end{theorem}

\section{Properties of the Threshold $C$}

Here, we determine the exact value of the threshold $C$ we defined in \cref{sec:model}.

\begin{lemma} \label{lem:c-thr}
    For any valid noisy channel $\vp$, as a function of $c$, $ \min\cbc{\KL{c}{\pflipoi}, \KL{c}{\pflipii}}$ has a unique maximizer $C=\ln\bc{\frac{p_{00}}{p_{10}}}/\ln\bc{\frac{p_{11}p_{00}}{p_{01}p_{10}}}$ in $[\pflipoi, \pflipii]$. Further, $C\in (\pflipoi, \pflipii)$ and
     $\KL{C}{\pflipoi} = \KL{C}{\pflipii}\in (0,\infty)$.
\end{lemma}
\begin{proof}
Recall that 
$\KL{p}{q} = p\log\frac{p}{q} + (1-p)\log\frac{1-p}{1-q}$
is a continuous differentiable function of $p$ in $(0,1)$. Then,
\begin{align*}
    \frac{\partial \KL{p}{q} }{\partial p}=\ln\frac{p(1-q)}{q(1-p)}.
\end{align*}
Hence, given $q\in(0,1)$, $\KL{p}{q}$ is a strictly decreasing function in $[0,q]$ and a strictly decreasing function on $[q,1]$, with a minimum $0$ $0$ at $p=q$. Consequently, $\KL{c}{\pflipoi}$ is a strictly increasing continuous function in $[\pflipoi, \pflipii]$, with a minimum $0$ at $c=\pflipoi$; $\KL{c}{\pflipii}$ is a strictly decreasing continuous function in $[\pflipoi, \pflipii]$, with a minimum $0$ at $c=\pflipii$. Hence, the two function $\KL{c}{\pflipoi} $ and $ \KL{c}{\pflipii}$ intersect at precisely one point $C$ in $[\pflipoi, \pflipii]$ and $C\in (\pflipoi, \pflipii)$. Further, this $C$ is the unique maximizer of $\min\cbc{\KL{c}{\pflipoi}, \KL{c}{\pflipii}}$ in $[\pflipoi, \pflipii]$ by the monotonicity of $\KL{c}{\pflipoi} $ and $ \KL{c}{\pflipii}$ in this interval. Since $\KL{c}{\pflipoi} = \KL{c}{\pflipii}$, we compute that
\begin{align*}
   &&C\log\frac{C}{\pflipoi} + (1-C)\log\frac{1-C}{1-\pflipoi}&=C\log\frac{C}{\pflipii} + (1-C)\log\frac{1-C}{1-\pflipii},\\
   \Leftrightarrow&&(1-C)\ln\frac{\pflipoo}{\pflipio}=(1-C)\ln\frac{1-\pflipoi}{1-\pflipii}&=C\ln\frac{\pflipii}{\pflipoi},\\
   \Leftrightarrow&&C&=\frac{\ln\bc{\frac{p_{00}}{p_{10}}}}{\ln\bc{\frac{p_{11}p_{00}}{p_{01}p_{10}}}},
\end{align*}
as desired.
\end{proof}

\section{Full Proofs for Non-Adaptive Group Testing}

\setcounter{subsection}{-1}
\subsection{Proof of Lemma \ref{lem:geniebetter}: genie estimator bests MAP estimator} \label{apx:geniebetter}
\geniebetter*
\begin{proof}
Let $\map$ be the MAP estimator defined as in \Cref{def:map} and $\genie$ be the genie estimator defined as in \Cref{def:genie}. Since we only deal with test design $G$ and observed tests $\otests$ in the proof, we simplify $\map^{G,\otests}$ as $\map$ and $\genie^{G,\otests}(i)$ as $\genie$. We further denote $\map(-i)$ as the MAP estimator $\map$ restricting its domain on $[n]\backslash \cbc{i}$. 
    We first notice that $\map(i)=\genie^{G,\otests,\map(-i)}(i)$, where $\genie^{G,\otests,\map(-i)}(i)$ is the output of the genie estimator on the $i$th coordinate conditioning on $\pats_{-i}=\map(-i)$.
    
    Indeed, given $G,\otests$, assume that for some $i\in [n]$, $\map(i)\neq\genie(i)$ when conditioning on $\pats_{-i}=\map(-i)$. Then, either
    \begin{align}\label{eq:ass-map-i}
        \map(i)\neq\argmax_{s\in \{0,1\}} \ProbCond{\pats(i) = s }{ G,\otests,\pats_{-i}=\map(-i)},
    \end{align}
or both $0$ and $1$ maximize
\begin{align}\label{eq:ass-map-i-2}
    \ProbCond{\pats(i) = s }{ G,\otests,\pats_{-i}=\map(-i)}, 
\end{align}while $\genie^{G,\otests,\map(-i)}(i)=1$.
    
    If \eqref{eq:ass-map-i} holds, then \eqref{eq:genie} yields that
    \begin{align}\label{eq:cond--i-map-less-genie}
        \ProbCond{\pats(i) = \map(i) }{ G,\otests,\pats_{-i}=\map(-i)}
        <\ProbCond{\pats(i) = \genie(i) }{ G,\otests,\pats_{-i}=\map(-i)}.
    \end{align}
If \eqref{eq:ass-map-i-2} holds, then 
\begin{align}\label{eq:cond--i-map-less-genie-2}
        \ProbCond{\pats(i) = \map(i) }{ G,\otests,\pats_{-i}=\map(-i)}
        =\ProbCond{\pats(i) = \genie(i) }{ G,\otests,\pats_{-i}=\map(-i)},
    \end{align}
and $\genie^{G,\otests,\map(-i)}(i)=0$.

    Let 
    \begin{align*}
        \map^{\ast(i)}(j)=\begin{cases}
            \map(j),\quad&j\neq i;\\
            \genie^{G,\otests,\map(-i)}(i),\quad&j=i.
        \end{cases}
    \end{align*}

Then, by the definition of the conditional probability, 
    \begin{align*}
    \ProbCond{\pats = \map}{ G,\otests}
    &= \ProbCond{\pats(i)=\map(i)}{ G,\otests,\pats_{-i}=\map(-i)} \ProbCond{\pats_{-i}=\map(-i)}{G,\otests},
\end{align*}
and
\begin{align*}
    \ProbCond{\pats = \map^{\ast(i)}}{ G,\otests}
    &= \ProbCond{\pats(i)=\genie(i)}{ G,\otests,\pats_{-i}=\map(-i)}
        \ProbCond{\pats_{-i}=\map(-i)}{ G,\otests}.
\end{align*}
Since $\ProbCond{\pats = \map}{G,\otests}=\max_{\hat{\sigma}\in \cbc{0,1}^n} \ProbCond{\hat{\sigma}=\pats}{ G,\otests}\geq 2^{-n}>0$, we have 
\begin{align*}
    \ProbCond{\map(-i)=\pats_{-i}}{ G,\otests}>0.
\end{align*}
Hence, under \eqref{eq:ass-map-i}, we conclude from \eqref{eq:cond--i-map-less-genie} that
    \begin{align*}
        \ProbCond{\pats = \map}{ G,\otests}<\ProbCond{\pats = \map^{\ast(i)}}{ G,\otests},
    \end{align*}
    which contradicts \eqref{eq:map-argmax}. On the other hand, under \eqref{eq:ass-map-i-2}, we conclude from \eqref{eq:cond--i-map-less-genie-2} that
    \begin{align*}
        \ProbCond{\pats = \map}{ G,\otests}=\ProbCond{\pats = \map^{\ast(i)}}{ G,\otests},
    \end{align*}
while $\map^{\ast(i)}$ has one more zero compared to $\map$, which contradicts \Cref{def:map}.
    
    Consequently, for any $i\in [n]$, 
    \begin{align}\label{eq:map-i}
        \map(i)=\genie^{G,\otests,\map(-i)}(i).
    \end{align}
Finally, by \eqref{eq:map-i},
\begin{align*}
\Prob{\pats = \map}&=\Prob{ \bigwedge_{i \in [n]} \pats(i)=\map(i) \wedge \pats_{-i}=\map(-i)}
\\ &\leq \Prob{\bigwedge_{i\in [n]} \pats(i)=\genie^{G,\otests,\map(-i)}(i)=\genie^{G,\otests,\pats_{-i}}(i)}
\\ &\leq \Prob{\pats = \genie},
\end{align*}
as desired.
\end{proof}

\subsection{Impossibility}\label{apx:nonad-imp}

\subsubsection{Proof of Lemma \ref{lem:geniethr}: characterization of genie estimator} \label{apx:lem:geniethr}
\geniethr*
\begin{proof} The genie estimator declares individual $i$ as healthy when $g_i(\pats)>0$ if the conditional probability of the individual being healthy is higher than that of being infected.
If the probabilities are equal, the genie estimator also declares it as healthy by definition.
This condition turns out to be
    \begin{align}
        &&\ProbCond{ \pats_i =1}{ \good_i(\pats), \good_i^+(\pats)} &\leq \ProbCond{ \pats_i =0 }{ \good_i(\pats), \good_i^+(\pats)} \\
        \Leftrightarrow&& p_{11}^{\good_i^+(\pats)}p_{10}^{\good_i(\pats)-\good_i^+(\pats)}\alpha &\leq p_{01}^{\good_i^+(\pats)}p_{00}^{\good_i(\pats)-\good_i^+(\pats)}(1-\alpha) \\
        \Leftrightarrow&& \bc{\frac{p_{11}}{p_{01}}}^{\good_i^+(\pats)} \bc{\frac{p_{10}}{p_{00}}}^{\good_i(\pats)} \bc{\frac{p_{10}}{p_{00}}}^{-\good_i^+(\pats)} &\leq \frac{1-\alpha}{\alpha} \\
        %\Leftrightarrow&& \good_i^+(\pats)\ln\bc{\frac{p_{11}}{p_{01}}} +\good_i(\pats)\ln\bc{\frac{p_{10}}{p_{00}}} -\good_i^+(\pats)\ln\bc{\frac{p_{10}}{p_{00}}} &\leq \ln\bc{\frac{1-\alpha}{\alpha}} \\
        \Leftrightarrow&& \good_i^+(\pats)\ln\bc{\frac{p_{11}p_{00}}{p_{01}p_{10}}} -\good_i(\pats)\ln\bc{\frac{p_{00}}{p_{10}}} &\leq \ln\bc{\frac{1-\alpha}{\alpha}}
        %\\ \Leftrightarrow&& \good_i^+(\pats) &\leq \underbrace{\frac{\ln\bc{\frac{\alpha}{1-\alpha}}}{\ln\bc{\frac{p_{01}p_{10}}{p_{00}p_{11}}}}}_{\tconst(\alpha,\vp)} + \good_i(\pats) \underbrace{\frac{\ln\bc{\frac{p_{00}}{p_{10}}}}{\ln\bc{\frac{p_{11}p_{00}}{p_{01}p_{10}}}}}_{C},
    \end{align}
    from which the claim follows immediately.
\end{proof}

%\subsubsection{Proof of Lemma \ref{lem:degtest}: High-degree tests may be omitted} \label{apx:lem:degtest}
\subsubsection{Proof of Lemmas \ref{lem:degtest} and \ref{lem:modifiedsetup}: modification of test design} \label{apx:lem:degtest}
\degtest*
\begin{proof}
Let $L = \cbc{a \in F \mid \abs{\partial a} \geq \log^2(n) }$ be the set of tests larger than $\log^2(n)$.
We know that $\abs{L} \leq m \leq cn\log(n)$ for some constant $c > 0$ and sufficiently large $n$ since $m = \Oh(n\log(n))$.
So
    \begin{align*}
         \Prob{\bigvee_{a \in L} \abs{ \cbc{i\in \partial a \mid \pats(i)=1}}\leq 1} &\leq \abs{L} \bc{ (1-\alpha)^{\log^2(n)} + \alpha(1-\alpha)^{\log^2(n)-1}\log^2(n)} \\
        & \leq  cn\log(n) \bc{ (1-\alpha)^{\log^2(n)} + \alpha(1-\alpha)^{\log^2(n)}\log^2(n)} \\
        &=  (1-\alpha)^{\log^2(n)} \cdot cn\log(n) \bc{ 1 + \alpha\log^2(n)} \\
        &=  \frac{cn\log(n) \bc{ 1 + \alpha\log^2(n)}}{(1-\alpha)^{-\log^2(n)}} \\
        &\leq \frac{2cn \log^3(n)}{n^{\log(n)(-\log(1-\alpha))}} \\
        &= o(1),
    \end{align*}
    as claimed.
\end{proof}

%\subsubsection{Proof of Lemma \ref{lem:modifiedsetup}: Genie estimator is better on modified test design}
\label{apx:lem:modifiedsetup}

\modifiedtest*
\begin{proof}
    The first inequality of \Cref{eq:easierlong} follows from \Cref{lem:geniebetter},  the last inequality follows from \Cref{lem:degtest}.
    Let $\pats_J$ be the infection vector for the high degree individuals in $J$. To prove the second inequality we couple the tests in $G'$ and $G''$ is such a way that the results of test $a$ in $G'$ and the corresponding test $a'$ in $G''$ are the same if all of individuals in $J\cap \partial a$ are healthy. Let $\otests'$ be the observed test results for $G'$ and $\otests''$ be the corresponding results in $G''$.
    \begin{align}
    \Prob{\pats = \map}
        \leq \sum_{\sigma_J} \ProbCond{\pats[\nnew]=\map[\nnew]}{\pats_J = \sigma_J} \Prob{\pats_J = \sigma_J}
    \end{align}
    Let use write $m' = \abs{F(G')}$ be the number of tests in $G'$ (and analogously write $m''$, $m_{\gdep}$ for the number of tests in $G''$ and $\modGraph$).
    To upper bound this expression we can just upper bound the conditional probability in it the following way,
    \begin{align}
        \ProbCond{\pats[\nnew]=\map[\nnew]}{\pats_J = \sigma_J} \leq \sup_{\hat{\sigma}:\{0, 1\}^{m'} \to \cbc{0,1}^{\nnew}} \ProbCond{\pats[\nnew]=\hat{\sigma}(\otests')}{\pats_J = \sigma_J}
    \end{align}
    Now, given $\pats_J$, if a high degree individual in a test $a$ in $\otests'$ is positive, then the test result does not depend on $\pats[\nnew]$, if none are positive then through the coupling the result is the same as of $a'$ in $\otests''$.
    \begin{align}
        \sup_{\hat{\sigma}: \{0,1\}^{m'} \to \cbc{0,1}^{\nnew}} \ProbCond{\pats[\nnew]=\hat{\sigma}(\otests')}{\pats_J = \sigma_J}
        &= \sup_{\hat{\sigma}:\{0,1\}^{m''} \to \cbc{0,1}^{\nnew}} \ProbCond{\pats[\nnew]=\hat{\sigma}(\otests'')}{\pats_J = \sigma_J} \\
        &= \sup_{\hat{\sigma}:\{0,1\}^{m''} \to \cbc{0,1}^{\nnew}} \Prob{\pats[\nnew]=\hat{\sigma}(\otests'')} \\
        &\leq \sup_{\hat{\sigma}:\{0,1\}^{m_{\gdep}} \to \cbc{0,1}^{\nnew}} \Prob{\pats[\nnew]=\hat{\sigma}(\otests_{\gdep})} \\
        &= \Prob{\pats[\nnew]=\map[\nnew]} \qedhere
    \end{align}
\end{proof}

\subsubsection{Proof of Lemma \ref{lem:genie_correct_upper_bound}: bound on probability of classifying distant sets}
\label{apx:lem:genie_correct_upper_bound}
\geniecorrectupperbound*
\begin{proof}
    First, by definition of $\misD$, conditioned on $\pats$,
        the genie estimator is correct only if $\abs{\misD = 0}$,
        so that
        \[\ProbCond{\pats[\nnew] = \genie^{\modGraph}[\nnew]}{\pats} \leq \ProbCond{\abs{\misD} = 0}{\pats}.\]
    Since the good tests of individuals in a distant set of tests $D$ don't overlap, the classifications of the genie estimator for individuals in $D$---and hence the events $(i \in \misD)_{i \in D}$---are independent.
    So $\abs{\misD}$ (conditioned on $\pats$) is a sum of independent indicator random variables,
        and $\ExpCond{\abs{\misD}}{\pats} \geq \VarCond{\abs{\misD}}{\pats}$.
    And hence, by Chebyshev's inequality,
        \[\ProbCond{\abs{\misD}=0}{\pats} \leq \frac{\VarCond{\abs{\misD}}{\pats}}{\ExpCond{\abs{\misD}}{\pats}^2} \leq \frac{1}{\ExpCond{\abs{\misD}}{\pats}}.\qedhere\]
\end{proof}

\subsubsection{Proof of Lemma \ref{lem:genie_problow}: lower bound on misclassification of individual}
\label{apx:lem:genie_problow}
\genieproblow*
\begin{proof}
Let $\modgood_i(\pats)$ (resp., $\modgood_i^-(\pats)$/$\modgood_i^+(\pats)$) be the number of good (resp., negatively/positively displaying good) tests of $i$ in $\modGraph$ under $\pats$.
Then by \Cref{def:genie} 
    \begin{align}
        \ProbCond{\genie^{\modGraph}(i)\neq \pats(i) }{\modgood_i(\pats)} &\geq \ProbCond{ \modgood_i^+(\pats) < \modgood_i C+\tconst \wedge \pats(i)=1 }{ \modgood_i(\pats)} \\
        &\quad + \ProbCond{ \modgood_i^+(\pats) \geq \modgood_i C+\tconst \wedge \pats(i)=0 }{ \modgood_i(\pats)} \\
        &= \ProbCond{ \modgood_i^+(\pats) < \modgood_i C+\tconst }{ \pats(i)=1, \modgood_i(\pats)} \cdot \alpha \\
        &\quad + \ProbCond{ \modgood_i^-(\pats) \leq \modgood_i (1-C)-\tconst }{ \pats(i)=0, \modgood_i(\pats)} \cdot (1-\alpha)
    \end{align}
    By Cram\'er's theorem (\cref{thm:cramers}) we get for any $\delta\in (0,C-p_{01})$, $\delta'\in(0, p_{11}-C)$:
    \begin{align}
        \lim_{\modgood_i(\pats) \to \infty} \frac{1}{\modgood_i(\pats)}\log \ProbCond{\modgood_i^+(\pats) < \modgood_i(\pats) (C- \delta/2) }{ \pats(i)=1,\modgood_i(\pats)} &= - \KL{C-\delta/2}{p_{11}}  \\
        \lim_{\modgood_i(\pats) \to \infty} \frac{1}{\modgood_i(\pats)}\log\ProbCond{\modgood_i^-(\pats) < \modgood_i(\pats) (1-C- \delta'/2) }{ \pats(i)=0,\modgood_i(\pats)} &= - \KL{1-C-\delta'/2}{p_{00}} 
    \end{align}
    And since  $C\in (p_{01},p_{11})$ for sufficiently large $\modgood_i(\pats)$ we get:
    \begin{align}
        \ProbCond{\modgood_i^+(\pats) < \modgood_i(\pats) (C- \delta/2) }{ \pats(i)=1,\modgood_i(\pats)} &\geq \exp\bc{-\modgood_i(\pats) \KL{C-\delta}{p_{11}} } \\
         \ProbCond{\modgood_i^-(\pats) < \modgood_i(\pats) (1-C- \delta'/2) }{ \pats(i)=0,\modgood_i(\pats)} &\geq \exp\bc{-\modgood_i(\pats) \KL{1-C-\delta'}{p_{00}} }
    \end{align}
    %Since $C<p_{11}$ and $1-C<p_{00}$ we know that $\KL{C-\eps/2}{p_{11}}<\KL{C-\eps}{p_{11}}$ we get
    Since $\tconst$ is a constant there exists a $g_0$ s.t. for any $\modgood_i>g_0$ we get $\min\cbc{\abs{\delta \modgood_i}, \abs{\delta' \modgood_i}}>\abs{\tconst}$.
    \begin{align}
        \ProbCond{\modgood_i^+(\pats) < \modgood_i(\pats) C+\tconst }{ \sigma(i)=1, \modgood_i(\pats)} &> \exp\bc{-\modgood_i(\pats) \KL{C-\delta}{p_{11}}} \\
        \ProbCond{\modgood_i^-(\pats) < \modgood_i(\pats) (1-C)+\tconst }{ \sigma(i)=0, \modgood_i(\pats)} &> \exp\bc{-\modgood_i(\pats) \KL{1-C-\delta'}{p_{00}}}
    \end{align}
    %Which means there exists an $\good_0\in \mathbb{N} $ so that the upper equation holds for $\modgood_i>\good_0$. Since $\good_i^{\gdep} > \eta\ln(n)$ tends to infinity we get
    %It follows that for any $0<\eps<1-p_{10}-C$ there exists a $D(\eps)= \min_{\modgood_i \leq \good_0} \frac{\ProbCond{\modgood_i^+(\pats) > \modgood_i (C- \eps/2) }{ \modgood_i(\pats)}}{\exp\bc{-\modgood_i \KL{C-\eps}{p_{11}}}} >0$ s.t. for any number $\modgood_i(\pats)$
    Since the Kullback--Leibler divergence is a continuous function and $\KL{C}{p_{11}}=\KL{1-C}{p_{00}}$ we can pick $\delta'$ s.t. $\KL{C-\delta}{p_{11}}=\KL{1-C-\delta'}{p_{00}}$.
    We also know that $\modgood_i\geq \eta\log(n)$ (since the modified test design contains at least that many individual tests per individual). Thus for sufficiently large $n$, $\modgood_i$ also becomes sufficiently large. In combination we get for any $\delta >0$ that
    \begin{align}
        \ProbCond{\genie^{\modGraph}(i)\neq \pats(i) }{ \modgood_i(\pats)} \geq \exp\bc{-\modgood_i \KL{C-\delta}{p_{11}}},
    \end{align}
        from which the claim follows immediately.
\end{proof}

\subsubsection{Proof of Lemma \ref{lem:atozero}: existence of difficult-to-classify distant set}

\label{apx:lem:atozero}
\lematozero*
\begin{proof}
    Let $b = \KL{C - \delta}{p_{11}}$,
        noting that $b > \KL{C}{p_{11}} = \beta$.
    Let further $\bm A_i=\exp\bc{-b \good_i(\pats, \modGraph)}$ and $\bm A=\sum_{i\in D} \bm A_i$. 
    Since $\bm A_i\leq 1$, it follows that $\Exp{\bm A_i} \geq \Exp{\bm A_i^2}$. 
    Using Chebyshev's inequality we get the following for every $y>0$:
    \begin{align}
        \Prob{\abs{\vA-\Exp{\vA}} \geq y \Exp{\vA}} \leq \frac{\Var{\vA}}{y^2\Exp{\vA}^2}    = \frac{\sum_{i\in D} \Var{\vA_i}}{y^2\Exp{\vA}^2}  \leq \frac{\sum_{i\in D} \Exp{\vA_i^2}}{y^2\Exp{\vA}^2} \leq \frac{\sum_{i\in D} \Exp{\vA_i}}{y^2\Exp{\vA}^2} = \frac{1}{y^2\Exp{\vA}}
    \end{align}
    Therefore, if $1/\Exp{\vA} \leq \delta''$ with $y=1/2$, we have
    \[\Prob{\vA \leq \frac{1}{2\delta''}}
        \leq \Prob{\vA \leq \frac{1}{2}\Exp{\vA}}
        \leq \frac{4}{\Exp{\vA}}
        \leq 4\delta'',\]
        so it is enough to show $1/\Exp{\vA} \leq \delta''$.
To prove that for arbitrarily small $\delta''$ there exists an $n_0$ s.t. \cref{eq:atozero} holds, we need to show that $1/\Exp{\vA}$ vanishes as $n$ goes to infinity. 

To this end, we bound $\Exp{\vA_i}$ and thus $\Exp{\vA}$ from below.
Expanding $\good_i(\pats, \modGraph) = \sum_{a \in \partial i} \vecone\cbc{\textup{$a$ is good for $i$}}$, we see that
\begin{equation}\label{eqn:exp_ai_good_fns}
    \Exp{\vA_i} = \Exp{\exp\bc{-b \good_i(\pats, \modGraph)}}
    = \Exp{\prod_{a \in \partial i} \exp\bc{-b \vecone\cbc{\textup{$a$ is good for $i$}}}}.
\end{equation}
Now increasing values in $\pats$, i.e., turning individuals infected, can only turn previously good tests for $i$ into not good ones.
So the functions $\exp(-b \vecone\cbc{\textup{$a$ is good for $i$}}$ (for $a \in \partial i$) are all increasing in $\pats$, and so, by \eqref{eqn:exp_ai_good_fns} and the FKG inequality (\cref{thm:fkg}),
\begin{equation}\label{eqn:exp_ai_good_fns_2}
    \Exp{\vA_i} \geq  \prod_{a \in \partial i} \Exp{\exp\bc{-b \vecone\cbc{\textup{$a$ is good for $i$}}}} \, .
\end{equation}
And since a test $a$ is good for $i$ if none of the individuals in $\partial a \setminus \{i\}$ is infected,
\begin{align}
    \Exp{\vA_i}
       &\geq \prod_{a\in \partial i} (1-(1-\alpha)^{\abs{\partial a}-1}(1-\eul^{-b}))
        = \exp\bc{\sum_{a\in \partial i} \log(1-(1-\alpha)^{\abs{\partial a}-1}(1-\eul^{-b}))}. \label{eq:bound_exp_ai_explog}
    \end{align}
For now, let us claim the following:
\begin{restatable}{claim}{davrg}\label{lem:davrg} 
    Say $b=\KL{C-\delta}{ p_{11}}$. For sufficiently small $\delta > 0$ and all sufficiently large $n$, we can choose a distant set $D$ such that for all $i \in D$ ,
    \begin{align}
        \sum_{a: a\in \partial i} -\log\bc{1-(1-\alpha)^{\abs{\partial a}-1}(1-\eul^{-b})} \leq -(1+\delta)\frac{\sum_{i\in [n]} \sum_{a\in \partial i} \log\bc{1-(1-\alpha)^{\abs{\partial a}-1}(1-\eul^{-b})}}{n}
    \end{align}
\end{restatable}
Applying the claim to \eqref{eq:bound_exp_ai_explog} the claim's choice of $D$, we have
    \begin{align}
    \Exp{\vA_i}
       &\geq \exp\bc{(1+\delta) \frac{\sum_{j\in [n]} \sum_{a\in \partial j} \log(1-(1-\alpha)^{\abs{\partial a}-1}(1-\eul^{-b}))}{n}} \\
       &= \exp\bc{(1+\delta) \frac{ \sum_{a\in \partial j} \sum_{j\in [n]} \log(1-(1-\alpha)^{\abs{\partial a}-1}(1-\eul^{-b}))}{n}} \\
       &= \exp\bc{(1+\delta) \frac{ \sum_{a\in \partial j} \abs{\partial a} \log(1-(1-\alpha)^{\abs{\partial a}-1}(1-\eul^{-b}))}{n}} 
    \end{align}
    Consider that the modified test design has a total number of 
    $\mnew=m+ \lfloor\eta \log(n) \rfloor \leq ((1-\varepsilon)\cna(\alpha, \bm{p}) + \eta)n\log(n)$
tests, where $\cna = \mna / (n\log n)$.
Furthremore, let 
        $\cna' = \min_{t\in \mathbb{N}^+} \cbc{(- t \log(1-(1-\alpha)^{t-1}(1-\eul^{-b})))^{-1}}$.
Then
    \begin{align}
       \Exp{\vA_i} &\geq \exp\bc{(1+\delta) \frac{ -((1-\eps)\cna n\log(n) + \eta n\log(n))\frac{1}{\cna'}}{n}} \\ 
       &= \exp\bc{-(1+\delta)((1-\eps)\frac{\cna}{\cna'} + \frac{\eta}{\cna'})\log(n)}
    \end{align}
    Hence,
    \begin{align}
        \Exp{\vA} &\geq \frac{n}{\log^{13}(n)} \exp\bc{-(1+\delta)\bc{(1-\eps)\frac{\cna}{\cna'} + \frac{\eta}{\cna'}}\log(n)}\\
        &= \frac{n}{\log^{13}(n)} n^{-(1+\delta)\bc{(1-\eps)\frac{\cna}{\cna'} + \frac{\eta}{\cna'}}} \label{eq:erwA} \, .
    \end{align}
To analyze $\Exp{\vA}$ further, we need the following claim ensuring that $\cna/\cna'$ approaches $1$ as $\eta$ vanishes.
\begin{restatable}{claim}{claimcna}\label{claim:cna'}
   As $\delta$ vanishes, i.e., $\delta \downarrow 0$, that $\cna'$ approaches $\cna$, i.e., $\cna'\downarrow\cna$.
\end{restatable}

    As a result, there exist  sufficiently small $\delta$ and $\eta$ (depending on $\eps$) such that $-(1+\delta)((1-\eps)\frac{\cna}{\cna'} + \frac{\eta}{\cna'})$, the exponent of \cref{eq:erwA}, is larger than  $-1$. 
    Then
    \begin{align}
        \lim_{n \to \infty}  \frac{1}{\Exp{\vA}}   = 0 \,.
    \end{align}
Hence, $1/\Exp{\vA} \leq \delta''$ for sufficiently large $n$, which concludes the proof.
\end{proof}

We now give the proofs of the claims above.

%\davrg*
\begin{proof}[Proof of \cref{lem:davrg}]
    Say $R=[\nnew]$ at first. We pick $D$ greedily by choosing individual $i$ with the smallest $\sum_{a: a\in \partial i}-\log(1-(1-\alpha)^{\abs{\partial a}-1}(1-\eul^{-b}))$ out of $R$ and include it in $D$.
    Afterwards we remove all individuals in the fourth neighborhood of $i$ (due to the bounded degrees there are a maximum of $\log^{12}(n)$ of those) from $R$ and repeat. Furthermore $-\log(1-(1-\alpha)^{\abs{\partial a}-1}(1-\eul^{-b})) \in (0,b)$.
    When we have picked all individuals we know 
    \begin{align}
        \abs{ R} \leq \frac{\nnew}{\log^{13}(n)} \log^{12}(n) = \frac{\nnew}{\log(n)} \;.
    \end{align}
Note that $-\log(1-(1-\alpha)^{\abs{\partial a}-1}(1-\eul^{-b}))\geq 0$.    For sufficiently large $n$ such that $n/(\nnew-\nnew/\ln(n))\leq 1+\delta$, we can now calculate for any $i\in [\nnew]\setminus R$ that 
    \begin{align}
        \sum_{a: a\in \partial i}-\log(1-(1-\alpha)^{\abs{\partial a}-1}(1-\eul^{-b})) &\leq \frac{\sum_{i\in[\nnew]\backslash R}\sum_{a\in \partial i}-\log(1-(1-\alpha)^{\abs{\partial a}-1}(1-\eul^{-b}))}{\nnew-|R|}\\
        &\leq \frac{\sum_{i\in[\nnew]}\sum_{a\in \partial i}-\log(1-(1-\alpha)^{\abs{\partial a}-1}(1-\eul^{-b}))}{\nnew-\frac{\nnew}{\log(n)}}\\
&\leq \frac{\sum_{i\in[n]}\sum_{a\in \partial i}-\log(1-(1-\alpha)^{\abs{\partial a}-1}(1-\eul^{-b}))}{n}(1+\delta) \; .\qedhere
    \end{align}
\end{proof}

%\claimcna*
\begin{proof}[Proof of \cref{claim:cna'}]\label{prf:claim:cna'}
    Indeed, recall from \eqref{def:ma} that $$
   \cna = \min_{t\in \mathbb{N}} \cbc{(- t \log(1-(1-\alpha)^{t-1}(1-\eul^{-\beta})))^{-1}} \,.
$$ Given $t$, $(- t \log(1-(1-\alpha)^{t-1}(1-\eul^{-x})))^{-1}$ is a strictly decreasing function of $x$. Since $b>\beta=\KL{C}{ p_{11}}$, we have $\cna'<\cna$. On the other hand, as $t\to\infty$, for fixed $x>0$,
\begin{align*}
(- t \log(1-(1-\alpha)^{t-1}(1-\eul^{-x})))^{-1}\sim  (t(1-\alpha)^{t-1}(1-\eul^{-x}))^{-1}\to\infty.
\end{align*}
Consequently, there exists a $t_0\in\NN^+$ where the the minimum of $\min_{t\in \mathbb{N}^+}\cbc{(- t \log(1-(1-\alpha)^{t-1}(1-\eul^{-\beta})))^{-1}}$ is obtained at $t=t_0$. Hence,
\begin{align*}
    \liminf_{\eps \downarrow 0}\cna'\geq \lim_{\eps \downarrow 0}(- t_0 \log(1-(1-\alpha)^{t_0-1}(1-\eul^{-b})))^{-1}=(- t_0 \log(1-(1-\alpha)^{t_0-1}(1-\eul^{-\beta})))^{-1}=\cna,
\end{align*}
which yields the claim.
\end{proof}

\subsection{Achievability}

\subsubsection{Proof of Lemma \ref{obs:naalg_xi}: properties of $\xi$}\label{apx:nonad:xi_calc}

\newcommand{\tempconst}{c}
\xicalc*
\begin{proof}
    \emph{(a)}
    For the lower bound on $\xi = \xi(\alpha, \beta, \Gamma)$,
        note that $1-(1-\alpha)^{\Gamma-1} \cdot (1-\eul^{-\beta}) \geq \eul^{-\beta}$,
        so that $\xi \geq \bc{-\log(\eul^{-\beta})}^{-1} = \beta^{-1}$.
    For the upper bound,
        first note that $\Gamma - 1 \leq \alpha^{-1}$
            (since $\Gamma \leq \ceil{\alpha^{-1}}$)
        and that, for $x \in [0, 1]$, it is the case that $(1-x)^{1/x} \geq (1-x)/e$.
    So then $\xi^{-1} = -\log(1-(1-\alpha)^{\Gamma-1}(1-\eul^{-\beta}))
        \geq -\log(1 - (1-\alpha)(1-\eul^{-\beta})/e) \geq (1-\alpha)(1-\eul^{-\beta}) / e$.
    
    To see that $\xi$ is increasing in $\alpha$,
        first see that $(-\ln(x))^{-1}$ is increasing in $x$ for $x \in (0, 1)$.
    Then note that $(1-\alpha)^{\Gamma-1}$ is decreasing in $\alpha$ and takes values in $(0, 1]$ (since $\alpha \in (0, 1)$),
        and so $1 - (1-\alpha)^{\Gamma-1} (1-\eul^{-\beta(\vp)})$ is increasing in $\alpha$ and also takes values in $(0, 1)$ (since $\beta(\vp) > 0$).

    \emph{(b)}
    Let $z = z(\Gamma) = (1-\alpha)^{\Gamma}$, and $\tempconst = \tempconst(\alpha, \beta) = \bc{1-\eul^{-\beta(\vp)}} \bc{1-\alpha}^{-1}$ (which is strictly positive for all allowed values of $\beta(\vp)$ and $\alpha$).
    Then $z \in (0, 1-\alpha]$, and the function to be optimized is
        $- \bc{\log\bc{1-\alpha}}^{-1} \log(z) \log\bc{1 - \tempconst \cdot z}$,
        which is proportional to
        \(f_\tempconst(z) = \log(z)\log(1-\tempconst z).\)
    Now
    \[\frac{\partial}{\partial z} f_\tempconst(z) = \frac{\log(1-\tempconst z)}{z} - \frac{\tempconst\log(z)}{1-\tempconst z},\]
        which is $0$ for $\tempconst=0$ (which is not an attainable value).
    And
    \[\frac{\partial^2}{\partial z \partial c} \log(z)\log(1-\tempconst z)
        = \frac{-z}{z(1-\tempconst z)} - \frac{\log(z) (1-\tempconst z) - \tempconst \log(z) (-z)}{(1-\tempconst z)^2}
        = \frac{-\log(z)-1+\tempconst z}{(1-\tempconst z)^2},\]
        which is positive for all $\tempconst \geq 0$ and $z \in (0, 1/\eul)$.
    So any maximum must either be at the upper boundary of the range of $z$ (i.e., be $z = 1-\alpha$),
        or have $z \geq 1/e$.
    In the former case, the optimal $\Gamma$ is $1$, trivially at most $\ceil{1/\alpha}$ (as $\alpha \leq 1$).
    In the latter case,
        we have $\Gamma \leq \ceil{-\bc{\log(1-\alpha)}^{-1}} \leq \ceil{\alpha^{-1}}$ as well.
\end{proof}

\subsubsection{Proof of Lemma \ref{lem:naalg_test_count}: bound on number of used tests}\label{apx:lem:naalg_test_count}
\lemnaalgtestcount*
\begin{proof}
    First, we see that $\abs{\NAIdvTests} = \ceil{\eta \log n}\cdot n \leq \eta n \log n + n$,
        and that $\abs{\NAGroupTests} \leq \bc{1 + \frac{\varepsilon}{3}}\frac{\xi}{\Gamma}n\log n + 1$.
    So it is sufficient to show that $\abs{\NAAddIdvTests} = \Oh(n)$.
    
    To that end, let us bound the probability that $\NAAddIdvTests$ contains any test for a particular $i \in [n]$:
        letting $X = \abs{\NAGroupTests \cap \partial i}$ be the number of tests in $\NAGroupTests$ containing $i$,
        this happens only if $X \leq \bc{1+\varepsilon/6} \xi \log n$.
    As each test in $\NAGroupTests$ is independently chosen among all tests of size $\Gamma$,
        $X$ is binomially distributed as $\bin{\ceil{\bc{1+\varepsilon/3}\frac{\xi}{\Gamma}n\log n}, \Gamma / n}$.
    Hence \(\bc{1+\frac{\varepsilon}{3}}\xi \log n \leq \Exp{X} \leq \bc{1+\frac{\varepsilon}{3}}\xi \log n + 1,\)
        and by \cref{thm:chernoff_t},
    \begin{align*}
        \Prob{X \leq \bc{1+\frac{\varepsilon}{6}}\xi\log n}
           &\leq \Prob{X \leq \Exp{X} - \frac{\varepsilon}{6} \cdot \xi \log n}
        \\ &\leq \exp\bc{-\,\frac{\bc{\frac{\varepsilon}{6} \cdot \xi \log n}^2}{2\bc{\bc{1+\varepsilon/3}\xi \log n + 1}}}
            \leq \exp\bc{-\,\frac{\varepsilon^2 \xi \log n}{72 \bc{1+\frac{\varepsilon}{3} + \frac{1}{\xi \log n}}}}
        \\ &= \exp\bc{-\xi \log n \cdot \frac{\varepsilon^2}{72\bc{1+\frac{\varepsilon}{3}}} \cdot \bc{1 - o(1)}}
            = n^{-\Omega(1) \cdot \xi},
    \end{align*}
        where we used both $\varepsilon > 0$ and $\xi \geq \beta^{-1} = \Omega(1)$ (\Cref{obs:naalg_xi} (a)).

    Now $\NAAddIdvTests$ contains at most $(1+\frac{\varepsilon}{6})\xi \log n + 1$ tests for each individual $i$ for which it contains tests at all,
        giving us the upper bound $\Exp{\abs{\NAAddIdvTests}} = \Oh(\xi \log n \cdot n^{-\xi \cdot \Omega(1)}) = o(n)$,
        so that by Markov's inequality we have $\NAAddIdvTests = \Oh(n)$ \whp,
        which is what we needed to show.
\end{proof}

\subsubsection{Proof of Lemma \ref{lem:naalg_pseudogenie}: error probability of pseudo-genie}

\label{apx:lem:naalg_peudogenie}
\naalgpseudogenie*
\begin{proof}
    That the events are independent follows from the fact that
        the value of $\patSPG{i}$ only depends on the displayed results of $i$'s individual tests in $\NAIdvTests$, which in turn only depend on $\pats(i)$,
        and that the $\pats(i)$ as well the test noise of different tests are independent.
    From this description, and recalling that $\NAIdvTests(i)$ contains the $N = \ceil{\eta \log n}$ individual tests for $i$,
        we can also see that $\abs{\NAIdvTests^+(i)}$
            (i.e., the number of tests in $\NAIdvTests$ displaying positively containing $i$)
        is distributed as $\bin{N, \pflipoi}$ when $\pats(i) = 0$
        and as $\bin{N, \pflipii}$ when $\pats(i) = 1$.
    So as $\pflipoi \leq C \leq \pflipii$, we have
    \begin{align*}
    \ProbCond{\patSPG{i} = 0 }{ \pats(i) = 1}
       &= \Prob{\frac{1}{N} \cdot \bin{N, \pflipii} < C}
        \leq \exp\bc{-N \KL{C}{\pflipii}
        \leq \exp\bc{-N\beta}},
\\  \ProbCond{\patSPG{i} = 1 }{ \pats(i) = 0}
       &= \Prob{\frac{1}{N} \cdot \bin{N, \pflipoi} \geq C}
        \leq \exp\bc{-N \KL{C}{\pflipoi}
        \leq \exp\bc{-N\beta}},
\\  \Prob{\patSPG{i} \neq \pats(i)}
       &= \alpha \cdot \ProbCond{\patSPG{i} = 0 }{ \pats(i) = 1}
        + (1-\alpha) \cdot \ProbCond{\patSPG{i} = 1 }{ \pats(i) = 0}
    \\ &\leq \exp\bc{-N\beta},
\\ \exp\bc{-N\beta} &= \exp(-\ceil{\eta \log n} \beta) \leq \exp(-\beta \eta \log n) = n^{-\beta \eta},
    \end{align*}
        as claimed.
\end{proof}

\subsubsection{Proof of Lemma \ref{lem:naalg_test_tainted}: probability of pseudo-good test not being good}

\lemnaalgtesttainted*
\label{apx:lem:naalg_test_tainted}
\begin{proof}
    If $a \in P_i$, then all individuals in $\partial a$ (except perhaps $i$ itself) have $\patSPG{i} = 0$.
    So for an individual $j \in \partial a \setminus \cbc{i}$ to be misclassified,
        it needs to be the case that $\pats(j) = 1$.
    So let us bound this probability from above using Bayes' theorem as well as \Cref{lem:naalg_pseudogenie}:
    \[\ProbCond{\pats(i)=1 }{ \patSPG{i} = 0}
        = \frac{\ProbCond{\patSPG{i} = 0 }{ \pats(i) = 1} \cdot \Prob{\pats(i) = 1}}{\Prob{\patSPG{i} = 0}}
        \leq \frac{n^{-\beta\eta} \cdot \alpha}{(1-\alpha)(1 - n^{-\beta\eta})}
        = \frac{\alpha}{1-\alpha} \cdot \frac{n^{-\beta\eta}}{1 - n^{-\beta\eta}}.\]
    The claim now follows via union bound over the $\Gamma - 1 \leq \Gamma$ individuals in $a$ not including $i$,
        the fact that $\Gamma \leq \ceil{\hat{\alpha}^{-1}} \leq \ceil{\alpha^{-1}}$ (as $\alpha \leq \hat{\alpha}$)
        as well as $\alpha \in (0, 1)$ being bounded away from $1$.
\end{proof}

\subsubsection{Proof of Lemma \ref{lem:naalg:misclassify_probcond_abs_Pi}: probability of correct classification conditioned on number of pseudo-good tests}

\label{apx:lem:naalg:misclassify_probcond_abs_Pi}
\lemnaalgmisclassifyprobcondabsPi*
\begin{proof}
    As $P_i \subseteq D_i$, the observed test results $(\otests(a))_{a \in P_i}$ are independent when conditioning on $\pats(i)$, the true infection status of $i$.
    To bound the probability of $i$ being misclassified, we consider the cases $\pats(i) = 1$ and $\pats(i) = 0$ separately.
    
    If $\pats(i) = 1$, then all tests $a \in P_i$ contain an infected individual (namely $i$),
        and as a consequence, $\ProbCond{\otests(a) = 1 }{ \pats(i) = 1, a \in P_i} = \pflipii$,
        so that $\sum_{a \in P_i} \otests(a)$ conditioned on $\abs{P_i}$ is distributed as $\bin{\abs{P_i}, p_{11}}$.
    And $i$ is classified as uninfected by $\SPOG$ if at most a $C$ proportion of tests in $P_i$ display positively,
        so that, by the Chernoff bound,
        \[\ProbCond{\patNA{i} = 0 }{ \pats(i) = 1, \abs{P_i}}
            = \ProbCond{\frac{1}{\abs{P_i}} \sum_{a \in P_i} \otests(a) < C }{ \pats(i)=1, \abs{P_i}}
            \leq \exp\bc{-\abs{P_i}\KL{C}{\pflipii}}.\]
    
    If $\pats(i) = 0$, tests $a \in P_i$ contain an infected individual if and only if some individual $j \in \partial a \setminus \cbc{i}$ was misclassified by $\patsSPG$
        (recalling that $a \in P_i$ only if for all such $j$, we have $\patsSPG(j)=0$).
    The probability of this occurring for a given test is in $\Oh(n^{-\beta\eta})$ by \cref{lem:naalg_test_tainted},
        and so 
        \[\ProbCond{\otests(a) = 1 }{ \pats(i) = 0, a \in P_i}
        = \pflipoi + (\pflipii - \pflipoi) \cdot \Oh(n^{-\beta \eta})
        = \pflipoi + \Oh(n^{-\beta \eta}).\]
    And as $i$ is classified as infected by $\SPOG$ if at least a $C$ proportion of tests in $P_i$ display positively,
    \begin{align*}
        \ProbCond{\patNA{i} = 1 }{ \pats(i) = 0, \abs{P_i}}
           &= \ProbCond{\frac{1}{\abs{P_i}} \sum_{a \in P_i} \otests(a) \geq C }{ \pats(i)=0, \abs{P_i}}
        \\ &\leq \exp\bc{-\abs{P_i}\KL{C}{\pflipoi + \Oh\bc{n^{-\beta \eta}}}}.
    \end{align*}
    Now by definition of $\beta$ and $C$,
        $\min\cbc{\KL{C}{\pflipoi}, \KL{C}{\pflipii}} = \beta$
        and $\pflipoi \leq C \leq \pflipii$ (and $\pflipoi < \pflipii$ by assumption on the channel),
        and so $\KL{C}{\pflipoi + \Oh(n^{-\beta\eta}} = \KL{C}{\pflipoi} - \Oh(n^{-\beta\eta})$.
    Combined with the above, this yields
        \[\ProbCond{\patNA{i} \neq \pats(i) }{ \abs{P_i}} \leq \exp\bc{-\abs{P_i}\bc{\beta - \Oh(n^{-\beta\eta})}}. \qedhere\]
\end{proof}

\subsubsection{Proof of Lemma \ref{lem:naalg_many_distinctive_tests}: lower bound on number of distinctive tests}

\label{apx:lem:naalg_many_distinctive_tests}

\lemnaalgmanydistinctivetests*
\begin{proof}
    By construction of $\vG_{\SPOG}$, we always have $\abs{\TestSet{2} \cap \partial i} \geq \ceil{\bc{1 + \frac{\varepsilon}{6}}\xi \log n + 1}$ for all $i$.
    Now let $X_i$ be the set of the first $\ceil{\bc{1 + \frac{\varepsilon}{6}} \xi \log n + 1}$ tests $a$ considered in determining $D_i$.
    We show that \whp $\abs{X_i \setminus D_i} \leq 1$ for all $i$ simultaneously, which then immediately yields the claim.

    In the loop determining $D_i$,
        for each test $a \in X_i$ under consideration,
            we add at most $\Gamma - 1$ tests to $S$
            (since each test contains at most $\Gamma$ individuals).
    So in \emph{any} iteration of the loop in step 5,
        we have $\abs{S} \leq \abs{X_i} \cdot \Gamma$.
    So for any $a \in \partial i$,
        \[\Prob{\partial a \cap S \neq \emptyset}
            \leq \Gamma \cdot \frac{\abs{X_i} \cdot \Gamma}{n}.\]
    So the total number of such intersections for an individual is majorized by
        $\bin{\abs{X_i}, \frac{\Gamma^2}{n} \abs{\NAGroupTests \cap \partial i}}$.
    Now for any $k \geq 2$ and $0 < p < 1$,
    \begin{align*}
    \Prob{\bin{k,p} \geq 2}
       &= 1 - (1-p)^k - k p (1-p)^{k-1} = 1 - (1-p)^{k-1} \bc{1-p + kp}
    \\ &\leq 1 - \bc{1-p(k-1)}\bc{1+p(k-1)} = 1 - 1 + p^2 (k-1)^2 \leq p^2 k^2,
    \end{align*}
    and hence
        \[\Prob{\abs{X_i \setminus D_i} \geq 2 }
            \leq \abs{X_i}^4 \Gamma^4 n^{-2}.\]
    Applying the union bound over all $i \in [n]$
        as well as the fact that $\abs{X_i} = \Oh(\log n)$ (since $\xi = \Oh(1)$ by \cref{obs:naalg_xi} (a)),
        we have the claim.
\end{proof}

\subsubsection{Proof of Lemma \ref{lem:naalg_misclassify}: probability of misclassification}

\label{apx:lem:naalg_misclassify}
\lemnaalgmisclassify*
\begin{proof}
    Pick any $i \in [n]$.
    Conditioning on the set $D_i$,
        the events $(a \in P_i)_{a \in D_i}$ are independent by construction of $D_i$
        (tests in $D_i$ only intersect in $i$, and membership in $P_i$ only depends on the classification of individuals in $a$ except for $i$).
    As \[\Prob{\patSPG{j} = 1} \leq \Prob{\pats(j) = 1 \vee \patSPG{j} \neq \pats(j)} \leq \alpha + n^{-\beta \eta}\] for all $j \in [n]$ (by union bound and \cref{lem:naalg_pseudogenie}),
        we have \[\ProbCond{a \in P_i }{ a \in D_i} = \Prob{\bigwedge_{j \in \partial a \setminus \{i\}} \patSPG{j} = 0}
            \geq \bc{1 - \alpha - n^{-\beta\eta}}^{\Gamma - 1}.\]
    Taken together, this means that $\abs{P_i}$ (conditioned on $\abs{D_i}$) stochastically dominates the binomial distribution $\bin{\abs{D_i}, (1-\alpha - n^{-\beta \eta})^{\Gamma - 1}}$.
    As \cref{lem:naalg:misclassify_probcond_abs_Pi}'s bound on the probability of misclassification conditioned on $\abs{P_i}$ is monotonically decreasing in $\abs{P_i}$,
        writing $q = (1 - \alpha - n^{-\beta\eta})^{\Gamma - 1}$, we thus have
    \begin{align*}
        \ProbCond{\patNA{i}\neq \sigma(i) }{ \abs{D_i}}
           &\leq \sum_{\ell=0}^{\abs{D_i}} \binom{\abs{D_i}}{\ell} q^\ell (1-q)^{\abs{D_i}-\ell} \cdot \eul^{-\ell \cdot \bc{\beta-\Oh(n^{-\beta \eta})}}
        \\ &= \bc{q\eul^{-\bc{\beta - \Oh(n^{-\beta\eta})}} + (1-q)}^{\abs{D_i}}
            = \bc{1 - q \cdot \bc{1 - \eul^{-\bc{\beta-\Oh(n^{-\beta\eta})}}}}^{\abs{D_i}}
        \\ &= \bc{1 - q \cdot \bc{1 - \eul^{-\beta} \cdot (1 + \Oh(n^{-\beta\eta})}}^{\abs{D_i}}
            = \bc{1 - q \cdot \bc{1 - \eul^{-\beta}} + \Oh(n^{-\beta\eta})}^{\abs{D_i}}.
    \end{align*}
    With $q = \bc{1-\alpha}^{\Gamma-1} \cdot \bc{1-\frac{n^{-\beta\eta}}{1-\alpha}}^{\Gamma-1}
        \geq \bc{1-\alpha}^{\Gamma-1} \cdot \bc{1 - \frac{\Gamma}{(1-\alpha) n^{\beta\eta}}} = \bc{1-\alpha}^{\Gamma-1} \cdot (1 - \Oh(n^{-\delta}))$,
        this yields
        \[\ProbCond{\patNA{i}\neq \sigma(i) }{ \abs{D_i}}
        \leq \bc{1 - (1-\alpha)^{\Gamma-1} \cdot \bc{1 - \eul^{-\beta}} + \Oh(n^{-\delta})}^{\abs{D_i}}.
        \]
    As a consequence,
        with $q' = 1 - (1-\alpha)^{\Gamma-1} \bc{1-\eul^{-\beta}} = \exp(-1/\xi(\alpha, \beta, \Gamma))$,
        and recalling that $\mathcal{D}$ is the event that $\abs{D_i} \geq \bc{1 + \frac{\varepsilon}{6}}\xi(\hat{\alpha}, \beta, \Gamma) \cdot \log n$ for all $i$,
        we have
    \begin{align*}
    \ProbCond{\patNA{i}\neq \sigma(i) }{ \mathcal{D}}
       &\leq \bc{q' + \Oh(n^{-\delta})}^{\bc{1+\frac{\varepsilon}{6}} \xi(\hat{\alpha}, \beta, \Gamma) \log n}
        = \bc{q' \cdot \bc{1 + \frac{\Oh(n^{-\delta})}{q'}}}^{\bc{1 + \frac{\varepsilon}{6}} \xi(\hat{\alpha}, \beta, \Gamma) \cdot \log n}
    \\ &= n^{\bc{1+\frac{\varepsilon}{6}}\xi(\hat{\alpha}, \beta, \Gamma) \cdot \bc{\log(q') + \Oh(n^{-\delta})/q'}}.
    \end{align*}
    Now $\log(q') = -1/\xi(\alpha, \beta, \Gamma)$,
        and $\xi(\alpha, \beta, \Gamma)$ is increasing in $\alpha$ (\Cref{obs:naalg_xi} (a)),
        so as $\alpha \leq \hat{\alpha}$,
        we have $\xi(\hat{\alpha}, \beta, \Gamma) \cdot \log(q') \leq -1$.
    Furthermore, by \Cref{obs:naalg_xi} (a),
        $\xi(\alpha, \beta, \Gamma) \in \brk{\beta^{-1}, \frac{e}{(1-\alpha)(1-\eul^{-\beta})}}$,
        which is in $\Oh(1)$ as $\alpha, \beta > 0$,
        and so $q' = \Oh(1)$ as well.
    Taken together, this yields 
    \[
    \ProbCond{\patNA{i}\neq \sigma(i) }{ \mathcal{D}}
        \leq n^{-\bc{1+\frac{\varepsilon}{6}} + \Oh(n^{-\delta})},
    \]
        as claimed.
\end{proof}

\subsubsection{Proof of Proposition \ref{prop:sublin}: $\SPOG$ for sub-constant $\alpha$}

\label{apx:prop:sublin}
\propsublin*
\begin{proof}
    We run $\SPOG$ on $G \sim \vG_{\SPOG}(n, n^{-\hat{\theta}}, \vp, \Gamma, \eta, \varepsilon)$,
        where $\Gamma = \ceil{\log(n)}$ (which is $\leq \ceil{(n^{-\hat{\theta}})^{-1}} \leq \ceil{\alpha^{-1}}$ for sufficiently large $n$),
        and $\eta = \varepsilon / 2$.
    This graph contains at most 
    $\frac{\varepsilon}{2} n \log n + (1+\varepsilon/3) \frac{\xi}{\log n} n \log n + \Oh(n)
        = \frac{\varepsilon}{2} n \log n + \Oh(n)$ tests \whp by \cref{lem:naalg_test_count} (using $\xi \leq \frac{e}{1-\eul^{-\beta(\vp)}} = \Oh(1)$ by \cref{obs:naalg_xi} (a)),
        which is at most $\varepsilon n \log n$ for sufficiently large $n$.
        
    Now $\mathcal{D}$ holds with probability $1 - \Oh(\Gamma^4 \log^4 n / n) = 1 - \Oh(\log^8 n / n)$ by \cref{lem:naalg_many_distinctive_tests} and choice of $\Gamma$.
    And as $\beta(\vp) \eta = \Omega(1)$,
        we have $\Gamma = \log(n) = \Oh(n^{\beta \eta - \delta})$ for any $\delta < \beta \eta$.
    Then as above, 
        combining this with \cref{lem:naalg_misclassify} (choosing $\delta = \beta\eta / 2$ arbitrarily) and the union bound over all $i \in [n]$,
        the probability of there being \emph{any} individual which is misclassified is in $\Oh(n^{-\frac{\varepsilon}{6} + \Oh(n^{-\beta\eta/2})}) = \Oh(n^{-\frac{\varepsilon}{6}})$,
        which is at most $\varepsilon$ for sufficiently large $n$, as claimed.
\end{proof}

\section{Full Proofs for Adaptive Group Testing}

\subsection{Impossibility}
\label{app-sec-imp-ada}

\subsubsection{Proof of Lemma \ref{lemma:prob_J}: ground truth is typical w.h.p.}
\label{sec_lemma:prob_J}

\probJ*

\begin{proof}
Recall that $\infected_{\pats}$ is  the set of infected individuals under $\pats$.
    Then, $\abs{\infected_{\pats}} \sim \bin{n,\alpha}$,
        and Chebyshev's inequality yields that
    \begin{align}\label{eq:lbngt:revise1}
        \Prob{\abs{\abs{\infected_{\pats}}-n\alpha}> \frac{\vepp}{2} n}\leq \frac{\Var{\bin{n,\alpha}}}{(\vepp/2)^2 n^2}=\frac{4\alpha(1-\alpha)}{\vepp^2 n}.
    \end{align}

Now fix $\pats$, and an $i \in \infected_{\pats}$.
We will bound the probability that $i \not\in \infectedprime^{\vepp}_{\pats,G_{\alg'},\otest_{\alg'}}$,
    i.e., that $\abs{\good_i^+(\pats, G_{\alg'}) - p_{11}\good_i(\pats, G_{\alg'})} > \vepp g_i(\pats, G_{\alg'})$.
To that end, consider an algorithm $\alg$ chosen to maximize this probability
    which has access to $\pats$ as well as the results of the $\rou{\eta \log n}$ individual tests for $i$ which are added to obtain~$\alg'$,
    imagining these added individual tests as being performed first.
Then $\alg$ cannot do better than to keep adding good tests for~$i$ until the number of positively displayed good tests deviates from the mean by a sufficient amount, at which point it would stop adding tests.
So let $\bm s_i(\ell)$ be the number of positively displayed tests among the first $\ell$ good tests for~$i$; 
    we define this even for $\ell$ larger than the number of tests performed, assuming that more individual tests are performed for $i$.
Then
    \[\ProbCond{i \not\in \infectedprime^{\vepp}_{\pats,G_{\alg'},\otest_{\alg'}}}{\pats,i \in \infected_{\pats}}
        \leq \ProbCond{\bigvee_{\ell = \rou{\eta \log n}}^\infty \abs{s_i(\ell) - p_{11}\ell} > \vepp \ell}{\pats,i \in \infected_{\pats}}.\]
As the observed results of good tests for $i$ are independent of the choices of the algorithm,
    we have $\bm s_i(\ell) \sim \bin{\ell, p_{11}}$.
And hence,
    for $\delta = \vepp / p_{11}$,
    using a two-sided Chernoff bound (\cref{thm:chernoff_eps_twoside}), we have
\begin{align*}
\ProbCond{i \not\in \infectedprime^{\vepp}_{\pats,G_{\alg'},\otest_{\alg'}}}{\pats}
   &\leq \sum_{\ell=\rou{\eta \log n}}^\infty \Prob{\abs{\bin{\ell, p_{11}} - p_{11}k} > \delta p_{11} \ell}
\\ &\leq \sum_{\ell=\rou{\eta \log n}}^\infty 2\exp\bc{-\,\frac{\delta^2 p_{11} \ell}{3}}
    \leq 2 \int_{\rou{\eta \log n} - 1}^\infty \exp\bc{-\,\frac{\delta^2 p_{11} \ell}{3}} \dif \ell
\\ &= 2 \cdot \frac{\exp\bc{-\frac{\delta^2 p_{11}}{3}\bc{\rou{\eta \log n} - 1}}}{\frac{\delta^2 p_{11}}{3}}
    \leq \frac{6}{\delta^2 p_{11}} \exp\bc{-\frac{\delta^2 p_{11}}{3}\bc{\eta \log n - 2}}
\\ &\leq \frac{6\exp\bc{\frac{2}{3} \cdot \delta^2 p_{11}}}{\delta^2 p_{11}} \cdot n^{-\delta^2 p_{11} \eta /3}
    = \frac{6 \exp\bc{\frac{2}{3} \vepp^2 / p_{11}}}{\vepp^2 / p_{11}} n^{-\vepp^2 \eta / (3 p_{11})}.
\end{align*}

Given this, we can bound the expected number of infected individuals \emph{not} in $\infectedprime^{\vepp}$
    by summing over all infected individuals and applying the tower rule:
    \begin{align*}
    \Exp{\abs{\infected_{\pats} \setminus \infectedprime_{\pats,G_{\alg'}, \otests_{\alg'}}^{\vepp}}}
       &= \Exp{\sum_{i \in \infected_{\pats}} \vecone\cbc{i \not\in \infectedprime_{\pats,G_{\alg'}, \otests_{\alg'}}^{\vepp}}}
        = \Exp{\sum_{i \in \infected_{\pats}} \ProbCond{i \not\in \infectedprime_{\pats,G_{\alg'}, \otests_{\alg'}}^{\vepp}}{\pats}}
    \\ &\leq n \cdot \frac{6 \exp\bc{\frac{2}{3} \vepp^2 / p_{11}}}{\vepp^2 / p_{11}} n^{-\vepp^2 \eta / (3 p_{11})}
        = \Oh(n^{1-\vepp \eta / (3 p_{11})}).
    \end{align*}
    And hence by Markov's inequality, $\abs{\infected_{\pats} \setminus \infectedprime_{\pats,G_{\alg'}, \otests_{\alg'}}^{\vepp}} \leq \frac{\vepp}{2} n$ \whp
    As $\abs{\abs{\infected_{\pats}} - \alpha n} \leq \frac{\vepp}{2} n$ \whp as well by \cref{eq:lbngt:revise1},
        we thus have $\abs{\abs{\infectedprime_{\pats,G_{\alg'}, \otests_{\alg'}}^{\vepp}} - \alpha n} \leq \vepp n$ \whp,
            and thus $\pats \in \TypicalConfigs(G_{\alg'}, \otests_{\alg'})$ \whp, as claimed.
\end{proof}

\subsubsection{Proof of Lemma \ref{lem:hd:lbound-nei-vec}: bound on posterior odds ratio of neighboring infection vectors}
\label{apx:lem:hd:lbound-nei-vec}

Recall that for any infection vector $\pat \in \{0,1\}^n$ and any $j \in \infected_{\pat}$,
    we write $\pat^{\downarrow j} = (\pat(i) \cdot \ind(i \neq j))_{i \in [n]} \in \{0, 1\}^n$ for the infection vector obtained from $\pat$ by setting the $j$th coordinate to $0$.
As $j$ is infected, we have $g_j^\pm(\pat, G, \otest) = g_j^\pm(\pat^{\downarrow j}, G, \otest)$,
    so we shall simply write $g_j^\pm$ throughout this section.

We first express the  posterior odds ratio of neighboring vectors as a function of $g_j^\pm$:
\begin{lemma}\label{cl:hd:pos-ratio}
For any $\pat \in \{0, 1\}^n$ and $j \in \infected_{\sigma}$,
\begin{align}\label{eq:hd:ratio-con-Atau}
    \frac{\ProbCond{\pats=\pat^{\downarrow j}}{G_\alg=G,\otests_\alg=\otest}}{\ProbCond{\pats=\pat}{G_\alg=G,\otests_\alg=\otest}}
        =\frac{1-\alpha}{\alpha} \cdot \exp\bc{-\good_j^- \cdot \ln \frac{p_{10}}{p_{00}}-\good_j^+ \cdot \ln \frac{p_{11}}{p_{01}}}.
\end{align}
\end{lemma}
\begin{proof}
    \newcommand{\algdet}{\mathcal{A}}
    To ease notation,
        we write $\algdet$ for the choices of $\alg$ leading to $G$,
        and write $\otests = \otests_{\alg}$

    By Bayes' theorem,
    \begin{align}\label{eq:hd:ratio-bayes}
        \frac{\ProbCond{\pats=\pat^{\downarrow j}}{ \alg=\algdet,\otests=\otest}}{\ProbCond{\pats=\pat}{\alg=\algdet,\otests=\otest}}
        = \frac{\Prob{\pats=\pat^{\downarrow j}}}{\Prob{\pats=\pat}}
        \cdot \frac{\ProbCond{\alg=\algdet,\otests=\otest}{\pats=\pat^{\downarrow j}}}{\ProbCond{\alg=\algdet,\otests=\otest}{\pats=\pat}}.
    \end{align}
    Now $\pat^{\downarrow j}$ and $\pat$ only differ in the $j$th coordinate, and the entries of $\pats$ are independent, so that
    \begin{align}\label{eq:hd:ratio-alpha}
      \frac{\Prob{\pats=\pat^{\downarrow j}}}{\Prob{\pats=\pat}}=\frac{\Prob{\pat_j=0}}{\Prob{\pat_j=1}}=\frac{1-\alpha}{\alpha}.
    \end{align}
    For the second factor of \eqref{eq:hd:ratio-bayes}, see first that, writing $\vx_{[i]}$ for the restriction of $\vx$ to the first $i$ coordinates,
    \begin{align*}
    \ProbCond{\alg=\algdet,\otests=\otest}{\pats=\pat}
       &=\prod_{i\in[m]}\ProbCond{\alg_i=\algdet_i,\otests_i=\otest_i}{\pats=\pat,\alg_{[i-1]}=\algdet_{[i-1]}, \otests_{[i-1]}=\otest_{[i-1]}}
    \\ &=\prod_{i\in[m]}\ProbCond{\otests_i=\otest_i}{\pats=\pat,\alg_{[i]}=\algdet_{[i]},\otests_{[i-1]}=\otest_{[i-1]}} \\ &\hspace{3em} \cdot \ProbCond{\alg_i = \algdet_i}{\pats=\pat,\alg_{[i-1]}=\algdet_{[i-1]},\otests_{[i-1]}=\otest_{[i-1]}}.
    \end{align*}

    By \Cref{def:ada-alg}, $\alg_i$ is a possibly random function of $\alg_{[i-1]}$ and $\otests_{[i-1]}$, and independent of $\pats$, while $\otest_i$ is independent of $\alg_{[i-1]}$ and $\otests_{[i-1]}$ given $\alg_i$.  Hence,
    \begin{align*}
    &\ProbCond{\alg=\algdet,\otests=\otest}{\pats=\pat}
        =\prod_{i\in[m]}\ProbCond{\otests_i=\otest_i}{\pats=\pat,\alg_i=\algdet_i}\ProbCond{\alg_i=\algdet_i}{\alg_{[i-1]}=\algdet_{[i-1]},\otests_{[i-1]}=\otest_{[i-1]}}.
    \end{align*}
    The same holds when replacing $\sigma$ by any other infection vector, like $\sigma^{\downarrow j}$,
        so that
    \[\frac{\ProbCond{\alg=\algdet,\otests=\otest}{\pats=\pat^{\downarrow j}}}{\ProbCond{\alg=\algdet,\otests=\otest}{\pats=\pat}}
        = \prod_{i \in [m]} \frac{\ProbCond{\otests_i=\otest_i}{\pats=\pat^{\downarrow j},\alg_i}}{\ProbCond{\otests_i=\otest_i}{\pats=\pat,\alg_i=\algdet_i}}.\]
    Now let $a_i=\max_{h \in \alg_i} \pat(h)$ and $a_i^j=\max_{h \in \alg_i} \pat^{\downarrow j}(h)$
        by the (hypothetical) actual test results of test $\alg_i$ under $\pat$ and $\pat^{\downarrow j}$.
    Then
    \begin{align*}
        \frac{\ProbCond{\otest_i}{\pats=\pat^{\downarrow j},\alg_i}}{\ProbCond{\otest_i}{\pats=\pat,\alg_i}}
        = \frac{p_{a_i^j,\otests_i}}{p_{a_i,\otests_i}}
        = \prod_{u,v \in \{0, 1\}} p_{uv}^{\ind\cbc{a_i^j=u,\otests_i=v} - \ind\cbc{a_i=u,\otests_i=v}}
        = \prod_{u,v \in \{0, 1\}} p_{uv}^{(\ind\cbc{a_i^j=u} - \ind\cbc{a_i=u})\ind\cbc{\otests_i=v}},
    \end{align*}
        and hence
    \[
    \frac{\ProbCond{\alg=\algdet,\otests=\otest}{\pats=\pat^{\downarrow j}}}{\ProbCond{\alg=\algdet,\otests=\otest}{\pats=\pat}}
        = \prod_{i \in [m]} \prod_{u,v \in \{0, 1\}} p_{uv}^{(\ind\cbc{a_i^j=u} - \ind\cbc{a_i=u})\ind\cbc{\otests_i=v}}
        = \prod_{u,v \in \{0, 1\}}p_{uv}^{\sum_{i=1}^m\bc{\ind\cbc{a_i^j=u}-\ind\cbc{a_i=u}}\ind\cbc{\otest_i=v}}.
    \]
    As $\pat$ and $\pat^{\downarrow j}$ only differ at coordinate $j$, the values $a_i$ and $a_i^j$ differ only if $\alg_i\cap \infected_{\pat}=\cbc{j}$,
        i.e., if test $\alg_i$ is good for $j$ under $\sigma$,
        and in this case, $a_i=1$ and $a_i^j=0$.
    Hence, for any $u, v \in \{0, 1\}$,
    \[
    \sum_{i=1}^m\bc{\ind\cbc{a_i^j=u}-\ind\cbc{a_i=u}}\ind\cbc{\otest_i=v}
        = \sum_{i=1}^m (-1)^u \ind\cbc{\textup{$\alg_i$ is good for $j$ under $\pat$}}\ind\cbc{\otests_i=v}.
    \]
    As $g_j^\pm=g_j^\pm(\pat)$ is the number of good tests for $j$ appearing positive/negative under $\pat$,
        combining everything, we see that
    \begin{align}\label{eq:hd:Atau-simple}
    \frac{\ProbCond{\alg=\algdet,\otests=\otest}{\pats=\pat^{\downarrow j}}}{\ProbCond{\alg=\algdet,\otests=\otest}{\pats=\pat}}=p_{00}^{\good_j^-}p_{01}^{\good_j^+}p_{10}^{-\good_j^-}p_{11}^{-\good_j^+}.
    \end{align}
    Equation \Cref{eq:hd:ratio-con-Atau} then follows directly from the combination of \Cref{eq:hd:ratio-bayes} and \Cref{eq:hd:ratio-alpha} and \Cref{eq:hd:Atau-simple}.
\end{proof}

Given this, we can now prove \cref{lem:hd:lbound-nei-vec}.
To that end, recall from \Cref{def:typical_config} that
$\infectedprime^{\vepp}_{\sigma}$ is 
the set of infected individuals with a typical ratio of good tests displaying positively as
    \begin{align*}
    \infectedprime^{\vepp}_{\pat}
        = \infectedprime^{\vepp}_{\pat,G,\otest}
        = \cbc{i\in \infected_{\pat} \mid \abs{\good_i^+(\pat)-p_{11}\good_i(\pat)} \leq \vepp \good_i(\pat)}.
    \end{align*}

\lboundneivec*

\begin{proof}
We first express the posterior ratio of neighboring vectors on the left-hand side of \Cref{eq:hd:lbound-nei-vec} as a function of $g_j^\pm$ through the following claim, which we prove in \Cref{app-sec-imp-ada} utilizing Bayes’ theorem as well as the definition of good tests:
    For $j\in \infectedprime^{\vepp}_{\sigma}$, by definition of $\infectedprime^{\vepp}$ and as $g_j = g_j^+ + g_j^-$,
    \begin{align*}
        \abs{\good_j^- - p_{10} \good_j}
            =\abs{\good_j-\good_j^+-(1-p_{11})\good_j}
            =\abs{\good_j^+ - p_{11}\good_j}
            \leq \vepp \good_j.
    \end{align*}
From this, the triangle inequality, and the definition of KL-divergence, we can see that
\begin{align*}
\abs{\bc{\good_j^-\ln\frac{p_{10}}{p_{00}}+\good_j^+\ln\frac{p_{11}}{p_{01}}}-\good_j\KL{p_{11}}{p_{01}}}
   &= \abs{\bc{g_j^- - g_j p_{10}} \ln \frac{p_{10}}{p_{00}} + \bc{g_j^+ - g_j p_{11}}\ln\frac{p_{11}}{p_{01}}}
\\ &\leq \vepp \good_j\bc{\ln\frac{p_{00}}{p_{10}}+\ln\frac{p_{11}}{p_{01}}}.
\end{align*}
From this, the lemma follows immediately follows by \Cref{cl:hd:pos-ratio}.
\end{proof}

\subsubsection{Proof of Lemma \ref{lem:hd:ubound-posterior-typical}: posterior probability of typical infection vectors vanishes}
\label{sec:hd:ubound-posterior-typical}

\upt*

\begin{proof}
Consider any test design $G$ and observed test vector $\otest$,
    and a $\pat \in \TypicalConfigs^{\vepp}(G, \otest)$.

For distinct infected individuals $j_1,j_2$, the infection vectors $\pat^{\downarrow j_1}$ and $\pat^{\downarrow j_2}$ are distinct.
Consequently, 
\begin{align*}
    \sum_{j\in \infectedprime^{\vepp}_{\pat}}\ProbCond{\vsigma=\pat^{\downarrow j}}{G_{\alg'}=G,\otests_{\alg'}=\otest}\leq 1.
\end{align*}
Then by \Cref{lem:hd:lbound-nei-vec} and Jensen's inequality,
    and writing $c= \ln \frac{p_{00}}{p_{10}} + \ln\frac{p_{11}}{p_{01}}$, we have
\begin{align}\label{eq:hd:jensen-ub-map}
    1&\geq \sum_{j\in \infectedprime^{\vepp}_{\pat}}\ProbCond{\vsigma=\pat^{\downarrow j}}{G_{\alg'}=G,\otests_{\alg'}=\otest}\\
    &\geq\sum_{j\in \infectedprime^{\vepp}_{\pat}}\frac{1-\alpha}{\alpha}\exp\bc{-\good_j(\pat) \bc{\KL{p_{11}}{p_{01}}+c\vepp}} \ProbCond{\vsigma=\pat}{G_{\alg'}=G,\otests_{\alg'}=\otest}\nn\\
    &\geq \abs{\infectedprime^{\vepp}_{\pat}}\frac{1-\alpha}{\alpha}\exp\bc{-\frac{\sum_{j\in \infectedprime^{\vepp}_{\pat}}\good_j(\pat)}{\abs{\infectedprime^{\vepp}_{\pat}}} \bc{\KL{p_{11}}{p_{01}}+c\vepp}}\ProbCond{\vsigma=\pat}{G_{\alg'}=G,\otests_{\alg'}=\otest}.\nn
\end{align}
As the sets of good tests for a group of infected individuals is disjoint,
    the sum of the sizes of these test is at most the total number of tests. 
And this is the sum of the number of tests used by the original algorithm $\alg$---at most$(1-\varepsilon)\adm$---and the at most $\frac{\varepsilon}{2} \adm$ additional tests we added for $\alg'$.
So recalling the definition of $\adm$, we have
    \[\sum_{i\in \infectedprime^{\vepp}_{\pat}}\good_i(\pat) \leq \bc{1-\frac{\eps}{2}}\adm = \bc{1-\frac{\eps}{2}}\frac{\alpha}{\KL{p_{11}}{p_{01}}}n\ln(n).\]

Now as $\pat \in \TypicalConfigs^{\vepp}(G, \otest)$, we have $\abs{\infectedprime^{\vepp}_{\sigma}} \geq (\alpha - \vepp)n$.
We thus conclude from \Cref{eq:hd:jensen-ub-map} that
\begin{align*}
\ProbCond{\vsigma=\pat}{G_{\alg'}=G,\otests_{\alg'} = \otest}
    &\leq \abs{\infectedprime^{\vepp}_{\pat}}^{-1}\frac{\alpha}{1-\alpha}\exp\bc{\frac{\sum_{j\in \infectedprime^{\vepp}_{\pat}}\good_j(\pat)}{\abs{\infectedprime^{\vepp}_{\pat}}} \bc{\KL{p_{11}}{p_{01}}+c\vepp }}\\
    &\leq (\alpha-\vepp)^{-1}n^{-1}\frac{\alpha}{1-\alpha}\exp\bc{\frac{(1-\frac{\eps}{2})\frac{\alpha}{\KL{p_{11}}{p_{01}}}n\ln(n)}{(\alpha-\vepp)n} \bc{\KL{p_{11}}{p_{01}}+c\vepp}}\nn\\
    &=(\alpha-\vepp)^{-1}\frac{\alpha}{1-\alpha}n^{-1+(1-\frac{\eps}{2})\frac{\alpha}{\alpha-\vepp}\bc{1+c\vepp\bc{\KL{p_{11}}{p_{01}}}^{-1}}}.
\end{align*}
% As this is universal over all $\pat, G, \otest$ with $\pat \in \TypicalConfigs^{\vepp}(G, \otest)$,
%     we have
% \[\ProbCond{\vsigma=\pat}{\pat\in\TypicalConfigs(G_{\alg'},\otests_{\alg'})}
%     \leq (\alpha-\vepp)^{-1}\frac{\alpha}{1-\alpha}n^{-1+(1-\frac{\eps}{2})\frac{\alpha}{\alpha-\vepp}\bc{1+c \vepp\bc{\KL{p_{11}}{p_{01}}}^{-1}}}.\]
For sufficiently small $\vepp$, the exponent of $n$ is negative, so that, for all such $\vepp$, the upper bound is in $o(1)$, as claimed.
\end{proof}

\subsection{Achievability}
% \subsubsection{Stage One}\label{apx:adalg:stageone}
\subsubsection{Proof of Lemma \ref{lem_sizeS}: bounds on $\abs{S_0}$ and $\abs{S_1}$}\label{sec:prf:lem_sizeS}
\lemsizeS*
\begin{proof}
    Recall that we write $\infected_{\pats}$ for the set of infected individuals under $\pats$.
    First, we bound the probability of individuals being in the ``wrong'' set $S_k$ conditioned on their infection status using a Chernoff bound (\cref{thm:chernoff_kl}):
    \begin{align}
        \ProbCond{i \in S_1 }{ i \notin \infected_{\pats}} \leq& \exp\bc{- \eta \log(n) \KL{C}{p_{01}}} = n ^ {- \eta \KL{C}{p_{01}}}\label{eq:S1H}  \\
        \ProbCond{i \in S_0 }{ i \in \infected_{\pats}} \leq& \exp\bc{- \eta \log(n)\KL{C}{p_{11}}} = n ^ {- \eta \KL{C}{p_{11}}}\label{eq:S0I}
    \end{align}
    Now, if $\abs{S_1}$ deviates from $\alpha n$ by more than $\epssizes \alpha n$,
        then either $\abs{S_1 \cap \infected_{\pats}} \leq \abs{S_1} < (1-\epssizes)\alpha n$,
        or $\abs{S_1} > (1+\epssizes)\alpha n$.
    And in the latter case, one of $\abs{S_1 \cap \infected_{\pats}} > \bc{1+\frac{\epssizes}{2}}\alpha n$ or $\abs{S_1 \cap \overline{\infected_{\pats}}} > \frac{\epssizes}{2}\alpha n$ must hold (since otherwise this is a contradiction).
    So
    \begin{align*}
    \Prob{\abs{\abs{S_1} - \alpha n } \geq \alpha n \epssizes }
        \leq \Prob{\abs{S_1 \cap \infected_{\pats}} < (1-\epssizes) \alpha n} + \Prob{\abs{S_1 \cap \overline{\infected_{\pats}}} > \frac{\epssizes}2 \alpha n} 
            + \Prob{\abs{\infected_{\pats}} >  \bc{1+\frac{\epssizes}{2}} \alpha n}
     \end{align*}
     The last term can easily be verified to be in $o(1)$, as $\abs{\infected_{\pats}} \sim\bin{n, \alpha}$ so that by a Chernoff bound,
         \begin{align*}
             \Prob{\abs{\infected_{\pats}} >  (1+\epssizes/2) \alpha n}  = 
             \Prob{\bin{n, \alpha} > (1+\epssizes/2) \alpha n} 
             \leq \exp\bc{-n\KL{(1+\epssizes/2) \alpha }{\alpha}}
         \end{align*}
     Since $\abs{S_1\cap \infected_{\pats}} \sim \bin{n, \alpha \ProbCond{i \in S_1 }{i \in \infected_{\pats}} }$,
        we can bound the probability of $\abs{S_1\cap\infected_{\pats}}$ being too small by combining \cref{eq:S1H} and another Chernoff bound:
     \begin{align*}
        \Prob{\abs{S_1 \cap \infected_{\pats}} < (1-\epssizes) \alpha n }  
        % =\Prob{\abs{S_0 \cap \infected_{\pats}} >\epssizes \alpha n}    
        &= \Prob{\bin{n, \alpha \ProbCond{i \in S_1 }{i \in \infected_{\pats}} } <(1-\epssizes)\alpha  n }\\
        % = & \Prob{\bin{n, 1-\alpha \ProbCond{i \in S_1 }{i \in \infected_{\pats}} } \geq n - (1-\epssizes)\alpha  n } \\
        &\leq \exp\bc{-  n \KL{\alpha(1-\epssizes) }{\alpha \bc{1- n ^{-\eta\KL{C}{p_{11}}}}}}
    \end{align*}
    Since $\alpha(1-\epssizes)$ is a constant and $\alpha \bc{1- n ^{-\eta\KL{C}{p_{11}}}} \stackrel{n \to \infty}{\to} \alpha$, the outer Kullback--Leibler divergence is strictly positive for sufficiently large $n$, and $\Prob{\abs{S_1 \cap \infected_{\pats}} < (1-\epssizes) \alpha n } $ vanishes.
    We bound the number of uninfected individuals in $S_1$ from above in a similar way, since $\abs{S_1\cap \overline{\infected_{\pats}}} \sim \bin{n ,(1-\alpha )\ProbCond{i \in S_1}{i \notin \infected_{\pats}}}$
    \begin{align*}
        \Prob{\abs{S_1 \cap\overline{\infected_{\pats}}} > \frac{\epssizes}{2} \alpha n} &= \Prob{\bin{n ,(1-\alpha )\ProbCond{i \in S_1}{i \notin \infected_{\pats}}} > \frac{\epssizes}{2} \alpha n }\\
        &\leq \exp\bc{- n \KL{\frac{\epssizes}{2} \alpha}{(1-\alpha)n^{-\eta\KL{C}{p_{01}}}}}
    \end{align*}
    Since $n^{-\eta\KL{C}{p_{01}}}$ tends to zero, $\Prob{\abs{S_1 \cap\overline{\infected_{\pats}}} > \frac{\epssizes}{2} \alpha n} = o(1)$.
    Thus, the claimed bound on $S_1$ \whp follows.
    The claimed bound on $\abs{S_0}$ then immediately follows as $S_0$ and $S_1$ partition $[n]$.
\end{proof}

% \subsubsection{Proof of Lemma \ref{lem:size_s1_notinf}}\label{sec:prf:lem:size_s1_notinf}
\subsubsection{Proof of Lemma \ref{lem_falsPos}: all individuals in $U_1$ infected}\label{sec:prf:lem_falsPos}
\lemfalsPos*
\begin{proof} % lem_falsPos
    Similar to the first stage we can bound the probabilities of misclassified individuals as follows 
    \begin{align}
        \ProbCond{i \in U_1 }{ i \in S_1 \cap \overline{\infected_{\pats}}} 
        &\leq n ^{- (1+\frac\eps4) \frac{\KL{p_{11} -\epsuthresh }{p_{01}}}{\KL{p_{11} }{p_{01}}}} \label{eq:U1H} \\
        \ProbCond{i \notin U_1 }{ i \in S_1 \cap \infected_{\pats}} 
        &\leq n ^{- (1+\frac\eps4) \frac{\KL{p_{11} -\epsuthresh  }{p_{11}}}{\KL{p_{11} }{p_{01}}}} \label{eq:U0I}
    \end{align}
    Now we can bound the probability that $U_1 \not\subseteq \infected_{\pats}$, or equivalently that $\overline{\infected_{\pats}} \cap U_1 \neq \emptyset$, by using \cref{eq:U1H} and a union bound over all individuals:
    \begin{align*}
        \Prob{\overline{\infected_{\pats}} \cap U_1 \neq \emptyset} 
        &\leq \sum_{i \in [n] }\Prob{ i \notin \infected_{\pats} }\ProbCond{i \in S_1}{i \notin  \infected_{\pats}} \ProbCond{i \in U_1}{ i \in S_1 \cap \overline{ \infected_{\pats}}}\\
        &\leq n \cdot
        \ProbCond{i \in U_1}{ i \in S_1 \cap \overline{ \infected_{\pats}}}\\
        &\leq n ^{1 -  (1+\frac\eps4) \frac{\KL{p_{11} - \epsuthresh }{p_{01}}}{\KL{p_{11} }{p_{01}}}}.
    \end{align*}
    As the exponent is negative by \eqref{eq:eps2_bound2}, the claim follows.
\end{proof}

% \subsubsection{Stage Three}
\subsubsection{Proof of Lemma \ref{lem:use_of_naalg_correct}: decoding $S_0$ and $U_0$ correctly using $\SPOG$}\label{sec:prf:lem:use_of_naalg_correc}
Before proving \cref{lem:use_of_naalg_correct}  directly, we need the following claim that gives a \emph{lower} bound on the number of uninfected individuals in $S_1$.
\begin{claim}[Lower bound on $\abs{S_1 \cap \overline{\infected_{\pats}}}$]\label{lem:size_s1_notinf}
    There is a $\epsnaalgproof > 0$ such that for all $\epsnaalg \in (0, \epsnaalgproof)$, \whp, \[\abs{S_1 \cap \overline{\infected_{\pats}}} \geq (1-\epsnaalg)(1-\alpha)n^{1-\eta\KL{C-\epsnaalg}{p_{01}}}\]
\end{claim}
% \lemsizeslnotinf*
\begin{proof}%(Lower bound on $\abs{S_1\cap \overline{\infected}}$)
    First, by Cram\'er's theorem (\cref{thm:cramers}),
        for any $\epsnaalg > 0$, it holds for all sufficiently large $n$ that
        \[\ProbCond{i \in S_1}{i \not\in \infected_{\pats}} \geq \exp\bc{-\ceil{\eta \log(n)} (\KL{C-\epsnaalg}{p_{01}})} \geq n^{-\eta(1+\frac{1}{\eta \log n}))\KL{C-\epsnaalg}{p_{01}}}.\]
    And then with $\nu = \eta\bc{1+\frac{1}{\eta \log n}}\KL{C-\epsnaalg}{p_{01}}$,
        by a Chernoff bound (\cref{thm:chernoff_eps}),
        \begin{align*}
        \Prob{\abs{S_1 \cap \overline{\infected_{\pats}}} < (1-\epsnaalg)(1-\alpha)n^{1-\nu}}
           &\leq \Prob{\bin{n, (1-\alpha)n^{-\nu}} < (1-\epsnaalg)(1-\alpha)n^{1-\nu}}
        \\ &\leq \exp\bc{-\,\frac{-\epsnaalg^2 (1-\alpha)n^{1-\nu}}{2}} = o(1),
        \end{align*}
        as for sufficiently large $n$, we have $\nu = \eta\bc{1+\frac{1}{\eta \log n}}\bc{\KL{C-\epsnaalg}{p_{01}}} < \eta \KL{C}{p_{01}} < 1$.
\end{proof}

\lemuseofnaalgcorrect*
\begin{proof}
    First, note as membership in $S_0$ (resp., $U_0$) is decided by the results of individual testing, the infection status of individuals in $S_0$ (resp., $U_0$) remains i.i.d.

    Now, by Bayes' theorem,
    \begin{align}
        \ProbCond{\pats(i) = 1}{i \in S_0}
            = \frac{\ProbCond{i \in S_0}{i \in \infected_{\pats}} \cdot \Prob{\pats(i)=1}}{\Prob{i \in S_0}}
            \leq \frac{\Prob{i \in S_0 \cap \infected_{\pats}}}{\Prob{i \in S_0 \cap \overline{\infected_{\pats}}}}.
    \end{align}
    Now
    \[\Prob{i \in S_0 \cap \overline{\infected_{\pats}}} = \Prob{\pats \not\in \infected_{\pats}} \cdot \ProbCond{i \in S_0}{i \not\in \infected_{\pats}}
        = (1-\alpha) \cdot (1 - \ProbCond{i \in S_1}{i \not\in \infected_{\pats}}).\]
    Hence, as the results of individual tests are independent for different $i$, and inclusion in $S_0$ or $S_1$ is based only thresholding on individual tests, we have, by a Chernoff bound (\cref{thm:chernoff_kl}),
    \[\ProbCond{\pats(i) = 1}{i \in S_0}
        \leq \frac{\alpha n^{-\eta \KL{C}{p_{11}}}}{(1-\alpha)(1-n^{-\eta\KL{C}{p_{01}}})}
        = \Oh(n^{-\eta \KL{C}{p_{11}}}),\]
        which is at most $n^{-\hat\theta^S}$ for sufficiently large $n$ by definition of $\hat\theta^S$,
            as needed for \cref{prop:sublin}.

    Similarly:
    \begin{align}
        \ProbCond{\pats(i) = 1}{i \in U_0}
            = \frac{\ProbCond{i \in U_0}{i \in \infected_{\pats}} \cdot \Prob{\pats(i)=1}}{\Prob{i \in U_0}}
            \leq \frac{\ProbCond{i \in U_0 \cap \infected_{\pats}}{i \in S_1}}{\Prob{i \in U_0 \cap \overline{\infected_{\pats}}}}.
    \end{align}
    Now
        \[\Prob{i \in U_0 \cap \overline{\infected_{\pats}}}
            = \ProbCond{i \in U_0}{i \in S_1 \cap \overline{\infected_{\pats}}} \ProbCond{i \in S_1}{i \not\in \infected_{\pats}} \Prob{i \not\in \infected_{\pats}}\]
        and so by a similar argument as above, a Chernoff bound yields
        \[\ProbCond{i \in U_0}{i \in S_1 \cap \overline{\infected_{\pats}}}
            = 1 - \ProbCond{i \in U_1}{i \in S_1 \cap \overline{\infected_{\pats}}}
            \geq 1 - n^{-\bc{1+\frac{\eps}{4}}\frac{\KL{p_{11} - \epsuthresh}{p_{01}}}{\KL{p_{11}}{p_{01}}}},\]
        and by Cram\'er's theorem for arbitrarily small $\vepp$ and sufficiently large $n$,
        \[\ProbCond{i \in S_1}{i \not\in \infected_{\pats}} \geq n^{-\eta(1+\frac{1}{\eta \log n}) \bc{\KL{C-\vepp}{p_{01}}}}.\]
    Taken together this is
        \begin{align}
        \ProbCond{\pats(i) = 1}{i \in U_0}
           &\leq \frac{\alpha n^{-\bc{1+\frac \eps 4}\frac{\KL{p_{11}-\epsuthresh}{p_{11}}}{\KL{p_{11}}{p_{01}}}}}{(1-\alpha) n^{-\eta (1+\frac{1}{\eta \log n})\KL{C-\vepp}{p_{01}}}\bc{1-n^{-\bc{1+\frac{\eps}{4}}\frac{\KL{p_{11} - \epsuthresh}{p_{01}}}{\KL{p_{11}}{p_{01}}}}}}.
        \\ &= \Oh\bc{n^{\eta (1 + \frac{1}{\eta \log n})\KL{C-\vepp}{p_{01}}-\bc{1+\frac \eps 4}\frac{\KL{p_{11}-\epsuthresh}{p_{11}}}{\KL{p_{11}}{p_{01}}}}}.
        \end{align}
    Now for sufficiently large $n$, we have
        \[\eta \bc{1 + \frac{1}{\eta \log n}}\KL{C-\vepp}{p_{01}} \leq \eta \KL{C}{p_{01}},\]
        so that the exponent of $n$ is bounded from above by a negative constant by \eqref{eq:eps2_nu}.
    And hence $\hat\theta^U > 0$ for sufficiently large $n$,
        and $\ProbCond{\sigma(i)=1}{i \in U_0} \leq n^{-\hat\theta^U}$ for sufficiently large $n$ by definition of $\hat\theta^U$, as needed for \cref{prop:sublin}.

    Lastly, by \cref{lem_sizeS,lem:size_s1_notinf},
        $\abs{S_0}$ and $\abs{S_1 \cap \overline{\infected_{\pats}}}$ diverge \whp,
        and by \cref{lem_falsPos}, $U_0 \supseteq S_1 \cap \overline{\infected_{\pats}}$ \whp,
        so that indeed the sizes of the sets of individuals given to $\SPOG$ diverge \whp and the probability of $\SPOG$ failing is in $o(1)$.
\end{proof}

\bibliographystyle{alpha}
\DeclareRobustCommand{\VAN}[3]{#3}
\bibliography{references}

\end{document}